\documentclass[12pt]{amsart}
\usepackage{latexsym,amssymb}
\usepackage{graphics}
\usepackage{harvard}
\usepackage{url}
\usepackage{color}
\usepackage{epsfig,subfigure}
\usepackage{bm}
\usepackage{multirow}
\usepackage[textsize=tiny]{todonotes}
\numberwithin{equation}{section}
\numberwithin{figure}{section}
\numberwithin{table}{section}
\def\RR{\mathbb{R}}
\def\OO{\mathcal O}

\def\PP{\mathbb{P}}

\def\EE{\mathbb E}

\def\GG{\mathcal G}
\def\NN{\mathcal N}

\def\PP{\mathbb{P}}

\def\sgn{\mbox{sgn}}
\def\argmin{\mbox{argmin}}

\def\Pois{\mbox{Pois}}
\def\pperp{\perp\!\!\!\perp}

\DeclareMathAlphabet{\mymathbb}{U}{BOONDOX-ds}{m}{n}
\def\11{\mymathbb 1}

\newtheorem{thm}{Theorem}[section]
\newtheorem{assumption}{Assumption}
\newtheorem{example}[thm]{Example}
\newtheorem{remark}[thm]{Remark}
\newtheorem{lemma}[thm]{Lemma}
\newtheorem{corollary}[thm]{Corollary}
\newtheorem{proposition}[thm]{Proposition}
\numberwithin{equation}{section}

\renewenvironment{proof}{\noindent {\bf Proof.\ }}{\hfill{\rule{2mm}{2mm}}}
\renewenvironment{remark}{\noindent {\bf Remark.\ }} {\hfill{ \rule{2mm}{2mm}}}
\renewenvironment{example}{\noindent {\bf Example.\ }} {\hfill{ \rule{0mm}{0mm}}}

\begin{document}
\bibliographystyle{econometrica}
\title{Invidious Comparisons:\\  Ranking and Selection as Compound Decisions }

\author{Jiaying Gu}
\author{Roger Koenker}
\thanks{Version:  \today . 
This paper was presented as the Walras-Bowley Lecture at the 2020 World Congress of the
Econometric Society, and is dedicated to the memory of Larry Brown who introduced us to
empirical Bayes methods.  We thank Michael Gilraine, Keisuke Hirano, Robert McMillan, Stanislav Volgushev and Sihai Dave Zhao
 for useful discussions. Jiaying Gu acknowledges financial support from 
Social Sciences and Humanities Research Council of Canada. 
}

\begin{abstract}
    There is an innate human tendency, one might call it the ``league table
    mentality,'' to construct rankings.  Schools, 
    hospitals, sports teams, movies, and myriad other objects are ranked even though their inherent 
    multi-dimensionality would suggest that -- at best -- only partial orderings were possible.  
    We consider a large class of elementary ranking problems in which we observe noisy, scalar
    measurements of merit for $n$ objects of potentially heterogeneous precision
    and are asked to select a group of the objects that are ``most meritorious.''
    The problem is naturally formulated in the compound decision framework of
    Robbins's (1956) empirical Bayes theory, but it also exhibits close connections to 
    the recent literature on multiple testing.  The nonparametric maximum likelihood
    estimator for mixture models (Kiefer and Wolfowitz (1956)) is employed to construct
    optimal ranking and selection rules.  Performance  of the rules is evaluated in
    simulations and an application to ranking U.S kidney dialysis centers.
\end{abstract}
\maketitle
\pagestyle{myheadings}
\markboth{\sc Invidious Comparisons}{\sc Gu and Koenker}

\section{Introduction}

In the wake of Wald's seminal monograph on statistical decision theory there was a growing awareness
that the Neyman-Pearson testing apparatus was inadequate for many important statistical tasks.
Ranking and selection problems featured prominently in this perception.  Motivated by a
suggestion of Harold Hotelling, \citeasnoun{Bahadur50} studied selection of the best of several
Gaussian populations.  Assuming that sample means were observed for each
of $K$ populations with means, $\theta_k$ and common variance, the problem of selecting the 
best population, $\theta^* = \max_i \{\theta_1, \ldots , \theta_K \}$, was formulated as choosing 
weights $z_1, \cdots , z_K$ to minimize, 
\[
L(\theta, z) = \theta^* - \sum_{k=1}^K z_k \theta_k/\sum_{k=1}^K z_k.
\]
Bahadur showed that among ``impartial decision rules,'' i.e. permutation equivariant rules, 
it was uniformly optimal  to select only the population with the largest
sample mean, that is to choose $z_i^* = 1$ if $\bar X_i = \max \{ \bar X_1, \cdots , \bar X_K \} $
and $z_i^* = 0$ otherwise,
thereby clearly demonstrating that procedures that did preliminary
tests of equality of means and then chose $z_i > 0$ for several or even all of the populations
when tests failed to reject were inadmissible.  This finding was reinforced in \citeasnoun{BR50}
who focused on the two-sample setting but relaxed the common variance assumption.
In related work, \citeasnoun{Bechhofer} and \citeasnoun{Gupta56} sought to optimize
the number of selected populations as well as their identities, see \citeasnoun{GuptaPan} 
and \citeasnoun{BKS} for extensive reviews of subsequent developments.

\citeasnoun{GoelRubin} pioneered the hierarchical Bayesian approach to selection
that has been adopted by numerous authors in the ensuing decades, early on by
\citeasnoun{BergerDeeley} and \citeasnoun{LairdLouis}. \citeasnoun{Portnoy} showed
that rankings based on best linear predictors were optimal in Gaussian multivariate variance components
models, but cautioned that departures from normality could easily disrupt this optimality.   
A notable feature of the hierarchical model paradigm is
the recognition that sample observations may exhibit heterogeneous precision; this is
typically accounted for by assuming known variances for observed sample means.  As ranking and
selection methods became increasingly relevant in genomic applications there has been renewed
interest in loss functions and linkages to the burgeoning literature on multiple testing.  Our
perspective is informed by recent developments in the nonparametric estimation of mixture models
and its relevance for a variety of compound decision problems.  This approach seeks to reduce
the reliance on Gaussian distributional assumptions that pervades the earlier literature. 
As we have argued elsewhere,
\citeasnoun{ProbRob}, and \citeasnoun{Minimalist} nonparametric empirical Bayes methods offer
powerful complementary methods to more conventional parametric hierarchical Bayes for multiple
testing and compound decision problems.  Our primary objective in this paper is to elaborate
this assertion for ranking and selection applications.  Throughout we try to draw parallels
and contrasts with the literature on multiple testing.  We will restrict our attention to settings
where we observe a scalar estimate of an unobserved latent quality measure accompanied by some 
measure of its precision, thereby evading more complex multivariate settings, as in \citeasnoun{accuracy}
who employ quantile regression methods.

An important motivation for revived interest in ranking and selection problems in econometrics
has been the influential work of \nocite{Chettya,Chettyb,ChettyMob} Chetty and his collaborators on 
teacher evaluation and geographic mobility in the U.S.  This has stimulated the important recent 
work of \citeasnoun{MRSW} proposing  new resampling methods for constructing confidence sets 
for ranking and selection for a finite population.  \citeasnoun{AKPM} propose an innovative
approach to the construction of confidence intervals for classical, linear shrinkage, empirical
Bayes estimators of the type used by Chetty.  Recent work by \citeasnoun{AKM} and \citeasnoun{GuoHe}
propose new confidence interval constructions for highly ranked individuals or treatments 
influenced by recent contributions to the ``inference after model selection'' literature.
In contrast to these inferential approaches we focus instead on the
complementary perspective of compound decision making, constructing decision rules for selecting 
the best, or worst, populations subject to control of the expected number of elements selected and 
among those selected, the expected proportion of false discoveries.  Rather than treating each 
selection decision in isolation, the compound decision framework tries to exploit their common 
structure to produce improved \emph{collective} performance.  Our approach is thus more closely 
aligned to that of \citeasnoun{KlineWalters} who study decision rules for assessing employer
discrimination from experiments involving fictitious job applications using closely related GMM
methods for binomial mixture models.

\citeasnoun{GGM} study teacher value added estimation employing nonparametric maximum likelihood methods
for estimating Gaussian mixture models as we advocate below.
Their analysis of data from both North Caroline and Los Angeles illustrates the advantages of more flexible
mixture models for latent value added.  In contrast to the present work, they focus on Bayes rules for
posterior means that are often used to study teachers' influence on students' future outcomes. 
These more flexible nonparametric empirical Bayes methods improve upon traditional linear shrinkage rules 
especially in the tails of the distribution where policy attention is usually focused.
This is a valuable, complementary perspective to the ranking and selection objectives of the present work.


Before proceeding it is important to acknowledge that despite its universal appeal and application
there is something inherently futile about many ranking and selection problems as intimated by our title.
If the latent measure of true quality is Gaussian, as assumed in virtually all of the econometric
applications of the selection problem, and we wish to select the top ten percent of 
individuals given that their true quality is contaminated by Gaussian noise,  accurate
selection can be very challenging when the signal to noise ratio is low.  We will see that conventional 
linear shrinkage as embodied in the classical James-Stein formula can improve performance considerably
over naive maximum likelihood (fixed effects) procedures, and some further improvement is possible
by carefully tailoring the decision rules for tail probability loss. However, we find that even 
oracle decision rules that incorporate complete knowledge of the precise distributional features of the
problem may not be able achieve better than about even odds that selected individuals have latent ability
above the selection thresholds when measurement error is comparable in magnitude to Gaussian variability in
latent ability.  When the latent distribution of ability is heavier tailed then
selection becomes somewhat easier, and more refined selection rules are more advantageous, but as
we will show the selection problem still remains quite challenging. 

Thus, a secondary objective of the paper is to add another cautionary voice to those who have 
already questioned the reliability of existing ranking and selection methods.
A critical overview of the role of ranking and selection in public policy applications is provided
by \citeasnoun{LeagueTables}.  It is widely acknowledged that league tables as currently employed
can be a pernicious influence on policy, a viewpoint underscored in \citeasnoun{GelmanPrice}. 
While much of this criticism can be attributed to inadequate data collection and inherently low
signal to noise ratios, we believe that there is also room for methodological improvements.

Section \ref{sec:Compound} provides a brief overview of compound decision theory
and describes nonparametric methods for estimation of Gaussian mixture models.
Section \ref{sec:HomoVar} introduces a basic framework for our approach to ranking and
selection in a setting with homogeneous precision of the observed measurements. 
In Section \ref{sec:KnownVariances} we introduce heterogeneous precision of 
known form, and Section \ref{sec:UnknownVar} considers settings in which the joint distribution of the 
observed measurements and their precision determines the form of the ranking and selection rules.  
Optimal ranking and selection rules are derived in each of these sections under the assumption that
the form of the mixing distribution of the unobserved, latent quality of the observations is known.
Section \ref{sec:adaptive} introduces feasible ranking and selection rules and conditions under which
they attain the same asymptotic performance as the optimal rules.
Section \ref{sec:Simulation} then compares several {\it{feasible}} ranking and selection methods, some that
ignore the compound decision structure of the problem, some that employ parametric empirical Bayes
methods and some that rely on nonparametric empirical Bayes methods.  
Finally, Section \ref{sec:Dialysis} describes an empirical application on evaluating the performance
of medical dialysis centers in the United States.  Proofs of all formal results are collected in
Appendix A.

\section{The Compound Decision Framework}
\label{sec:Compound}

\citeasnoun{Robbins.51} posed a challenge to the nascent minimax decision theory of 
\citeasnoun{Wald50}:
Suppose we observe independent Gaussian realizations, 
$Y_i \sim \NN(\theta_i , 1), \; i = 1, \cdots, n$
with means $\theta_i$ taking either the value $+1$ or $-1$.  
We are asked to estimate the $n$-vector
$\theta = (\theta_1, \cdots , \theta_n)$ subject to mean absolute error loss,
\[
L(\hat \theta , \theta) = n^{-1} \sum_{i=1}^n | \hat \theta_i - \theta_i |.
\]
When $n = 1$ Robbins shows that the minimax decision rule is $\delta (y) = \sgn (y)$;
in the least favorable variant of the problem malevolent nature chooses $\pm 1$ with
equal probability, and the optimal response is to estimate $\theta_i = +1$ when
$Y_i$ is positive, and $\theta_i = -1$ otherwise.  Robbins goes on to show that when
$n > 1$, this rule remains minimax, each coordinate is treated independently as if
viewed in complete isolation.  This is also the maximum likelihood estimator, and may
be viewed in econometrics terms as a classical fixed-effects estimator.  
But is it at all reasonable?

Doesn't our sample convey information about the relative frequency of $\pm 1$ that
might potentially contradict the pessimistic  presumption of the minimax rule?  If
we happened to know the unconditional probability, $p = \PP (\theta_i = 1)$, then the 
conditional probability that $\theta = 1$ given $Y_i = y$, is given by,
\[
\PP(\theta = 1 | y) = \frac{p \varphi(y-1)}{p \varphi(y-1) + (1-p) \varphi(y + 1)},
\]
where $\varphi$ denotes the standard Gaussian density.
We should guess $\hat \theta_i = 1$ if this probability exceeds 1/2, giving us
the revised decision rule,
\[
\delta_p (y) = \sgn ( y - \tfrac{1}{2} \log ( (1 - p)/p)).
\]
Each observed, $y_i$, is modified by a simple logistic perturbation  before 
computing the sign.  Our observed random sample, $\bm{y} = (y_1, \cdots , y_n)$, is informative about
$p$.  We have the log likelihood,
\[
\ell_n (p | y) = \sum_{i=1}^n \log ( p \varphi(y_i - 1) + (1-p) \varphi(y_i + 1)),
\]
which could be augmented by a prior of some form, if desired, to obtain a posterior
mean for $p$ and a plug-in Bayes rule for estimating each of the $\theta_i$'s.  The
Bayes risk of this procedure is substantially less than the minimax risk when $p \neq 1/2$
and is asymptotically equivalent to the minimax risk when $p = 1/2$.   This is the first
principle of compound decision theory:  borrowing strength across an entire ensemble of
related decision problems yields improved collective performance. 

What happens when we relax the restriction on the support of the $\theta$'s and allow support
on the whole real line?  We now have a general Gaussian mixture setting where the observed
$Y_i$'s have marginal density given by the convolution, $f = \varphi * G$, that is,
\[
f(y) = \int \varphi(y - \theta) dG(\theta),
\]
and instead of merely needing to estimate one probability we need an estimate of an
entire distribution function, $G$.  \citeasnoun{KW}, anticipated by an abstract of \citeasnoun{R50},
established that the nonparametric maximum likelihood estimator (NPMLE),
\[
\hat G = \argmin_{G \in \GG} \{ - \sum_{i=1}^n \log f(y_i) \; | \; 
    f(y_i) = \int \varphi(y_i - \theta)dG(\theta) \}
\]
where $\GG$ is the space of probability measures on $\RR$, is a consistent estimator of $G$.
This is an infinite dimensional convex optimization problem with a strictly convex objective
subject to linear constraints.  See \citeasnoun{Lindsay.95} and \citeasnoun{KM} for
further details on the geometry and computational aspects of the NPMLE problem.  
\citeasnoun{HS} pioneered this approach in econometrics to argue that more flexible models
of heterogeneity were needed to get reliable estimates of duration dependence in survival models.

A powerful consequence of the seemingly innocuous condition that $G$ must be non-decreasing
is that $\hat G$ must be atomic, a discrete distribution with fewer than $n$ atoms.  
A secondary consequence is that the NPMLE is ``self-regularizing,''  that is, the number, locations 
and mass of the atoms are all determined jointly by the optimization without any recourse to
auxiliary tuning parameters.  This is all a consequence of the classical Carath\'eodory theorem, 
but until quite recently little was known about the precise growth rate of the number of atoms 
characterizing the solutions, although empirical experience suggested it was quite slow.
\citeasnoun{PW20} have recently established that for $G$ with sub-Gaussian tails the cardinality of 
its support, i.e. the number of atoms, of $\hat G$,  grows like $\OO(\log n)$.  
Thus, without any further penalization, maximum likelihood automatically selects a highly parsimonious 
$\hat G$.  This is in sharp contrast to the notorious difficulties with maximum likelihood for finite 
dimensional mixture models, or with Gaussian deconvolution employing Fourier methods. 

Having seen that the upper bound on the complexity of the NPMLE $\hat G$ was only $\OO (\log n)$, 
one might wonder whether $\OO (\log n)$ mixtures are ``complex enough'' to adequately represent 
the process that generated our observed data.  \citeasnoun{PW20} also address this concern: they note that 
for any sub-Gaussian $G$, there exists a discrete distribution, $G_k$, with $k = \OO ( \log n)$ atoms,
such that for $f_k = \varphi * G_k$, the total variation distance, $TV(f, f_k) = o (1/n)$, and consequently
there is no statistical justification for considering estimators of $G$ whose complexity grows
more rapidly than $\OO (\log n)$.  This observation is related to recent literature on
generative adversarial networks, e.g. \citeasnoun{AIMM20}, that  target models and estimators 
that, when simulated, successfully mimic observed data.  

Other nonparametric maximum likelihood estimators for $G$ are potentially also of interest.
\citeasnoun{E16} has proposed an elegant log-spline sieve approach
that yields smooth estimates of $G$; this has advantages especially from an inferential
perspective, at the cost of reintroducing the task of selecting tuning parameters.
An early proposal of \citeasnoun{LL91} merged parametric empirical Bayes estimation
of $G$ with an EM step that pulled the parametric estimate back toward the NPMLE.

Given an estimate, $\hat G$, it is straightforward to compute posterior distributions
for each sample observation, or for that matter, for out-of-sample observations.
In effect, we have estimated the prior, as in \citeasnoun{Robbins.51} binary means problem,
but we have ignored the variability of $\hat G$ when we adopt plug-in procedures that
use it.  This may account for the improved performance of \textit{smoothed} estimates of $G$
in certain inferential problems,  as conjectured in \citeasnoun{K20} and studied in more
detail in \citeasnoun{JZ21}.  In the sequel we will compare
ranking and selection procedures based on various functionals of these posterior
distributions.  A leading example is the posterior mean, but ranking and selection problems
suggest other functionals of potential interest.

If we are asked to estimate the $\theta_i$'s subject to quadratic loss, and assuming standard
Gaussian noise, the Bayes rule is given by the posterior mean,
\begin{equation} \label{tweedie} 
\delta (y) = \EE (\theta | y) = y + f'(y)/f(y).
\end{equation}
\citeasnoun{efron.11} refers to this as Tweedie's formula, it appears in \citeasnoun{Robbins.56}
credited to M.C.K. Tweedie.  Appendix A of \citeasnoun{Haydn} provides an elementary derivation.
The nonlinear shrinkage term takes a particularly simple affine form when $G$ happens to 
be Gaussian, since in this case $f$ is itself also Gaussian and the formula reduces to 
well-known linear shrinkage variants of classical Stein rules.


We have focused in this brief overview on compound decision problems for Gaussian location
mixtures and posterior means, however the NPMLE is adaptable to a wide variety of other mixture problems
and other loss functions that imply other posterior functionals as we will see in the next section.
\citeasnoun{E19} and the discussion thereof offers a broader perspective on related methods.
Implementation of several NPMLE options are described in \citeasnoun{REBayesVig} and are available 
in the R package REBayes of \citeasnoun{REBayes}.

\section{Homogeneous Variances}
\label{sec:HomoVar}
Suppose that you are given real-valued
measurements, $y_i : i = 1, 2, \cdots , n$ of some attribute like
test score performance for students or their teachers, survival rates
for hospital surgical procedures, etc., and are told that 
the measurements are exchangeable and approximately Gaussian with 
unknown means $\theta_i$ and known variances $\sigma_i^2$ assumed 
provisionally to take the same value $\sigma^2$.
Your task, should you decide to accept it, is to choose a group of size
not to exceed $\alpha n$ of the elements with the largest $\theta_i$'s.
One's first inclination might be to view each $y_i$ as the maximum
likelihood estimate for the corresponding $\theta_i$, and select the
$\alpha n$ largest observed values, but the compound decision framework
suggests that it would be better to treat the problems as an ensemble.
A second natural inclination might be to compute posterior means of the
$\theta$'s with some linear or nonlinear shrinkage rule, rank them and
select the $\alpha$ best, but we will see that this too may be questionable.

\subsection{Posterior Tail Probability} \label{sec: lfdr}
A natural alternative to ranking by the posterior means is to rank by posterior tail probabilities.
Let $\theta_\alpha = G^{-1} (1 - \alpha)$,  and define,
$v_{\alpha}(y):= \mathbb{P}(\theta \geq \theta_\alpha | Y = y)$, 
then ranking by posterior tail probability gives the decision rule, 
\[
\delta(y) = \11\{v_\alpha(y) \geq \lambda_\alpha\}
\]
where the threshold $\lambda_\alpha$ is chosen so that $\mathbb{P}(v_\alpha(Y) \geq \lambda_\alpha) = \alpha$. 
This ranking criterion has been proposed by \citeasnoun{HN} motivated as a ranking device for 
a fixed quantile level $\alpha$. It can be interpreted in multiple testing terms:
$1-v_{\alpha}(y)$  is the local false discovery rate of \citeasnoun{ETST} and 
\citeasnoun{storey02}, for testing the hypothesis 
$H_0: \theta < \theta_\alpha$ vs. $H_A: \theta \geq  \theta_\alpha$.  To see this, let
$h_i$ be a binary random variable $h_i= \11\{\theta_i \geq \theta_\alpha\}$, 
the loss function for observation $i$ is 
\[
L(\delta_i, \theta_i) = \lambda \11\{h_i = 0, \delta_i = 1) + \11\{h_i = 1, \delta_i = 0\},
\]
for a generic Lagrange multiplier, $\lambda$.  The compound Bayes risk is, 
\[
\mathbb{E}[\sum_{i=1}^n L(\delta_i, \theta_i)] = n [ \alpha + \int \delta(y) [(1-\alpha) \lambda f_0(y) - \alpha f_1(y)] dy]
\]
where $f_0(y) = (1-\alpha)^{-1} \int_{-\infty}^{\theta_\alpha} \varphi(y|\theta,\sigma^2) dG(\theta)$ 
and $f_1(y) = \alpha^{-1} \int _{\theta_\alpha}^{+\infty} \varphi(y|\theta, \sigma^2) dG(\theta)$, 
$\varphi(y| \theta, \sigma^2) = \varphi((y - \theta)/\sigma)/\sigma$. 
The Bayes rule for a fixed $\lambda$ is 
\[
\delta(y_i) = \11\Big \{ v_\alpha(y_i) \geq \frac{\lambda}{1+\lambda}\Big \}
\]
where $v_\alpha(y) =\alpha f_1(y)/f(y) = \mathbb{P}(\theta \geq \theta_\alpha | Y = y)$,
and $f(y) = (1 - \alpha) f_0(y) + \alpha f_1(y).$
Provided that $v_\alpha(y)$ is monotone in $y$ a unique $\lambda_\alpha$ can be found 
such that $\mathbb{P}(\delta(Y) = 1) = \mathbb{P}(v_\alpha(Y)\geq \lambda_\alpha/(1+\lambda_\alpha)) = \alpha$. 

\begin{lemma} \label{lem:Nestedness}
For fixed $\alpha$, assuming $\mathbb{E}_{\theta|Y}[\nabla_y \log \varphi(y | \theta, \sigma^2)|Y]< \infty$, $v_\alpha(y)$ is monotone in $y$ and 
the sets $\Omega_{\alpha} := \{Y: v_\alpha(Y) \geq  \lambda_\alpha/(1+\lambda_\alpha)\}$ 
have a nested structure, that is if $\alpha_1 > \alpha_2$, then $\Omega_{\alpha_2} \subseteq \Omega_{\alpha_1}$. 
\end{lemma} 

Any implementation of such a Bayes rule requires an estimate of the mixing distribution, $G$, or 
something essentially equivalent
that would enable us to compute the local false discovery rates 
$v_\alpha (y)$ and the cut-off $\theta_\alpha$.  The NPMLE, or perhaps a smoothed version of it,
will provide a natural $\hat G$ for this task. 

\subsection{Posterior Tail Expectation and Other Losses}
Rather than assessing loss by simply counting misclassifications  
we might consider weighting such misclassifications by the magnitude of $\theta$, for example,
\[
L(\delta_i , \theta_i) = \sum_{i=1}^n (1 - \delta_i ) \11 (\theta_i \geq \theta_\alpha) \theta_i.
\]
This presumes, of course, that we have centered the distribution $G$ in some reasonable way,
perhaps by forcing the mean or median to be zero.
Minimizing with respect to $\delta$ subject to the constraint that $\PP(\delta(Y) = 1) = \alpha$ leads
to the Lagrangian,
\[
\underset{\delta}{\min} \int \int (1-\delta(y)) \11\{\theta \geq \theta_\alpha\} 
\theta \varphi(y|\theta, \sigma^2) dG(\theta) dy + \lambda \Big[\int \int  \delta(y) \varphi(y|\theta, \sigma^2) dG(\theta)dy - \alpha\Big]
\]
which is equivalent to
\begin{align*}
\underset{\delta}{\min} & \int \int \11\{\theta \geq \theta_\alpha\} 
(\theta - \lambda) \varphi(y|\theta, \sigma^2) dG(\theta) dy\\
&  - \int \delta(y) \Big[ \int \11\{\theta \geq \theta_\alpha\} 
(\theta - \lambda) \varphi(y|\theta,\sigma^2) dG(\theta) - \int \lambda \11\{\theta < \theta_\alpha\} \varphi(y|\theta,\sigma^2) dG(\theta)\Big] dy. 
\end{align*}
Ignoring the first term since it doesn't depend upon $\delta$, the Oracle Bayes rule, becomes,
choose $\delta (y) = 1$ if,
\[
\frac{\int \11\{\theta \geq \theta_\alpha \} \theta \varphi(y|\theta,\sigma^2) 
dG(\theta)}{\int \varphi(y|\theta,\sigma^2) dG(\theta)} \geq \lambda,
\]
with $\lambda$ chosen so that $\PP (\delta(Y) = 1) = \alpha$.  Such criteria are closely
related to expected shortfall criteria appearing in the literature on risk assessment.
Again, the NPMLE can be employed to construct feasible posterior ranking criteria.

Several other loss functions are considered by \citeasnoun{LLPR06} including  some based 
on global alignment of the ranks.  While intuitively appealing, such loss functions are
considerably less tractable than those we consider in the remainder of the paper.

\subsection{False Discovery and the $\alpha$-Level}
Although our loss functions yield distinct criteria for ranking, 
their decision rules lead to the same selections when the precision of the
measurements is homogeneous.  When variances are homogeneous there is a global cut-off, 
$\eta_\alpha$, and a decision rule, $\delta_\alpha(Y) = 1(Y \geq \eta_\alpha)$, 
determining a common selection for all decision rules.
\begin{lemma} \label{lem: mon1}
    For fixed $\alpha$ and homogeneous variance, posterior mean, posterior tail probability 
    and posterior tail expectation all yield the same ranking and therefore the same selection.
\end{lemma} 

The marginal false discovery rate for selection in our Gaussian mixture setting is,
\[
mFDR  = \PP (\theta < \theta_\alpha | \delta_\alpha (Y) = 1)
 = \alpha^{-1} \int_{- \infty}^{\theta_\alpha} \Phi ((\theta - \eta_\alpha)/\sigma) dG(\theta),
\]
The marginal false non-discovery rate is,
\[
mFNR  = \PP (\theta \geq \theta_\alpha | \delta_\alpha (Y) = 0)
 = (1-\alpha)^{-1} \int_{\theta_\alpha}^\infty 
    \Phi ((\eta_\alpha - \theta)/\sigma) dG(\theta),
\]

Figure \ref{fig.FDRFNR} shows the false discovery rate, and false non-discovery rate for 
a range of capacity constraints, $\alpha$, when the mixing distribution, $G$, is standard Gaussian and 
$\sigma^2 = 1$.
In this low signal to noise ratio case, the cut-off value $\eta_\alpha$ is the $(1-\alpha)$ quantile of 
$\mathcal{N}(0, 2)$, and it is very difficult to distinguish the meritorious from the merely lucky. 
For selecting individuals at the top $\alpha$ quantile, the false 
discovery rate is alarmingly high especially for smaller $\alpha$, implying 
that the selected set may consist of a very high proportion of false discoveries.  
When $\alpha = 0.10$ the proportion of selected observations 
with $\theta$ below the threshold $\theta_\alpha$ is slightly greater than 50 percent.
	
\begin{figure}[bp]
	\begin{center}
		\resizebox{.65\textwidth}{!}{\includegraphics{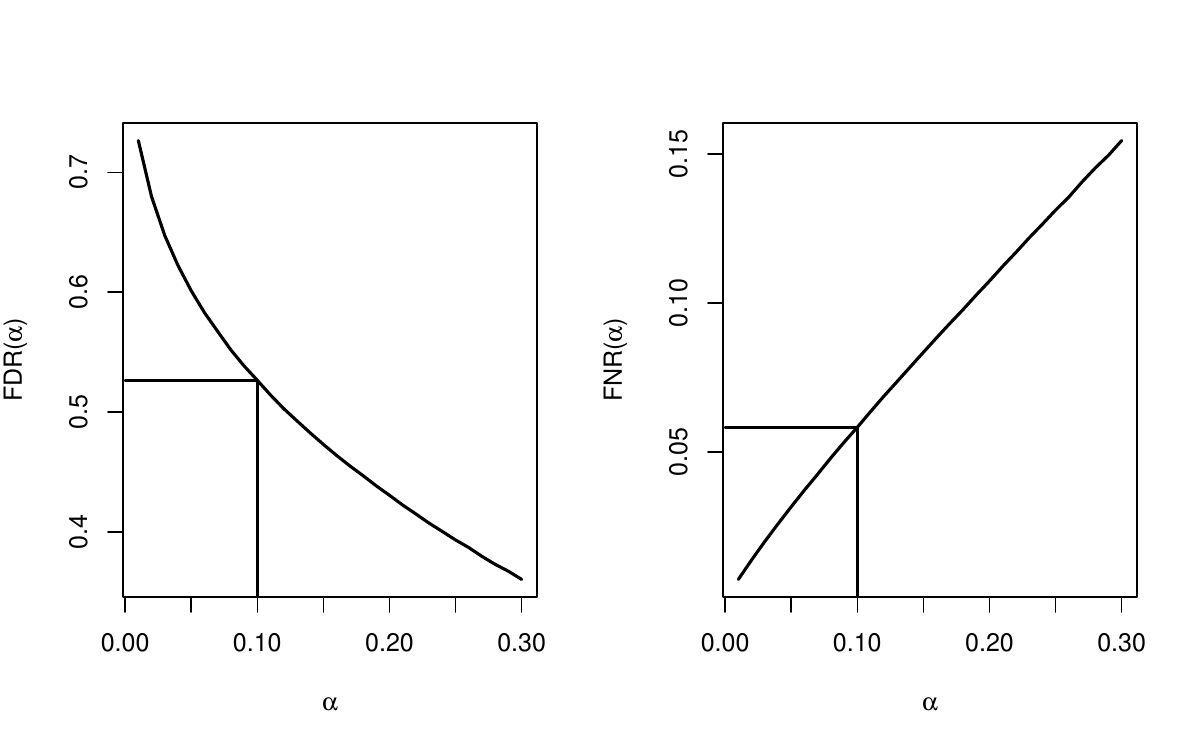}}
	\end{center}
	\caption{\small False Discovery Rates and False Non-Discovery Rates for a 
	Standard Gaussian Mixing Distribution.  }
	\label{fig.FDRFNR}
\end{figure}

{\tiny{
\begin{table}[!tbp]
\begin{center}
\begin{tabular}{lrcrrrrrcrrrrr}
\hline\hline
\multicolumn{1}{l}{\bfseries }&\multicolumn{1}{c}{\bfseries FDR}&\multicolumn{1}{c}{\bfseries }&\multicolumn{5}{c}{\bfseries Wrongly Selected}&\multicolumn{1}{c}{\bfseries }&\multicolumn{5}{c}{\bfseries Correctly Selected}\tabularnewline
\cline{2-2} \cline{4-8} \cline{10-14}
\multicolumn{1}{l}{}&\multicolumn{1}{c}{}&\multicolumn{1}{c}{}&\multicolumn{1}{c}{20\%}&\multicolumn{1}{c}{40\%}&\multicolumn{1}{c}{50\%}&\multicolumn{1}{c}{60\%}&\multicolumn{1}{c}{80\%}&\multicolumn{1}{c}{}&\multicolumn{1}{c}{20\%}&\multicolumn{1}{c}{40\%}&\multicolumn{1}{c}{50\%}&\multicolumn{1}{c}{60\%}&\multicolumn{1}{c}{80\%}\tabularnewline
\hline
$\sigma^2=1$&$0.526$&&$0.205$&$0.414$&$0.531$&$0.658$&$0.990$&&$0.188$&$0.391$&$0.504$&$0.632$&$0.966$\tabularnewline
$\sigma^2=2$&$0.421$&&$0.146$&$0.302$&$0.391$&$0.486$&$0.743$&&$0.197$&$0.394$&$0.503$&$0.625$&$0.951$\tabularnewline
$\sigma^2=3$&$0.361$&&$0.120$&$0.248$&$0.320$&$0.402$&$0.620$&&$0.197$&$0.393$&$0.500$&$0.618$&$0.931$\tabularnewline
$\sigma^2=4$&$0.319$&&$0.103$&$0.212$&$0.275$&$0.347$&$0.540$&&$0.198$&$0.392$&$0.495$&$0.609$&$0.915$\tabularnewline
$\sigma^2=5$&$0.296$&&$0.093$&$0.192$&$0.247$&$0.313$&$0.487$&&$0.196$&$0.383$&$0.484$&$0.596$&$0.897$\tabularnewline
\hline
\end{tabular}
\caption{FDR improves as the signal becomes more dispersed.  In settings with standard Gaussian measurement
error and Gaussian distribution, G, for the $\theta$'s, the variance of $G$ can be interpreted as a 
signal-to-noise ratio.  As the variance of $G$ increases selection becomes easier and FDR is reduced.
The wrongly selected units become more concentrated near the selection threshold.  Columns 2-6 of
the table report quantiles of the wrongly selected units measured in standard deviations from the threshold.
Columns 7-11 report corresponding quantiles for the correctly selected units.\label{tab3}}\end{center}
\end{table}

}}

When the variance of the $\theta$'s, the signal-to-noise ratio, increases from 1 to 5, 
the selection problem becomes somewhat easier. This is reflected not only by the false discovery 
rate decreasing from above one half to about a third, but is also reflected in results in 
Table \ref{tab3} that show that the wrongly selected individuals have their true values of 
$\theta$ clustered closer to the thresholding value $\theta_\alpha$ measured in terms of standard deviation. 

When $\sigma^2 = 1$, we have about 50\% falsely selected, and the $\theta$'s among the 
correctly-selected and the wrongly selected individuals are roughly symmetrically distributed around 
the thresholding value. In such a case, even oracle decision rules that incorporate complete knowledge 
of the precise distributional features of the problem may not be able to achieve better than about 
even odds that selected individuals have latent ability above the selection thresholds when 
measurement error is comparable in magnitude to Gaussian variability in latent ability. 
As variances of alpha increases to 5 then selection becomes somewhat easier, only 1/3 among the 
selected are falsely selected, and most (80\%) of these $\theta$ values are within 0.5 standard 
deviations away from the selection threshold. 

It is perhaps worth stressing that at the margin, near the decision boundary, it will always 
be difficult to distinguish true from false discoveries, but FDR measures the proportion of all 
selections that are incorrect, not just those near the threshold.  Other loss functions that penalize 
in a more continuous way may be considered to reflect information in Table \ref{tab3}. For example, losses that
weight the classification error by the magnitude of the discrepancy between the latent effect
and the threshold could be considered.  Such losses, however, make it more difficult to incorporate 
conventional forms of error control.

Thus far we have implicitly assumed that the size of the selected set is predetermined by the
parameter $\alpha$.  Having established a ranking based on a particular loss function, we simply 
select a subset of size $\lceil {\alpha n} \rceil$ consisting of the highest ranked observations.  
In the next subsection we begin to consider modifying this strategy by constraining the probability 
of false discoveries.  This will allow the size of the selected set to adapt to the difficulty 
of the  selection task.

\subsection{Guarding against false discovery} 

Recognizing the risk of false ``discoveries'' among those selected, 
we will consider an expanded loss function, 
\begin{equation} \label{loss}
L(\bm{\delta}, \bm{\theta}) = \sum_{i=1}^n h_i  (1-\delta_i) + 
\tau_1 \Big(\sum_{i=1}^n \Big \{ (1-h_i)\delta_i - \gamma \delta_i\Big\}\Big)  + 
\tau_2 \Big(\sum_{i=1}^n \delta_i  - \alpha n \Big) 
\end{equation}
where $h_i = \11\{\theta_i \geq  \theta_\alpha\}$.
If we set $\tau_1$ to zero, then minimizing the expected loss leads to the Bayes rule discussed in Section \ref{sec: lfdr}. 
On the other hand, if we set $\tau_2$ to zero, then minimizing expected loss leads to a decision rule that 
is equivalent to a multiple testing problem with null hypothesis $H_{0i}: \theta_i \leq  \theta_\alpha$; 
the goal is to minimize the expected number of over-looked discoveries subject to the constraint that the 
marginal FDR rate is controlled at level $\gamma$, that is,
$\mathbb{E}[\sum_{i=1}^n (1-h_i) \delta_i]/ \mathbb{E}[\sum_{i=1}^n \delta_i] \leq \gamma.$ 
When $\tau_1 = 0$, the risk can be expressed as,
\[
\mathbb{E}_{\bm{\theta}|Y}\Big [L(\bm{\delta}, \bm{\theta})\Big ] = \sum_{i=1}^n (1-\delta_i) v_\alpha(Y_i) + 
\tau_2 \Big( \sum_{i=1}^n \delta_i - \alpha n\Big)
\]
where $v_\alpha(y_i) = \mathbb{P}(\theta_i \geq \theta_\alpha|Y_i = y_i)$. 
Taking another expectation over $Y$, and minimizing over both $\bm{\delta}$ and $\tau_2$, leads to the decision rule, 
\[
\delta_i^* = \begin{cases}
1  , & \text{if  } v_\alpha(y_i) \geq \tau_2^* \\
0  , &\text{if } v_\alpha(y_i) < \tau_2^*.
\end{cases}
\]
The Lagrange multiplier is chosen so that the constraint 
$\mathbb{P}(\delta_i = 1) \leq  \alpha$ holds with equality:
\[
\tau_2^* = \text{min}\{\tau_2: \mathbb{P}(v_\alpha(y_i) \geq \tau_2) \leq \alpha\}
\]
Each selection improves the objective function by $v_\alpha(y_i)$, 
but incurs  a cost of $\tau_2$. Since all selections incur the same cost, 
we may rank according to $v_\alpha(y_i)$, 
selecting units until the capacity constraint $\alpha n$ is achieved. 
Selection of the last unit may need to be randomized to exactly satisfy the constraint,
as we note below.

When $\tau_2 = 0$ the focus shifts to the marginal FDR, the ratio of the expected number of false discoveries, 
to the expected number of selections. This is slightly different from the original FDR as defined in \citeasnoun{BH}. 
However, when $n$ is large the two concepts are asymptotically equivalent as shown by \citeasnoun{GenoWass}. 
Our objective becomes, 
\[
\mathbb{E}_{\bm{\theta}|\bm{Y}}\Big[ L(\bm{\delta}, \bm{\theta})\Big] = 
\sum_{i=1}^n (1-\delta_i) v_\alpha(Y_i) + \tau_1 \Big(\sum_{i=1}^n \{ \delta_i (1- v_\alpha(Y_i)) - 
\gamma \delta_i \}\Big). 
\]
Taking expectations again over $Y$ and minimizing over both $\bm{\delta}$ and $\tau_1$ yields,
\[
\delta_i^* = \begin{cases}
1  , & \text{if  } v_\alpha(y_i) > \tau_1^* ( 1- v_\alpha(y_i) - \gamma)  \\
0  , &\text{if }v_\alpha(y_i) \leq \tau_1^* ( 1- v_\alpha(y_i) - \gamma)
\end{cases}
\]
and the Lagrange multiplier takes a value $\tau_1^*$ to make the marginal FDR constraint hold with equality. 

When both constraints are incorporated we must balance the power gain from more selections 
and the cost that occurs from both the capacity constraint and FDR control. The Bayes rule solves,
\[
\min_{\bm{\delta}} \mathbb{E}\Big [\sum_{i=1}^{n} (1-\delta_i) v_{\alpha}(y_i) \Big] + 
\tau_1 \Big ( \mathbb{E}\Big[ \sum_{i=1}^n \Big \{(1- v_\alpha(y_i)) \delta_i - 
\gamma \delta_i \Big\}\Big] \Big) + \tau_2 \Big( \mathbb{E}\Big[ \sum_{i=1}^n \delta_i \Big] - \alpha n \Big).
\]
Given the discrete nature of the decision function, this problem appears to take the form of a
classical knapsack problem, however following the approach of \citeasnoun{Basu} we will consider a relaxed version 
of the problem in which units are selected sequentially until one or the other constraint would be violated, with the
final selection randomized to satisfy the constraint exactly.

\begin{remark}\label{remark: power}
Given the Lagrangian form of our loss function it is natural to consider an optimization perspective
for the selection problem.  Minimizing the expectation of the loss defined in \eqref{loss} is 
equivalent to minimizing $\mathbb{P}[\delta_i = 0, \theta_i \geq \theta_\alpha]$ subject to
the constraint that $\mathbb{P}[\delta_i = 1, \theta_i < \theta_\alpha]/\mathbb{P}[\delta_i = 1] \leq  \gamma$ 
and $\mathbb{P}[\delta_i = 1] \leq \alpha$.  So we are looking for a thresholding rule that minimizes 
the expected number of missed discoveries subject to the capacity constraint and the constraint that 
the marginal FDR rate of the decision rule is below level $\gamma$. This minimization problem 
is also easily seen, from a testing perspective, to be equivalent to maximizing power of the decision rule 
$\delta$, $\mathbb{P}[\delta_i = 1| \theta_i \geq \theta_\alpha]$, subject to the same two constraints. 
\end{remark}

\begin{proposition} \label{prop: homo} 
For any pair, $(\alpha, \gamma)$ such that $\gamma < 1- \alpha$, the optimal Bayes rule takes the form, 
$\delta_i^* = \11\{v_\alpha(y_i) \geq \lambda^*(\alpha, \gamma) \}$
where $\lambda^*(\alpha, \gamma) = v_\alpha(t^*)$ with $t^* = \max\{t_1^*, t_2^*\}$, 
\begin{align*}
t_1^* & = \min\Big \{t:  \frac{\int_{-\infty}^{\theta_\alpha} \tilde{\Phi}((t - \theta)/\sigma) dG(\theta)}{\int_{-\infty}^{+\infty} \tilde{\Phi}((t - \theta)/\sigma)dG(\theta)} - \gamma \leq0 \Big \}, \\
t_2^*  & = \min \Big \{t:  \int_{-\infty}^{+\infty} \tilde{\Phi}((t - \theta)/\sigma)dG(\theta)-\alpha \leq 0 \Big \}
\end{align*}
and $\tilde \Phi$ denoting the survival function of a standard normal random variable.
\end{proposition}


\begin{remark}	
The optimal cutoff $t^*$ depends on the data generating process and also the choice of $\alpha$ and $\gamma$. 
When data is noisy, the FDR control constraint may be binding before the capacity constraint is reached, 
and consequently the selected set may be strictly smaller than the pre-specified $\alpha$ proportion. 
On the other hand, when the signal is strong, the FDR control constraint is unlikely to bind before 
the capacity constraint is reached.
\end{remark}

We have seen that when variances are homogeneous, the optimal selection rule thresholds on $Y$, 
so it is clear then that any ranking that is based on a monotone transformation of $Y$ will lead to 
an equivalent selected set.  We should also stress that we have focused on a null hypothesis that 
depends on $\alpha$, while the multiple testing literature, for example \citeasnoun{ETST}, 
\citeasnoun{SunCai} and \citeasnoun{Basu}, typically focuses on the null hypothesis of 
$H_{0i}: \theta_i = 0$.
When variances are homogeneous, it doesn't matter whether we use an $\alpha$ dependent null 
or the conventional zero null, because the transformation based on the conventional null, 
$\mathbb{P}(\theta > 0|Y=y)$, is also a monotone function of $Y$, and therefore yields an 
equivalent decision rule.  However, when variances are heterogeneous, this invariance no longer holds; 
different transformations of the pair $(y, \sigma)$ lead to distinct decision rules that lead to 
distinct performance, and using the conventional null hypothesis is no longer advisable for the 
ranking and selection problem as we will show in the next section.

\section{heterogeneous known variances}
\label{sec:KnownVariances}
The homogeneous variance assumption of the preceding section is unsustainable in
most applications.  Batting averages are accompanied by a number of ``at bats'' and
mean test score performances are accompanied by student sample sizes.
In this section we will consider the expanded model, 
\[
Y_i \sim \mathcal{N}(\theta_i, \sigma_i^2), \quad \text{and   } 
\quad  \theta_i \sim G, \quad \sigma_i \sim H, \quad \sigma_i \pperp \theta_i 
\]
We will assume that we observe $\sigma_i$, an assumption that will be relaxed in the next section.

\subsection{Posterior Tail Probability} 
With the same alternative hypothesis as $H_{A}: \theta \geq \theta_\alpha$, 
it is natural to consider the posterior tail probability again, now as a function of the pair, $(y_i, \sigma_i)$, 
\[
v_\alpha(y_i,\sigma_i) = \mathbb{P}(\theta_i \geq \theta_\alpha | y_i, \sigma_i) = \frac{\int_{\theta_\alpha}^{+\infty} \varphi(y_i | \theta,\sigma_i^2) dG(\theta)}{\int_{-\infty}^{+\infty} \varphi(y_i | \theta,\sigma_i^2)dG(\theta)}.
\]

Solving the same decision problem with the loss function specified in (\ref{loss}), 
we have the conditional risk, 
\[ 
\mathbb{E}_{\bm{\theta}|\bm{Y}, \bm{\sigma}}\Big[ L(\bm{\delta}, \bm{\theta})\Big] = 
\sum_{i=1}^n (1- \delta_i) v_\alpha(Y_i, \sigma_i) + 
\tau_1 \Big( \sum_{i=1}^n \{\delta_i(1- v_\alpha(Y_i, \sigma_i)) - \gamma \delta_i \}\Big) + 
\tau_2 \Big( \sum_{i=1}^n \delta_i - \alpha n\Big). 
\]
Taking another expectation with respect to the joint distribution of the $(Y_i, \sigma_i)$'s, 
the Bayes rule solves 
\[
\min_{\bm{\delta}} \mathbb{E}\Big [\sum_{i=1}^{n} (1-\delta_i) v_{\alpha}(y_i, \sigma_i ) \Big] + \tau_1 \Big ( \mathbb{E}\Big[ \sum_{i=1}^n \Big \{(1- v_\alpha(y_i, \sigma_i )) \delta_i - \gamma \delta_i \Big\}\Big] \Big) + \tau_2 \Big( \mathbb{E}\Big[ \sum_{i=1}^n \delta_i \Big] - \alpha n \Big)
\]
The optimal selection rule can again be characterized as a thresholding rule on $v_\alpha(y_i, \sigma_i)$. 

\begin{proposition} \label{prop: rule_ysigma}
For a pre-specified pair $(\alpha, \gamma)$ such that $\gamma < 1-\alpha$, the Bayes rule takes the form, 
$\delta^* (y, \sigma) = \11\{v_\alpha(y, \sigma) \geq  \lambda^*(\alpha,\gamma)\}$
where $\lambda^*(\alpha,\gamma) = \max\{\lambda_1^*(\alpha,\gamma), \lambda_2^*(\alpha)\}$, 
		\begin{align*}
\lambda_1^*(\alpha,\gamma) &= \min\Big \{\lambda:  \frac{\int \int_{-\infty}^{\theta_\alpha} \tilde{\Phi}((t_\alpha(\lambda,\sigma) - \theta)/\sigma) dG(\theta)dH(\sigma)}{\int \int_{-\infty}^{+\infty} \tilde{\Phi}((t_\alpha(\lambda,\sigma) - \theta)/\sigma)dG(\theta)dH(\sigma)} - \gamma \leq0 \Big \}\\
\lambda_2^*(\alpha) & = \min \Big \{\lambda:  \int \int_{-\infty}^{+\infty} 
\tilde{\Phi}((t_\alpha(\lambda, \sigma) - \theta)/\sigma)dG(\theta)dH(\sigma)-\alpha \leq 0 \Big \}
	\end{align*}
and $t_\alpha(\lambda,\sigma)$ defined as 
$v_\alpha(t_\alpha(\lambda,\sigma), \sigma) = \lambda$ for all $\lambda \in [0,1]$. 
\end{proposition}


\begin{remark}
Note that although the thresholding value $\lambda^*$ does not depend on the value of $\sigma$, 
the ranking does depend on $\sigma$. One way to see this is that since $v_\alpha(y,\sigma)$ is monotone 
in $y$ for all $\sigma>0$, the optimal rule is equivalent to $\11\{y_i > t_\alpha(\lambda^*,\sigma)\}$, 
where $t_\alpha(\lambda,\sigma)$ is a function of $\sigma$. For a fixed value of $\lambda^*$, 
the selection region for $Y$ depends on $\sigma$ in a nonlinear way. 
Comparing individuals $i$ and $j$, it may be the case that $y_i > y_j$, but $y_j$ belongs to 
the selection region while $y_i$ does not.  An example to illustrate this appears below.
It should also be emphasized that when variances are heterogeneous, different loss functions
need not lead to equivalent selections. 
\end{remark}

\subsection{The Conventional Null Hypothesis} \label{sec.convnull}
The posterior tail probability criterion is motivated by viewing the ranking and 
selection problems as hypothesis testing while allowing the null hypothesis to be 
$\alpha$ dependent. The particular construction of the null hypothesis turns out 
to be critical for the ranking exercise. In this subsection we present a simple example to 
illustrate that tail probability based on the conventional null hypothesis of zero 
effect does not lead to a powerful ranking device. Consider data generated from 
a three component normal mixture model, 
\begin{equation} \label{eg: zeronull}
Y_i |\sigma_i \sim 0.85 \mathcal{N}(-1, \sigma_i^2) + 0.1 \mathcal{N}(0.5, \sigma_i^2) + 0.05 \mathcal{N}(5, \sigma_i^2), \quad \sigma_i \sim U[0.5, 4]
\end{equation}
Instead of transforming the data by $v_\alpha$, we consider the transformation,
\[
\texttt{T}(y_i,\sigma_i^2) = \mathbb{P}(\theta_i > 0| y_i ,\sigma_i^2) = \frac{\int_0^{+\infty} \varphi(y_i | \theta,\sigma_i^2) dG(\theta)}{\int_{-\infty}^{+\infty}  \varphi(y_i | \theta,\sigma_i^2) dG(\theta)}
\]
and rank individuals accordingly.  This transformation corresponds to the procedure 
proposed in \citeasnoun{SM}, and is motivated for multiple testing problems under the conventional 
null hypothesis $H_{0}: \theta \leq 0$. 
The decision rule $\delta_i^{\texttt{T}} = \11\{\texttt{T}(y_i,\sigma_i) \geq \lambda\}$ then 
chooses the cutoff value $\lambda$ that respects both the capacity constraint and the FDR control 
constraint for selecting the top $\alpha$ proportion. 

\begin{figure}[h!]
	\begin{center}
		\resizebox{.5\textwidth}{!}{\includegraphics{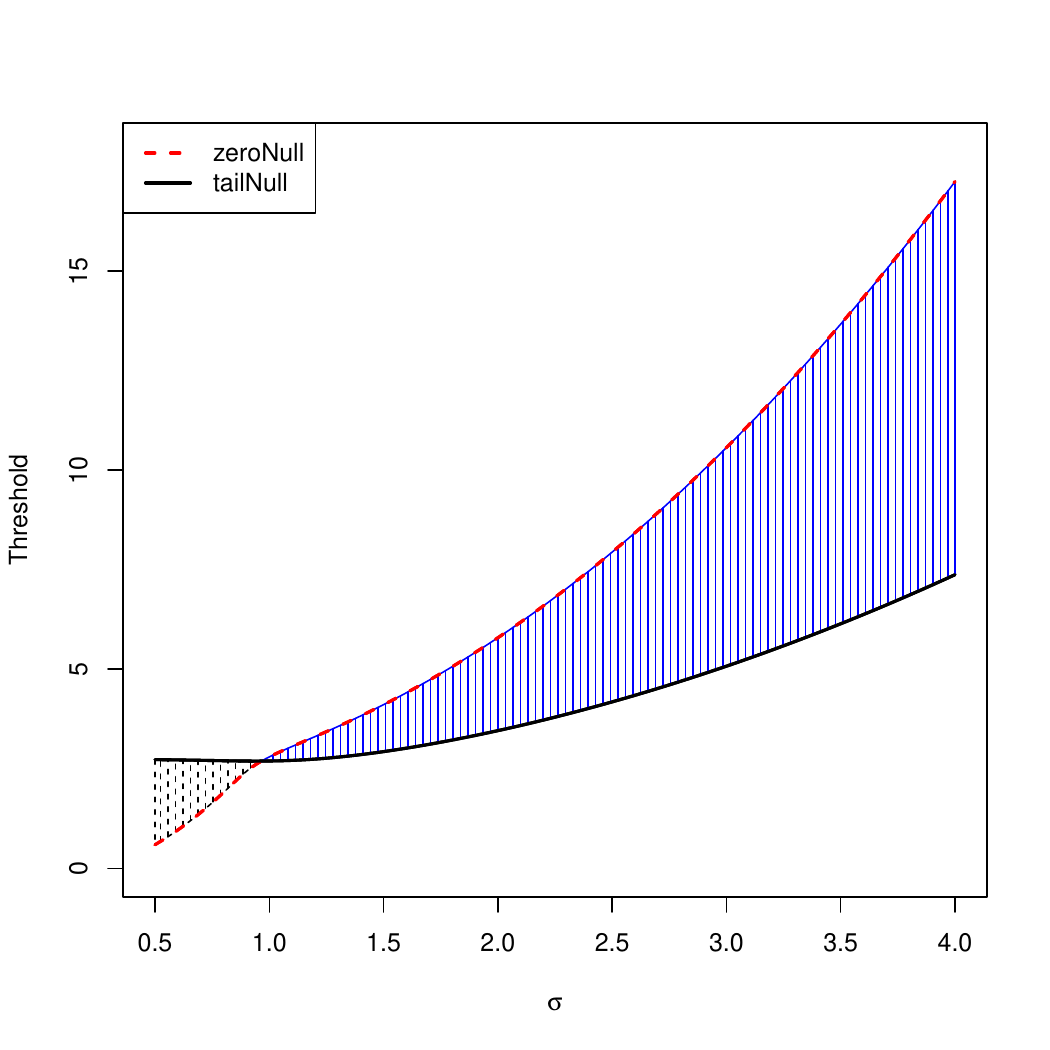}}
	\end{center}
	\caption{\small Selection boundaries based on the model (\ref{eg: zeronull}) with 
	$\alpha = 0.05$ and $\gamma = 0.1$. The solid black curve corresponds to the boundary of the 
	selection region based on transformation $v_\alpha$. The dash red curve corresponds to 
	the boundary of the selection region based on transformation $T$. Density of $\sigma$ is 
	assumed to be uniform on the interval $[0.5, 4]$.}
	\label{fig: zeronull}
\end{figure}

Figure \ref{fig: zeronull} compares the selection region for the two ranking procedures 
with $\alpha = 5\%$ and marginal FDR control at level $10\%$. The solid black line 
corresponds to the selection boundary using ranking based on transformation $v_\alpha$ 
and the dashed red line corresponds to the selection boundary using ranking based on the 
transformation $\texttt{T}$. The black highlighted area below the black selection boundary
corresponds to a region where the ranking method based on $\texttt{T}$ will select but 
the ranking method based on $v_\alpha$ does not. On the other hand, the blue highlighted 
area corresponds to a region selected by $v_\alpha$, but not for $\texttt{T}$. 
The transformation $\texttt{T}$ ranks those in the black region higher than those in the 
blue region because although they have a relatively smaller mean effect $y$, their associated 
variances are also smaller indicating stronger evidence that such individuals have a 
positive $\theta$ than those located in the blue area.  However,
our task is to find individuals with true effects, $\theta_i$, in the upper tail. 
For $\alpha = 5\%$, we aim to select all individuals with $\theta=5$, individuals in 
the black region present strong evidence that their true effect \textit{can not} be too large 
because their observed effect $y$ is small and their associated variance is also small, 
while those in the blue region, although their observed mean effects are associated with larger 
variances offer reasonable evidence that their associated true effect $\theta$ may be large. 
This evidence is not apparent in transformation $\texttt{T}$, but is captured in the 
transformation $v_\alpha$. 

Indeed, the average power of ranking based on the two different transformation 
$v_\alpha$ and $\texttt{T}$ differ significantly.  Defining the power of the selection
rule as $\beta(\delta) := 
\mathbb{P}(\theta_i \geq \theta_\alpha, \delta_i = 1)/\mathbb{P}(\theta_i \geq \theta_\alpha)$, 
the proportion of true top $\alpha$ cases selected based on decision rule $\bm{\delta}$, then 
$\beta(\bm{\delta^{\texttt{T}}}) = 39\%$ and $\beta(\bm{\delta^*}) = 69\%$.  Thus, although
much of the literature relies on ranking and selection rules based on some form of
posterior means and conventional hypothesis testing apparatus we would caution that such
methods can be quite misleading and inefficient.

\subsection{Nestedness of Selection Sets}
If we were to relax the capacity constraint to allow a larger proportion, $\alpha_1 > \alpha_0$
be be selected, while maintaining our initial false discovery control, we would expect that
members selected under the more stringent capacity constraint should remain selected under 
the relaxed constraint.  We now discuss sufficient conditions under which we obtain this
nestedness of the selection sets when using the posterior tail probability rule.
This is a natural condition in applications like our analysis of ranking and selection
of dialysis centers especially because we would like to assign ``letter grades'' to several 
several subgroups of the centers.

The optimal Bayes rule defines the selection set for each pair of $(\alpha, \gamma)$ as
\[
\Omega_{\alpha, \gamma} := \{(y, \sigma): v_\alpha(y, \sigma) \geq  \lambda^*(\alpha,\gamma)\}
\]
and when $\sigma$ is known, $v_\alpha(y,\sigma)$ is monotone in $y$ as shown in Lemma \ref{lem: mon1} for each fixed $\sigma$, hence the selection set can also be represented as 
\[
\Omega_{\alpha, \gamma} =\{(y,\sigma): y \geq  t_{\alpha}(\lambda^*(\alpha,\gamma),\sigma)\}
\]
It is also convenient for later discussion to define 
\begin{align*}
	\Omega_{\alpha,\gamma}^{FDR}: &= \{(y, \sigma): v_\alpha(y, \sigma) \geq \lambda_1^*(\alpha,\gamma)\} = \{(y, \sigma): y \geq t_\alpha(\lambda_1^*(\alpha,\gamma),\sigma)\}\\
	\Omega_{\alpha}^{C} &:= \{(y, \sigma): v_\alpha(y, \sigma) \geq \lambda_2^*(\alpha)\} = \{(y, \sigma): y \geq t_\alpha(\lambda_2^*(\alpha), \sigma)\}
\end{align*}
which are respectively the selection sets when the false discovery rate constraint or the capacity constraint is binding. It is easy to see that $\Omega_{\alpha,\gamma} = \Omega_{\alpha,\gamma}^{FDR} \cap \Omega_{\alpha}^C$.
\begin{lemma} \label{lem: nest1}
	Let the density function of $v_\alpha(y_i, \sigma_i)$ be denoted as $f_v(v; \alpha)$ and let 
	\[
	\lambda_1^*(\alpha,\gamma) = \min\Big \{\lambda:  \frac{\int \int_{-\infty}^{\theta_\alpha} \tilde{\Phi}((t_\alpha(\lambda,\sigma) - \theta)/\sigma) dG(\theta)dH(\sigma)}{\int \int_{-\infty}^{+\infty} \tilde{\Phi}((t_\alpha(\lambda,\sigma) - \theta)/\sigma)dG(\theta)dH(\sigma)} - \gamma \leq0 \Big \}
	\]
	with $t_\alpha(\lambda, \sigma)$ defined as $v_\alpha(t_\alpha(\lambda,\sigma), \sigma) = \lambda$ and $\tilde \Phi$ be the survival function of the standard normal random variable. 
	If $\nabla_{\alpha} \log f_v(v; \alpha)$ is non-decreasing in $v$, then for fixed $\gamma$, if $\alpha_1 > \alpha_2$, we have $\lambda_1^*(\alpha_1, \gamma) \leq  \lambda_1^*(\alpha_2,\gamma)$.
	
	
\end{lemma}

\begin{remark}
	The density function $f_{ v}(v; \alpha)$ can be viewed as a function of $v$ 
	indexed by the parameter $\alpha$. An explicit form for $f_{ v}(v;\alpha )$ 
	appears in Section \ref{sec:egNN} for the normal-normal model). 
	The condition imposed in Lemma \ref{lem: nest1} is equivalent to a monotone 
	likelihood ratio condition, that is that the likelihood ratio 
	$f_v(v; \alpha_1) / f_v(v; \alpha_2)$ is non-decreasing in $v$ 
	if $\alpha_1 > \alpha_2$. 
\end{remark}

\begin{corollary}\label{cor: nest1}
	If the condition in Lemma \ref{lem: nest1} holds, then $\Omega_{\alpha_2,\gamma}^{FDR} \subseteq \Omega_{\alpha_1,\gamma}^{FDR}$ for any $\alpha_1 > \alpha_2$. 
\end{corollary}

\begin{remark}
	The condition in Lemma \ref{lem: nest1} is sufficient but not necessary for nestedness of $\Omega_{\alpha,\gamma}^{FDR}$.  Even when $\lambda_1^*(\alpha_1,\gamma) > \lambda_1^*(\alpha_2,\gamma)$, we can still have $t_{\alpha_1}(\lambda_1^*(\alpha_1,\gamma),\sigma) < t_{\alpha_2}(\lambda^*_1(\alpha_2,\gamma),\sigma)$ because the function $v_{\alpha}(Y,\sigma)$ depends on $\alpha$, as does its inverse function $t_{\alpha}$.
\end{remark}

\begin{lemma} \label{lem: nest2}
	Let $\lambda_2^*(\alpha)$ be defined as in Proposition \ref{prop: rule_ysigma}.
If for any $\alpha_1 > \alpha_2$, $t_{\alpha_1}(\lambda_2^*(\alpha_1), \sigma) \leq t_{\alpha_2}(\lambda_2^*(\alpha_2),\sigma)$ for each $\sigma$, then $\Omega_{\alpha_2}^C \subseteq \Omega_{\alpha_1}^C$.
\end{lemma}

\begin{remark}
	The monotonicity here coincides with the condition in Theorem 3 of \citeasnoun{HN},
	who demonstrate that it holds when $G$ is Gaussian.
	However, it need not hold as shown in our counter-example in Section \ref{sec: S3eg}. 
\end{remark}

\begin{lemma}\label{lem: nest3}
	If $\nabla_\alpha \log f_v(v; \alpha)$ is non-decreasing in $v$ and the condition in Lemma \ref{lem: nest2} holds, then for a fixed $\gamma$, the selection region has a nested structure: if $\alpha_1 > \alpha_2$ then $\Omega_{\alpha_2, \gamma} \subseteq \Omega_{\alpha_1,\gamma}$.  
\end{lemma}

\subsection{Examples} \label{sec:egNN}
In this section we consider several examples beginning with the simplest classical case 
in which the $\theta_i$ constitute a random sample from the standard Gaussian distribution.
This Gaussian assumption on the form of the mixing distribution $G$ underlies almost
all of the empirical Bayes literature in applied economics; it is precisely what justifies
the linear shrinkage rules that are typically employed.

\begin{example}[{\bf{Gaussian $G$}}]
	Consider the normal-normal model, where $y|\theta, \sigma^2 \sim \mathcal{N}(\theta, \sigma^2)$ and $\theta \sim \mathcal{N}(0,\sigma_\theta^2)$ and $\sigma \sim H$ with density function $h(\sigma)$.  
	The marginal distribution of $y$ given $\sigma^2$ is $\mathcal{N}(0, \sigma^2 + \sigma_\theta^2)$ 
	and the joint density of $(y,\sigma)$ takes the form, 
	\[
	f(y, \sigma) = \frac{1}{\sqrt{2\pi (\sigma^2 + \sigma_\theta^2)}} \exp\Big\{ -\frac{y^2}{2(\sigma^2 + \sigma_\theta^2)}\Big\} h(\sigma).
	\]
	Given the normal conjugacy, the posterior distribution of $\theta | y, \sigma^2$ 
	follows $\mathcal{N} (\rho y, \rho \sigma^2)$ 
	where $\rho = \sigma_\theta^2/( \sigma_\theta^2 + \sigma^2)$.
	The random variable $v$ is thus a transformation of the pair $(Y, \sigma^2)$, 
	defined as, 
	\[
	 v = \psi(y, \sigma^2) := \mathbb{P}(\theta \geq  \theta_\alpha|y, \sigma^2) = 
	\Phi ( ( \rho y - \theta_\alpha ) /\sqrt{\rho \sigma^2}). 
	\]
	For fixed $\sigma^2$, $\psi$ is monotone increasing in $y$ and 
	$\psi^{-1}( v) = \theta_\alpha / \rho  + \sqrt{\sigma^2/\rho }\Phi^{-1}(v)$ with 
	$\nabla_{ v} \psi^{-1}( v) = \sqrt{\sigma^2/\rho}/\varphi(\Phi^{-1}( v))$. 
	The joint density of $ v$ and $\sigma$ is thus,
	\begin{align*}
		g( v, \sigma) &= f(\psi^{-1}( v), \sigma) 
		    |\nabla_{ v} \psi^{-1}( v)|\\
		& =  \frac{1}{\sqrt{2\pi (\sigma^2 + \sigma_\theta^2)}} 
		\exp \Big \{ -\frac{( \theta_\alpha /\rho +\sqrt{\sigma^2/\rho} 
		\Phi^{-1}( v))^2}{2(\sigma^2 + \sigma_\theta^2)}\Big\} 
		\frac{ \sqrt{\sigma^2/\rho}}{\varphi(\Phi^{-1}( v))}h(\sigma).
	\end{align*}
	Integrating out $\sigma$ we have the marginal density of $ v$, 
	\[
	f_{ v} (v; \alpha) = \int  
		\frac{1}{\sqrt{2\pi (\sigma^2 + \sigma_\theta^2)}} 
		\exp \Big \{ -\frac{( \theta_\alpha /\rho + \sqrt{\sigma^2/\rho} 
		\Phi^{-1}( v))^2}{2(\sigma^2 + \sigma_\theta^2)}\Big\} 
		\frac{ \sqrt{\sigma^2/\rho}}{\varphi(\Phi^{-1}(\tilde v))}
	dH(\sigma)
	\]
	The capacity constraint is $\mathbb{P}( v \geq \lambda^*_2) = \alpha$, 
	with cut-off value, $\lambda^*_2$, satisfying,
	\[
		\alpha = \mathbb{P}( v \geq  \lambda_2^*)
		 = 1- \int \Phi\Big(\theta_\alpha \frac{\sqrt{\sigma^2+\sigma_\theta^2}}{\sigma_\theta^2}   
		- \Phi^{-1}(1-\lambda_2^*) \sqrt{\sigma^2/\sigma_\theta^2} \Big) dH(\sigma).
	\]
To find $\lambda_1^*$, we can use the formula provided in Proposition \ref{prop: rule_ysigma}. 
A more direct approach is to recognize, see Section \ref{sec:adaptive}, that 
the FDR control constraint is defined as $\gamma = \mathbb{E}[(1- v) \11\{ v \geq 
\lambda_1^*\}]/\mathbb{P}( v \geq  \lambda_1^*)$, where the cut-off value $\lambda_1^*$ is defined through 
	\[
	\gamma = \int_{\lambda_1^*}^{1}(1-v) f_{ v}(v; \alpha) dv /\int_{\lambda_1^*}^{1} f_{ v}(v;\alpha) dv.
	\]
	Let $\lambda^* = \max\{\lambda_1^*, \lambda_2^*\}$, the selection region is then $\{(y, \sigma): y \geq t_\alpha(\lambda^*,\sigma)\}$ with 
	\[
	t_\alpha(\lambda^*,\sigma) = \theta_\alpha / \rho - 
	\Phi^{-1} (1-\lambda^*) \sqrt{\sigma^2/ \rho }.
	\]

	Suppose we use the posterior mean of $\theta$ as a ranking device, so 
	$\delta_i^{\texttt{PM}} = \11\{ y \rho  \geq \omega^*\}$ for some suitably chosen $\omega^*$ 
	that guarantees both capacity and FDR control. For the capacity constraint, 
	the thresholding value solves, 
	\begin{align*}
		1- \alpha &= \int \mathbb{P}(y \rho < \omega_2^*)dH(\sigma) \\
		& = \int \Phi \Big( \omega_2^* / (\sigma_\theta^2\sqrt{\sigma_\theta^2 + \sigma^2})\Big) dH(\sigma),
	\end{align*}
	while FDR control requires a thresholding value that solves, 
	\begin{align*}
	\gamma &= \int \mathbb{P}(y \geq \omega_1^* / \rho , \theta< \theta_\alpha)dH(\sigma)
	/\int \mathbb{P}(y \geq \omega_1^* / \rho )dH(\sigma)\\
	& =\int  \int _{[\omega_1^* / \rho, + \infty)} (1-\alpha) f_0(y) dy dH(\sigma)
	/\int 1- \Phi \Big ( \omega_1^* /  (\sigma_\theta^2\sqrt{\sigma_\theta^2 + \sigma^2}\Big)dH(\sigma),
	\end{align*}
	with 
	\[
	f_0(y) = \frac{1}{1-\alpha} \frac{1}{\sqrt{2\pi (\sigma_\theta^2+\sigma^2)}} 
	\exp\Big\{ - \frac{y^2}{2(\sigma_\theta^2+\sigma^2)}\Big\}
	\Phi\Big( \frac{(\theta_\alpha - y \rho)}{\sqrt{\rho \sigma^2} }\Big),
	\]
	denoting the density of $y$ under the null $\theta< \theta_\alpha$.
	Setting $\omega^* = \max\{\omega_1^*, \omega_2^*\}$, 
	the selection region is then $ \{(y,\sigma): y \geq  \omega^* / \rho \}$.

\begin{figure}[h!]
\begin{center}
\resizebox{.9\textwidth}{!}{\includegraphics{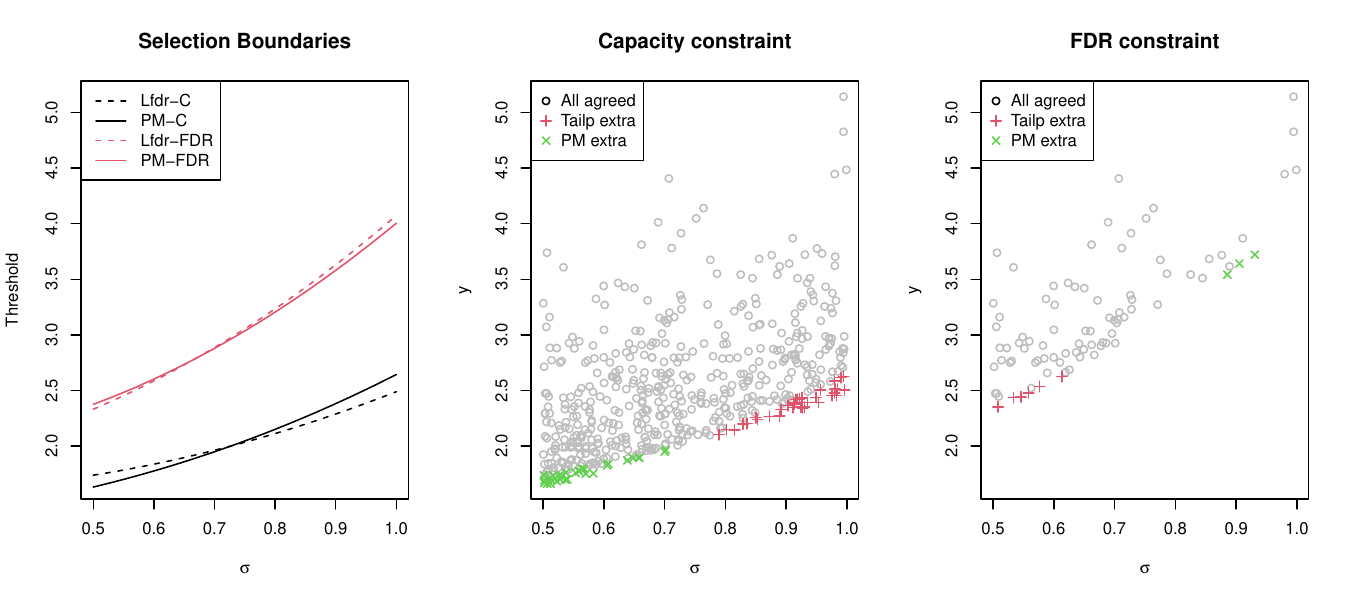}}
\end{center}
\caption{\small The left panel plots the selection boundaries for the normal-normal model with 
$\sigma_\theta^2 = 1$ and $\alpha = 0.05$ and $\gamma = 0.2$.  
The density of $\sigma$ is assumed to be uniform on the range $[0.5, 1]$.
Selected units must have $(y_i , \sigma_i )$ above the curves.  The red curves correspond 
to the selection region boundaries with FDR controlled at level $0.2$; 
solid lines for posterior mean ranking and dash line for posterior tail probability ranking. 
The black curves correspond to the selection boundaries with capacity control at level $0.05$.  
The middle and right panels illustrate the selected set from a realized sample 
of size 10,000. The grey circles correspond to individuals selected by both the posterior tail 
probability rule and the posterior mean rule. The green crosses depict individuals 
selected by the posterior mean rule but not the tail probability and the red crosses indicate 
individuals selected by the tail probability rule, but not by the posterior mean rule. }
\label{fig:nn}
\end{figure}

Figure \ref{fig:nn} plots the selection boundaries for both constraints with 
$\theta \sim \mathcal{N}(0,1)$ and $\sigma \sim U[0.5, 1]$. 
With $\alpha = 0.05$ and $\gamma = 0.2$, the FDR constraint is binding, 
but not the capacity constraint. In this example, if we only impose the capacity constraint to be 5 percent, 
even an Oracle totally aware of $G$, will face a false discovery rate of nearly 52 percent.  
In other words more than half of those selected to be in the right tail will be individuals with 
$\theta < \theta_\alpha$ rather than from the intended $\theta \geq \theta_\alpha$ group.  
This fact motivates our more explicit incorporation of FDR into the selection constraints.  
We may recall that in the homogeneous variance Gaussian setting we saw in Figure \ref{fig.FDRFNR} 
that FDR was also very high when $\alpha$ is set at 0.05. 
Figure \ref{fig:nn} also depicts the selected set with a realized sample of
10,000 from the normal-normal model. With capacity constraint alone, the posterior mean 
criteria favours individuals with smaller variances. When the FDR constraint is implemented, 
with $\gamma = 0.2$, it becomes the binding constraint in this setting, both criteria 
become more stringent and only a much smaller set of individuals are selected, and there
is less conflict in the selections. The corresponding selected sets are plotted in the 
right panel of Figure \ref{fig:nn}. 
When the variance parameter $\sigma_\theta^2$ in $G$ is not observed, we can estimate it via the MLE based on the marginal likelihood of $Y$. This leads to a generalized James-Stein estimator of the type
proposed in \citeasnoun{EfronMorris}. 
\end{example}

\begin{example}[{\bf{Discrete $G$}}] \label{eg: discreteS3}
Suppose $\theta \sim 0.85 \delta_{-1} + 0.1 \delta_{2} + 0.05 \delta_5$. Then the marginal density of $y$ given $\sigma^2$ takes the form, 
\begin{align*}
f(y|\sigma^2) & = \int \varphi(y|\theta,\sigma^2) dG(\theta)\\ 
& = \Big( 0.85 \varphi(y|1,\sigma^2) + 0.1 \varphi(y|2,\sigma^2) + 0.05 \varphi(y|5,\sigma^2)).
\end{align*}
And the random variable $ v$ is a transformation of the pair $(y, \sigma^2)$, defined as, 
\begin{align*}
v &= \psi(y, \sigma^2) := \mathbb{P}(\theta \geq  \theta_\alpha|y, \sigma^2) = \frac{\int_{\theta_\alpha}^{+\infty} \varphi(y|\theta,\sigma^2)dG(\theta)}{\int_{-\infty}^{\infty} \varphi(y|\theta,\sigma^2)dG(\theta)}.
\end{align*}
The capacity constraint leads to a thresholding rule on $v$ such that 
$\mathbb{P}(v \geq \lambda_2^*) = \alpha$, while the FDR control leads to a cutoff value, 
$\lambda_1^*$ defined through $\gamma = \mathbb{E}[(1- v )\11\{ v \geq \lambda_1^*\}]/\mathbb{P}( v \geq \lambda_1^*)$. Let $\lambda^* = \max\{\lambda_1^*, \lambda_2^*\}$, 
the selection region is then defined by $\{(y,\sigma): y \geq t_\alpha(\lambda^*, \sigma)\},$ 
and can be found easily numerically. 

\begin{figure}[h!]
	\begin{center}
		\resizebox{.9\textwidth}{!}{\includegraphics{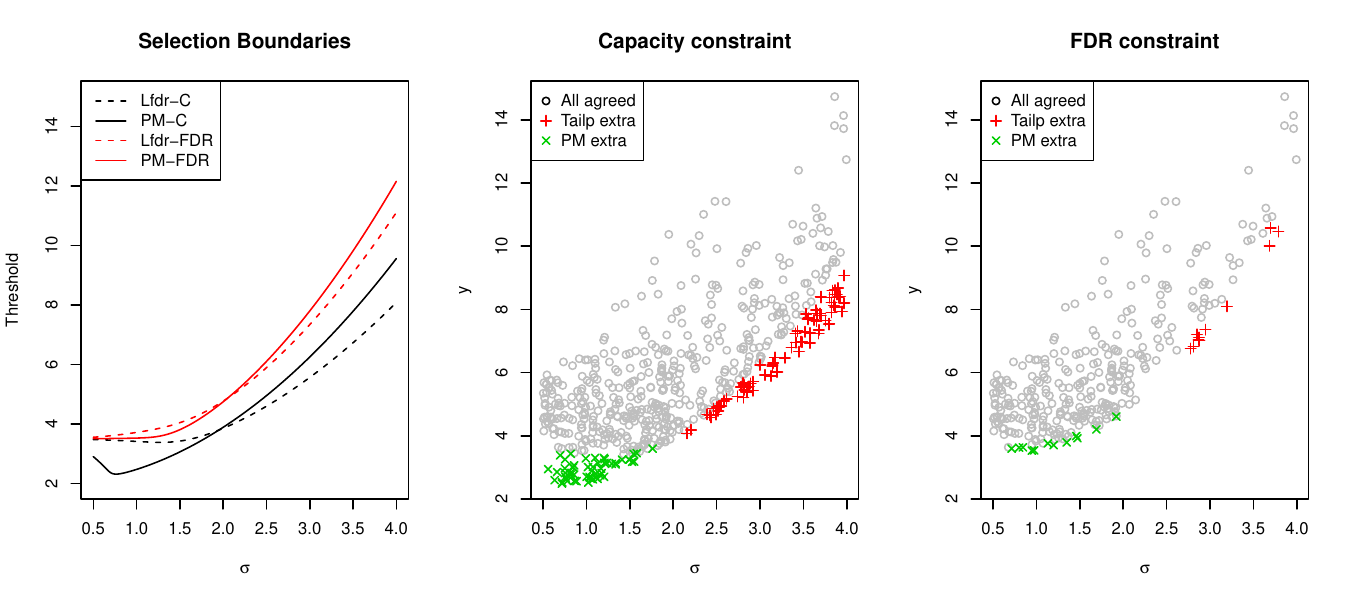}}
	\end{center}
	\caption{\small The left panel plots the selection boundaries for the normal-discrete model with 
	$\theta \sim  G = 0.85 \delta_{-1} + 0.1 \delta_{2} + 0.05 \delta_{5}$ and $\alpha = 0.05$ and 
	$\gamma = 0.1$.  The density of $\sigma$ is uniform on the range $[0.5, 4]$. 
	The red curves correspond to the selection region with FDR controlled at level $0.2$, 
	solid lines for posterior mean ranking and dashed lines for posterior tail probability ranking. 
	The black curves correspond to the selection region with capacity control at level $0.05$. 
	The other panels are structured as in the previous figure.}
	\label{fig:nd}
\end{figure}

Figure \ref{fig:nd} plots the selection boundaries for both constraints when
$\theta$ follows this discrete distribution.  We again set $\alpha = 0.05$ and $\gamma = 0.2$, 
so we would like to select all the individuals associated with the largest effect size,
$\{\theta = 5\}$, while controlling the FDR rate below $20\%$. 
The red curves again correspond to FDR control with the two ranking procedure, 
while the black curves correspond to capacity control. For the two regions to 
overlap with $\alpha$ fixed at $0.05$, we must be willing to tolerate $\gamma \approx 37\%$. 
In this case, we see that the posterior probability ranking procedure prefers individuals 
with larger variances, while the posterior mean ranking procedure prefers smaller variances. 
Based on a realized sample of 10,000, Figure \ref{fig:nd} again shows the
selected observations and once more we see that 
the posterior mean criteria favours individuals with smaller variances, both under 
the capacity constraint and the FDR constraint.  In contrast to the normal-normal
setting, now the FDR constraint is much less severe and allows us to select considerably
more individuals.

\end{example}

\section{Heterogeneous unknown variances}  \label{sec:UnknownVar}
Assuming that the $\sigma_i$'s are known, up to a common scale parameter, may be plausible
in some applications such as baseball batting averages, but it is frequently more plausible
to adopt the view that we are simply confronted with estimates of scale available perhaps from
longitudinal data.  In such cases we need to consider the pairs, $(y_i , S_i)$
as potentially jointly dependent random variables arising from the longitudinal model,
\[
Y_{it} = \theta_i + \sigma_i \epsilon_{it}, \quad  \epsilon_{it} \sim_{iid} \mathcal{N}(0,1), \quad (\theta_i, \sigma_i^2) \sim G,
\]
with sufficient statistics, $Y_i = T_i^{-1} \sum_{t=1}^{T_i} Y_{it}$ and 
$S_i = (T_i-1)^{-1} \sum_{t=1}^{T_i} (Y_{it} - Y_i)^2$, for $(\theta_i,\sigma_i^2)$. 
Conditional on $(\theta_i,\sigma_i^2)$, we have $Y_i |\theta_i,\sigma_i^2  \sim 
\mathcal{N}(\theta_i, \sigma_i^2/T_i)$ and $S_i |\sigma_i^2$ is distributed as Gamma
with shape parameter $r_i = (T_i-1)/2$, scale parameter, $\sigma_i^2/r_i$, 
and density function denoted as $\Gamma(S_i | r_i,\sigma_i^2/r_i)$. 

Given the loss function (\ref{loss}) and defining $\theta_\alpha$ as 
$\alpha = \mathbb{P}(\theta_i \geq \theta_\alpha) = 
\int \int_{\theta_\alpha}^{+\infty} dG(\theta, \sigma^2)$, the conditional risk is,
\begin{align*}
    \mathbb{E}_{\bm{\theta}|\bm{Y}, \bm{S}}\Big[ L(\bm{\delta}, \bm{\theta})\Big] & = 
   	\sum_{i=1}^n (1-\delta_i) v_\alpha(Y_i,S_i) \\
    & + \tau_1 \Big( \sum_{i=1}^n \{\delta_i (1-v_\alpha(Y_i, S_i)) - \gamma \delta_i\}\Big) 
	+ \tau_2 \Big( \sum_{i=1}^n \delta_i - \alpha n\Big) 
\end{align*}
with
\begin{align*}
v_\alpha(y_i,s_i) & = \mathbb{P}(\theta_i \geq \theta_\alpha | Y_i = y_i, S_i = s_i)\\ 
& = \frac{\int \int_{\theta_\alpha}^{+\infty} \Gamma(s_i | r_i, \sigma_i^2/r_i)\varphi(y_i | \theta,\sigma^2/T_i) dG(\theta,\sigma^2)}
{\int \int  \Gamma(s_i | r_i, \sigma_i^2/r_i)\varphi(y_i | \theta,\sigma^2/T_i) dG(\theta,\sigma^2)}.
\end{align*}
Taking expectations with respect to $(Y, S)$, the Bayes rule solves, 
\[
\underset{\bm{\delta}}{\min} \mathbb{E} \Big [ \sum_{i=1}^n (1-\delta_i) v_\alpha(y_i,s_i)\Big] + \tau_1 \Big( \mathbb{E}\Big[ \sum_{i=1}^n \Big\{(1-v_\alpha(y_i,s_i))\delta_i - \gamma \delta_i\Big\}\Big]\Big) + \tau_2 \Big (\mathbb{E}\Big[ \sum_{i=1}^n \delta_i\Big] - \alpha n\Big). 
\]

Before characterizing the Bayes rule any further, we should observe that when variances $\sigma^2$ 
are not directly observed, the tail probability $v_\alpha(Y,S)$ may no longer have the 
monotonicity property we have described above. 

\begin{lemma} \label{lem: nonmon}
	Consider the transformation $v_\alpha(Y,S) = \mathbb{P}(\theta \geq \theta_\alpha|Y,S]$, then for fixed $S = s$, the function $v_\alpha(Y, s)$ may not be monotone in $Y$; and for fixed $Y = y$, the function $v_\alpha(y,S)$ may not be monotone in $S$. 
\end{lemma}

\begin{proposition}\label{prop: nonmon}
For pre-specified $(\alpha,\gamma)$ such that $\gamma < 1-\alpha$, the Bayes selection rule takes the form 
\[
\delta_i^* = \11\{v_\alpha(Y,S) \geq  \lambda^*(\alpha,\gamma)\}
\]
where $\lambda^*(\alpha,\gamma) = \max\{\lambda_1^*(\alpha,\gamma), \lambda_2^*(\alpha)\}$ with
\[
\lambda_1^*(\alpha,\gamma) = \min \Big \{\lambda: \mathbb{E}\Big [(1-v_\alpha(Y,S)-\gamma)\11\{v_\alpha(Y,S) \geq \lambda\}\Big ]  \leq 0 \Big \}
\]
and 
\[
\lambda_2^*(\alpha) = \min\Big\{\lambda: \mathbb{P}(v_\alpha(Y,S) \geq \lambda) - \alpha \leq 0\Big\}
\]
\end{proposition}

Based on the Bayes rule, the selected set is defined as 
\[
\Omega_{\alpha, \gamma} = \{(Y,S): 
v_\alpha(Y,S) \geq \lambda^*(\alpha,\gamma)\}. 
\]


\begin{remark}
Note that for each prespecified pair $(\alpha,\gamma)$, $\Omega_{\alpha, \gamma}$  
is just the $\lambda^*(\alpha,\gamma)$-superlevel set of the function $v_\alpha(Y,S)$. 
For any $\alpha_1 > \alpha_2$, nestedness of the selected sets would mean that the 
$\lambda^*(\alpha_2,\gamma)$-superlevel set of the function $v_{\alpha_2}$ must be a subset 
of the $\lambda^*(\alpha_1, \gamma)$-superlevel set of the function $v_{\alpha_1}$.
The construction and the form of the optimal selection rule may appear to be very similar to the case where $\sigma_i^2$ is observed. 
However, the crucial difference is that we no longer require the independence between 
$\theta$ and $\sigma^2$ in this section. In contrast, when $\sigma_i^2$ is assumed to be 
directly observed, the independence assumption is critical for all the derivations. 
For instance, the non-null proportion, defined as $\mathbb{P}(\theta_i \geq \theta_\alpha)$, 
must change for different values of $\sigma_i$ if we allow the distribution of 
$\theta$ to depend on $\sigma$. 
\end{remark}

\subsection{A Conjugate Gaussian Example} 
\label{eg: normaNIX}

Suppose we have balanced panel data $y_{i1}, \dots, y_{iT} \sim \mathcal{N}(\theta,\sigma^2)$ 
with sample means $Y_i = \frac{1}{T}\sum_t y_{it}$ and sample variances 
$S_i = \frac{1}{T-1}\sum_t (y_{it} - Y_i)^2$. 
Further, suppose that $G(\theta,\sigma^2)$ takes the normal-inverse-chi-squared form,
$\mbox{NIX} (\theta_0, \kappa_0, \nu_0, \sigma_0^2) = 
\mathcal{N}(\theta | \theta_0, \sigma^2/\kappa_0) \chi^{-2}(\sigma^2 | \nu_0, \sigma_0^2)$. 
Integrating out $\sigma^2$, the marginal distribution of $\theta$ becomes a Student $t$ distribution,
\[
\frac{\theta - \theta_0}{\sigma_0/\sqrt{\kappa_0}} \sim t_{\nu_0}
\]
where $t_{\nu_0}$ is the $t$-distribution with degree of freedom $\nu_0$. 
Therefore, the $1-\alpha$ quantile of $\theta$, denoted $\theta_\alpha$ is simply,
\[
\theta_\alpha =\theta_0 + \frac{\sigma_0}{\sqrt{\kappa_0}} F^{-1}_{t_{\nu_0}}(1-\alpha) 
\]
where $F^{-1}_{t_{\nu_0}}$ denotes the quantile function of $t_{\nu_0}$. 

Conjugacy of the distribution $G$ implies that the posterior distribution of $(\theta, \sigma^2|Y,S)$ 
follows $NIX(\theta_T, \kappa_T, \nu_T, \sigma_T^2) = 
\mathcal{N}(\theta | \theta_T, \sigma_T^2/\kappa_T) \chi^{-2}(\sigma^2| \nu_T, \sigma_T^2)$  with 
\begin{align*}
\nu_T  &= \nu_0 + T\\
\kappa_T & = \kappa_0 + T\\
\theta_T & = \frac{\kappa_0 \theta_0 + T Y}{\kappa_T}\\
\sigma_T^2 & = \frac{1}{\nu_T} \Big( \nu_0 \sigma_0^2 + (T-1)S + \frac{T\kappa_0}{\kappa_0 + T}(\theta_0 - Y)^2\Big).
\end{align*}
Integrating out $\sigma^2$, the marginal posterior of $\theta$ again follows a $t$-distribution, 
\[
\frac{\theta - \theta_T}{\sigma_T/\sqrt{\kappa_T}} \sim t_{\nu_T}.
\]
It is thus clear that the posterior mean of $\theta$, is simply a linear function of $Y$
and independent of $S$,
\[
\mathbb{E}[\theta | Y,S] = \theta_T = \frac{\kappa_0 \theta_0 + T Y}{\kappa_T},
\]
and the posterior tail probability is given by,
\[
v_{\alpha}(Y,S) = \mathbb{P}(\theta \geq \theta_\alpha|Y,S) = \mathbb{P}\Big(\frac{\theta - \theta_T}{\sigma_T/\sqrt{\kappa_T}} \geq \frac{\theta_\alpha  - \theta_T}{\sigma_T/\sqrt{\kappa_T}} | Y,S\Big) = 1- F_{t_{\nu_T}}\Big( \frac{\theta_\alpha  - \theta_T}{\sigma_T/\sqrt{\kappa_T}}\Big). 
\]

\begin{figure}[bp]
\begin{center}
\resizebox{.85\textwidth}{!}{\includegraphics{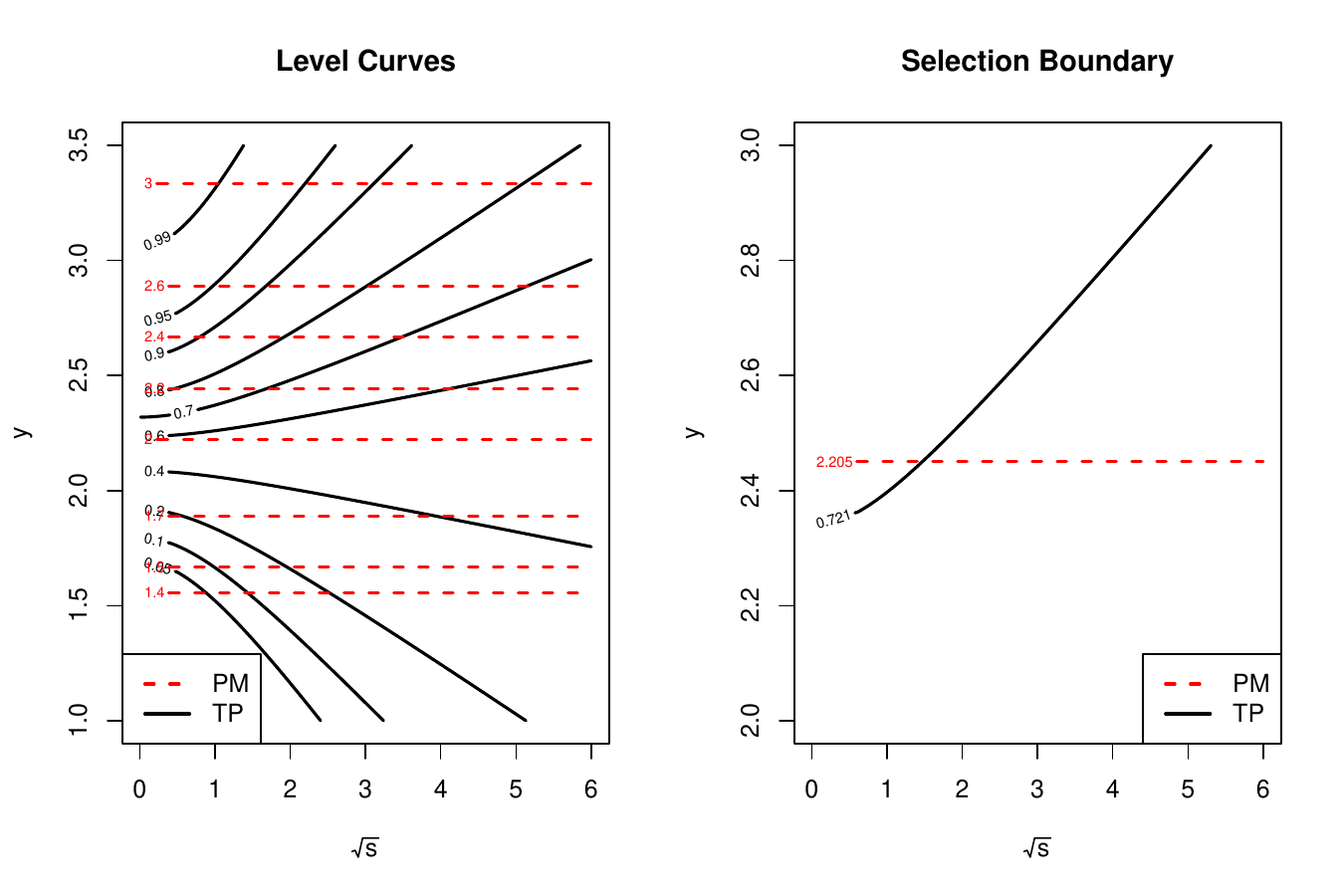}}
\end{center}
\caption{\small The left panel shows level curves of the posterior mean (marked as red dashed lines) 
and posterior tail probability (marked as black solid lines) for the normal model with 
$(\theta, \sigma^2) \sim NIX(0,1,6,1)$ and panel time dimension $T = 9$. 
The right plot shows the boundary of the selection region based on posterior mean ranking 
(marked as the red dashed line) and the posterior tail probability ranking (marked as the solid black 
line) with $\alpha = 5\%$ and $\gamma = 10\%$. }
\label{fig.normconj}
\end{figure}

\begin{figure}[h!]
	\begin{center}
		\resizebox{.85\textwidth}{!}{\includegraphics{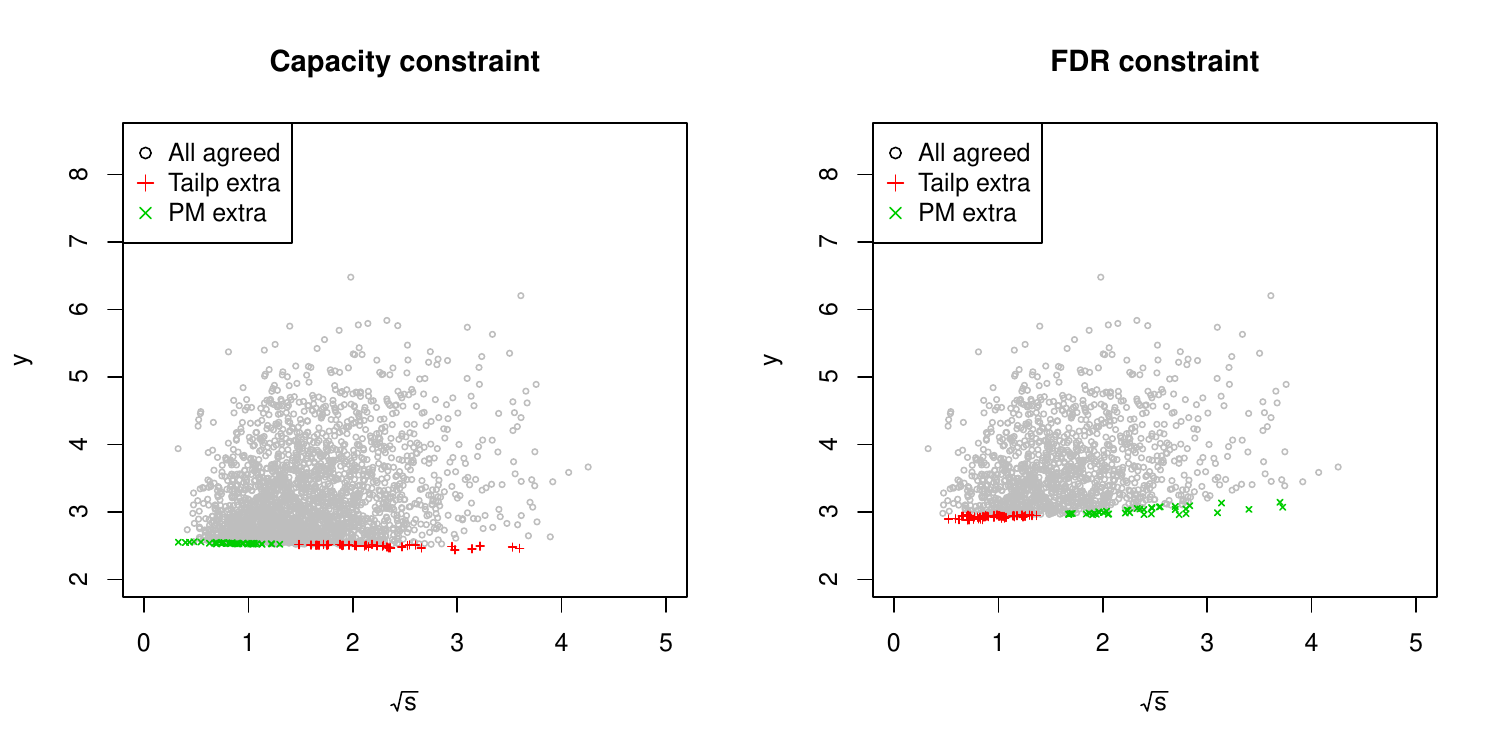}}
	\end{center}
	\caption{\small Selection set comparison for one sample realization from the normal model with 
		$(\theta, \sigma^2) \sim NIX(0,1,6,1)$ and panel time dimension $T = 9$: The left panel shows in grey  circles the agreed selected elements by both the posterior mean and posterior tail probability criteria under the capacity constraint, extra elements selected by the posterior mean are marked in green and extra elements selected by the posterior tail probability rule are marked in red. The right panel shows the comparison of the selected sets under both the capacity and FDR constraint with $\alpha = 5\%$ and $\gamma = 10\%$.}
	\label{fig.normconjc}
\end{figure}

To illustrate this case, suppose $\theta_0=0$, $\kappa_0 = 1$ , $\sigma_0^2 = 1$ and $\nu_0 = 6$ and $T=9$, 
it can be verified that $v_\alpha(Y,S)$ is in fact a monotone function of $Y$ for each fixed $S$ and any 
$\alpha > 0$, hence in this example we can invert the function $v_\alpha(y,s)$ to obtain the level curves. 
The left panel of Figure \ref{fig.normconj} shows the level curves for $v_\alpha(Y,S)$ and 
$\mathbb{E}(\theta|Y,S)$ for $\alpha = 5\%$. It is clear that the posterior mean is a constant 
function of $S$, while the posterior tail probability exhibits more exotic behaviour with 
respect to $S$, especially for more extreme values of $Y$. If we fix $S = s_0$, 
then $v_\alpha(Y,s_0)$ is an increasing function of $Y$. On the other hand, fixing $Y = y_0$, 
for small $y_0$ implies that $v_\alpha(y_0,S)$ is a increasing function of $S$, while for $y_0$ large, 
$v_\alpha(y_0,S)$ becomes a decreasing function of $S$. 

A capacity constraint of size $\alpha$ implies the thresholding rule, 
\[
\mathbb{P}(v_\alpha(Y,S)\geq  \lambda_2^*) = \alpha,
\]
while FDR control at level $\gamma$ leads to a cutoff value $\lambda_1^*$ defined as 
\[
\gamma = \mathbb{E}[(1-v_\alpha(Y,S) )\11\{v_\alpha(Y,S)\geq \lambda_1^*\}]/\mathbb{P}(v_\alpha(Y,S) \geq  \lambda_1^*).
\]
The larger of the two thresholds, denoted $\lambda^* = \max  \{\lambda_1^*, \lambda_2^*\}$ defines the selection region 
based on posterior tail probability ranking $\Omega_{\alpha,\gamma}= \{(Y,S): v_\alpha(Y,S) \geq \lambda^*\}$. 
For $\alpha = 5\%$ and $\gamma = 10\%$, the selection region
based on the tail probability rule is $ \{(Y,S): v_\alpha(Y,S)\geq 0.72\}$. 
The posterior mean ranking is defined as $\{(Y,S): \mathbb{E}[\theta|Y,S]\geq 2.2\}$. 
These selection boundaries are depicted as the red dashed line and black solid line respectively in 
the right panel of Figure \ref{fig.normconj}.  In this case, the FDR constraint binds. If only the capacity 
constraint were in place, we would have a cutoff for tail probability at $0.40$ and the cutoff for the
posterior mean at $1.84$.  Figure \ref{fig.normconjc} further shows the comparison of the selected set 
based on a sample realization from the model. 

In Appendix \ref{sec:DiscreteEx} we consider a more complex bivariate discrete example that illustrates
somewhat more exotic behavior of the decision boundaries and compares performance of several different
ranking and selection rules.

\subsection{Variants of the unknown variance model} \label{sec: variant}
We have assumed that the only scale heterogeneity is driven by $\sigma_i$ in the above model, but often there may be more heteroskedasticity that should be allowed in $\epsilon_{it}$. Here we consider a variant where 
\[
Y_{it} = \theta_i + \sigma_i \epsilon_{it}, \quad \epsilon_{it} \sim \mathcal{N}(0, 1/w_{it}), \quad (\theta_i, \sigma_i^2) \sim G.
\]
We will assume that $w_{it} \sim H$ are known quantities and are independent from $(\theta_i, \sigma_i^2)$. 
Denoting $w_i = \sum_{t=1}^{T_i} w_{it}$, the sufficient statistics now take the form 
$Y_i = \sum_{t=1}^{T_i} w_{it} Y_{it}/w_i$ and $S_i = (T_i-1)^{-1} \sum_{t=1}^{T_i} (Y_{it} - Y_i)^2$. 
In \citeasnoun{Bayesball} we have illustrated this formulation for predicting baseball batting averages;
in that setting ``at bats'' for player $i$ in year $t$ are given by the $w_{it}$, but there is still 
some player specific heterogeneity in the $\sigma_i$'s representing consistency of batting performance.
It is easy to show that $Y_i |\theta_i, \sigma_i^2 \sim \mathcal{N}(\theta_i, \sigma_i^2/w_i)$ and $S_i|\sigma_i^2$ follows Gamma distribution with shape parameter $r_i = (T_i-1)/2$ and scale parameter $\sigma_i^2/r_i$. The decision rules now become a function of the tuple $(Y_i, S_i, w_i)$, for instance the tail probability can be specified as 
\[
v_\alpha(y_i, s_i, w_i) = \mathbb{P}(\theta_i \geq \theta_\alpha| y_i, s_i, w_i) =  
\frac{\int_{\theta_\alpha}^{+\infty} f(y_i | \theta, \sigma^2/w_i) 
	\Gamma( s_i | r_i, \sigma^2/r_i) dG(\theta, \sigma^2) }
{\int_{-\infty}^{+\infty} f(y_i | \theta, \sigma^2/w_i) 
	\Gamma( s_i | r_i, \sigma^2/r_i) dG(\theta, \sigma^2) }.
\]
and the posterior mean takes the form 
\[
	\mathbb{E}[\theta_i | y_i,s_i,w_i] = 
\int \theta f(y_i, s_i|\theta,\sigma^2, w_i) dG(\theta,\sigma^2).
\]
The threshold values under either the capacity or the FDR constraint can be worked out in a similar fashion. 
For any ranking statistics $\delta(Y_i,S_i,w_i)$ together with a decision rule $\11\{\delta(Y_i,S_i,w_i) \geq \lambda\}$,  the capacity constraint requires choosing a thresholding value $\lambda_2^*(\alpha)$ such that 
\[
\alpha = \int \int \11\{\delta(y, s, w) \geq \lambda_2^*(\alpha)\} 
f(y, s| \theta, \sigma^2, w) dG(\theta, \sigma^2) dH(w), 
\]
while the thresholding value to control in addition the FDR rate under size $\gamma$ requires solving for $\lambda_1^*(\alpha,\gamma)$ such that 
\[
\gamma = \frac{\mathbb{P}(\delta(y,s,w) \geq \lambda_1^*(\alpha, \gamma); 
	\theta < \theta_\alpha)}{\mathbb{P}(\delta(y,s,w) \geq \lambda_1^*(\alpha, \gamma))}
\]
which can be further represented as 
\[
\gamma = \frac{\int \int \11\{\delta(y, s, w) \geq \lambda_1^*(\alpha, \gamma)\} 
	(1-\alpha) f_0(y, s|\theta, \sigma^2, w) dG(\theta, \sigma^2) dH(w)}
{\int \int \11\{\delta(t, s, w) \geq \lambda_1^*(\alpha, \gamma)\} 
	f(y, s|\theta, \sigma^2, w) dG(\theta, \sigma^2) dH(w)}
\]
where $f_0(y, s|\theta, \sigma^2, w)$ is the density of $(y, s)$ under the null hypothesis 
$\theta < \theta_\alpha$. 

%
We can again consider selection region as those plotted in Figure \ref{fig.normconj} and 
Figure \ref{fig.ndpanel} to appreciate how different decision criteria determines the selection. 
As soon as the ranking statistics depends on $w$, the selection region of the thresholding rule 
$\11\{\delta(y,s,w) \geq \lambda^*\}$ will also depend on the magnitude of $w$. 


\section{Asymptotic Adaptivity}  \label{sec:adaptive}
The previous sections propose Bayes rules for minimizing the expected number of missed discoveries 
subject to both capacity and FDR constraints under several modeling environments.  
In each of these environments, the Bayes rule takes the form $\delta^* = 1\{v_\alpha \geq \lambda^*\}$ 
where $v_\alpha$ is defined as the posterior probability of $\theta \geq  \theta_\alpha$ 
conditional on the data. The thresholding value $\lambda^*$ is defined to satisfy both the capacity and FDR constraints. The Bayes rule involves several unknown quantities, in particular the $v_\alpha$'s and the 
optimal thresholding value, $\lambda^*$, that require knowledge on the distribution of $\theta_i$ 
or the joint distribution of $(\theta_i, \sigma_i^2)$ when the variances are latent variables. 
For estimating this distribution of the latent variables, we propose a plug-in procedure that is 
very much in the spirit of empirical Bayes methods pioneered by \citeasnoun{Robbins.56}.  
In this Section we also establish that the resulting feasible rules achieve asymptotic validity 
and  asymptotically attain the same performance as the infeasible Bayes rule. 

We begin by discussing properties of the Oracle procedure assuming that $v_\alpha$ is known and we only need to estimate the optimal thresholding value. We establish asymptotic validity of this Oracle procedure and  then propose a plug-in method for both $v_\alpha$ and the thresholding value thereby establishing the asymptotic validity of the empirical rule. 
Before presenting the formal results, we introduce regularity conditions that will be required. 
We distinguish two cases depending on whether  the $\sigma_i^2$'s are observed. 
\begin{assumption} \label{ass1}
\begin{enumerate}
\item {(Variances observed)} $\{Y_i,\sigma_i^2, \theta_i\}$ are independent and identically 
    distributed with a joint distribution with $\sigma_i^2$ and $\theta_i$ independent. 
    The random variables $\theta_i$ and $\sigma_i^2$ have positive densities with respect to Lebesgue 
    measure on a compact set $\Theta \subset \mathbb{R}$ and $[\underline{\sigma}^2, \overline \sigma^2]$ 
    respectively for some $\underline{\sigma}^2 > 0$ and $\overline \sigma^2 < +\infty$. 
		\item {(Variance unobserved)} Let $S_i$ be an individual sample variance based on $T$ repeated measurements and $Y_i$ be the sample means with $T \geq 4$.  Suppose further that $\{Y_i, S_i, \theta_i, \sigma_i\}$ are independent and identically distributed and that the  random variables $\{\theta_i, \sigma_i^2\}$ have a joint distribution $G$ with a joint density positive everywhere on its support.
\end{enumerate}
\end{assumption}

\subsection{Optimal thresholding}

Whether $\sigma_i^2$ is observed or estimated, the optimal thresholding value can be defined in a unified manner by $\lambda^* = \max \{\lambda_1^*, \lambda_2^*\}$ with 
\begin{align*}
	\lambda_1^* &= \inf \{ t \in (0,1), H_v(t) \geq  1- \alpha\}\\
	\lambda_2^* & = \inf \{t \in (0,1) ,Q(t) \leq \gamma\} 
\end{align*}
where $H_v$ denotes the cumulative distribution of either $v_\alpha(y_i, \sigma_i)$ or $v_\alpha(y_i, s_i)$, induced by the marginal distribution of the data, either as the pair $\{y_i, \sigma_i\}$ when variances are observed or the pair $\{y_i, s_i\}$ otherwise. Hence $\lambda_1^*$ is the $1-\alpha$ quantile of $H_v$. 

The function $Q(t)$ is defined as $Q(t) = \mathbb{E}[(1-v_\alpha) 1\{v_\alpha \geq t\}] / \mathbb{E}[1\{v_{\alpha} \geq t\}]$. Its formulation recalls Proposition \ref{prop: nonmon} and the existence of $\lambda_2^*$ is guaranteed as long as $\alpha < 1- \gamma$. The thresholding value is also equivalent to those defined in Proposition \ref{prop: homo} and Proposition \ref{prop: rule_ysigma}. In particular, the thresholding values $t_1^*$ and $t_2^*$ in Proposition \ref{prop: homo} are cast in terms of $Y$ directly and it is easy to see $\lambda_j^* = v_\alpha(t_j^*)$ for $j = 1,2$ when variances are homogeneous. In a similar spirit, the explicit formulae for $\lambda_1^*$ and $\lambda_2^*$ in Proposition \ref{prop: rule_ysigma} are a result of invoking the monotonicity of $v_\alpha(y, \sigma)$ with respect to $y$ for each fixed value of $\sigma$. The function $Q(t)$ is the mFDR of the procedure $\delta= 1\{v_\alpha \geq t\}$ for any $\alpha\in (0,1)$, and is monotonically decreasing in $t$.  Monotonicity of $Q(t)$ is crucial to justify this thresholding procedure insuring that either the capacity constraint or the mFDR constraint must be binding.  
\citeasnoun{cao2013optimal} have observed  that a sufficient condition for monotonicity for a
broad class of multiple testing procedures is that the ratio of densities 
under the null and alternative of the test statistics employed for ranking be monotone
and they discuss the consequences of the violation of this condition.
For the posterior tail probability criterion this monotone likelihood ratio condition, 
as we will see, can be verified directly.

Recall that mFDR is defined as $\sum_{i=1}^n \mathbb{P}[\delta_i = 1, \theta_i < \theta_\alpha] / \sum_{i=1}^n \mathbb{P}(\delta_i = 1)$. It suffices to show that $\mathbb{P}[\delta_i = 1, \theta_i < \theta_\alpha] = \mathbb{E}[(1-v_{\alpha,i}) \delta_i]$. Since $v_{\alpha,i} = P[\theta_i \geq \theta_\alpha |D_i]   = \alpha f_1(D_i) / f(D_i)$ where $D_i$ is the individual data being either $\{y_i, \sigma_i\}$ or $\{y_i, s_i\}$ depending on the model and $f_1$ is the marginal density of the data when $\theta_i \geq \theta_\alpha$ and $f$ is the marginal density of $D_i$. Then it is easy to see that $\mathbb{P}[\delta_i = 1, \theta_i < \theta_\alpha] = (1-\alpha) \int 1\{v_{\alpha,i} \geq t \} f_0(D_i) dD_i = \int 1\{v_{\alpha, i} \geq t\} (1-v_{\alpha,i}) f(D_i) dD_i = \mathbb{E}[(1-v_{\alpha,i}) 1\{v_{\alpha,i} \geq t\}]$. Then $Q(t) = \int_{t}^1 (1- v) h_v dv/\int_{t}^1 h_v dv$ where $h_v$ is the density function of $v_\alpha$. Monotonicity of $Q(t)$ can be easily verified by showing that the derivative with respect to $t$ of the right hand side quantity is nonpositive.

\subsection{Oracle Procedures} 
The only unknown quantity in the Oracle procedure is the thresholding value and we now discuss how to estimate it to achieve asymptotic validity.  $H_v$ and $Q$ can be estimated by the following quantities:
\begin{align*}
	 H_n(t) &= \frac{1}{n} \sum_{i=1}^n 1\{ v_{\alpha,i} \leq t\}\\
	 Q_n(t) & = \frac{\sum_{i=1}^n (1-v_{\alpha,i}) 1\{v_{\alpha,i} \geq t\}}{\sum_{i=1}^n 1\{v_{\alpha,i} \geq t\}}
	\end{align*}
and the associated thresholding values are then defined as  
$\lambda_n = \max\{ \lambda_{1n},  \lambda_{2n}\}$, with
\begin{align*}
	 \lambda_{1n} &= \inf \{t \in [0,1],  H_n(t) \geq 1- \alpha\}\\
	 \lambda_{2n} & = \inf\{t \in [0,1],  Q_n(t) \leq \gamma\}
	\end{align*} 

\begin{thm}(Asymptotic validity of the Oracle procedure) \label{thm:oracle}
	Under Assumption \ref{ass1}, the procedure $\delta_i = 1\{v_{\alpha,i} \geq  \lambda_n\}$ asymptotically controls the false discovery rate below $\gamma$ and the expected proportion of rejections below $\alpha$ for any $(\alpha, \gamma) \in [0,1]^2$ and $\gamma < 1- \alpha$ when $n \to \infty$, more specifically
	\begin{align*}
	&\underset{n \to \infty} {\limsup} \mathbb{E}\Big[\frac{\sum_{i=1}^n 1\{\theta_i < \theta_\alpha, v_{\alpha,i} \geq   \lambda_n\}}{\sum_{i=1}^n 1\{v_{\alpha,i} \geq  \lambda_n\} \bigvee 1}\Big]  \leq \gamma\\
	 & \underset{n \to \infty}{\limsup} \mathbb{E} \Big [\frac{1}{n} \sum_{i=1}^n 1\{v_{\alpha,i} \geq  \lambda_n\} \Big ] \leq \alpha
	 \end{align*}
\end{thm} 	

\subsection{Adaptive Procedures} 
In practise the posterior tail probability also involves the unknown quantity $\theta_\alpha = G^{-1} (1 - \alpha)$ that needs to be estimated. We propose a plug-in estimator in the spirit of the empirical Bayes method: estimating $G$ by its nonparametric maximum likelihood estimator $\hat G_n$ and estimating $\theta_\alpha$ as its $1-\alpha$ quantile. 

Consistency of the nonparametric maximum likelihood estimator, $\hat G_n$, was first proven by \citeasnoun{KW} using Wald type arguments. A Hellinger risk bound for the associated marginal density estimate and adaptivity of $\hat G_n$ and a self-regularization property have been recently established in  \citeasnoun{saha2020nonparametric}  and \citeasnoun{PW20}.  In particular, the following established result, stated here as an assumption, is crucial for establishing the asymptotic validity of the adaptive procedure. 
\begin{assumption} \label{ass2}
The nonparametric maximum likelihood estimator $\hat G_n$ is strongly consistent for $G$. 
That is, for all continuity points $k$ of $G$, $\hat G_n(k) \to G(k)$ almost surely 
as $n \to \infty$. Furthermore, the estimated marginal (mixture) density converges almost 
surely in Hellinger distance.
\end{assumption}  

When variances are homogeneous or when variances are unknown but we have longitudinal  data so that we have a  mixture model for the pair $\{Y_i, S_i\}$, the Hellinger convergence is established in \citeasnoun{van1993hellinger}. When variances are heterogeneous but known, the Hellinger bound for marginal density has been established recently in \citeasnoun{jiang2020general}.  

The plug-in estimators for the posterior tail probability, $v_\alpha(y_i, \sigma_i)$ when variances are known or $v_\alpha(y_i, s_i)$ when variances are unknown is then defined respectively as 
\begin{align*}
	\hat v_\alpha(y_i, \sigma_i) & = \int_{\hat \theta_\alpha}^{+\infty} \varphi(y_i | \theta,\sigma_i^2) d\hat G_n(\theta)/\int_{-\infty}^{+\infty} \varphi(y_i | \theta, \sigma_i^2) d\hat G_n(\theta)\\
	\hat v_\alpha(y_i, s_i) & =  \int_{\hat \theta_\alpha}^{+\infty} f(y_i, s_i | \theta, \sigma) d\hat G_n(\theta,\sigma^2)/\int_{-\infty}^{+\infty} f(y_i, s_i | \theta, \sigma)  d\hat G_n(\theta, \sigma^2)\\
\end{align*}
where $f$ is the density function for $(y_i, s_i)$ which is a product of Gaussian and gamma densities. 
Abbreviating the estimated posterior tail probability by $\hat v_{\alpha,i}$, 
we mimic the Oracle procedure and estimate the thresholding value by 
$\hat \lambda_n  = \max \{\hat \lambda_{1n}, \hat \lambda_{2n}\}$, where,
\begin{align*}
	\hat \lambda_{1n} & = \inf \{t \in [0,1]: \frac{1}{n} \sum_{i=1}^n 1\{\hat v_{\alpha,i} \leq t\} \geq 1-\alpha \}\\
	\hat \lambda_{2n} & = \inf \{t \in [0,1]: \frac{\sum_{i=1}^n (1-\hat v_{\alpha,i}) 1\{\hat v_{\alpha,i} \geq t\}}{\sum_{i=1}^n 1\{\hat v_{\alpha,i} \geq t\}} \geq \gamma \}
\end{align*}

\begin{thm}(Asymptotic validity of adaptive procedure)\label{thm:adaptive}
	Under Assumptions \ref{ass1} and \ref{ass2}, the adaptive procedure $\delta_i = 1\{\hat v_{\alpha,i}\geq \hat \lambda_n\}$ asymptotically controls the false discovery rate below $\gamma$ and the expected proportion of rejections below $\alpha$ for any $(\alpha, \gamma) \in [0,1]^2$ with $\alpha < 1- \gamma$ when $n \to \infty$, more specifically
	\begin{align*}
	&\underset{n \to \infty} {\limsup} \mathbb{E}\Big[\frac{\sum_{i=1}^n 1\{\theta_i < \theta_\alpha, \hat v_{\alpha,i} \geq \hat  \lambda_n\}}{\sum_{i=1}^n 1\{\hat v_{\alpha,i} \geq \hat  \lambda_n\} \bigvee 1}\Big]  \leq \gamma\\
	& \underset{n \to \infty}{\limsup} \mathbb{E}\Big [\sum_{i=1}^n 1\Big \{\frac{1}{n} \hat v_{\alpha,i} \geq \hat \lambda_n \Big \} \Big ] \leq \alpha
	\end{align*} 
\end{thm} 	

	It is clear that given the Lagrangian formulation of the compound decision problem, it can be viewed equivalently as a constrained optimization problem.  See also the discussion in Remark \ref{remark: power}.  We seek to maximize power defined as $\beta(t) := \mathbb{P}(\theta_i \geq \theta_\alpha, \delta_i = 1)/\alpha $ 
subject to two constraints: the first is the marginal FDR rate and the other is the selected proportion. 
For each fixed pair of $\{\alpha, \gamma\}$, the Bayes rule achieves the best power among all 
thresholding procedures that respect the two constraints. The next theorem establishes that our 
feasible, adaptive procedure achieves the same power as the oracle rule asymptotically.  
It is supported by the simulation evidence presented in the next section.
In practice we suggest convolution smoothing of the discrete $\hat G$ as in \citeasnoun{JZ21} with
a bandwidth slowly tending to zero with $n$.  As they show, the smoothed mixing distribution is
also consistent, hence fulfilling Assumption \ref{ass2}  and therefore all our adaptivity results.

\begin{thm} \label{thm:adaptivepower}
	Under Assumption \ref{ass1} and \ref{ass2}, the adaptive procedure $\delta_i = 1\{ \hat v_{\alpha,i} \geq \hat \lambda_n\}$  attains the same power as the optimal Bayes rule asymptotically. In particular, as $n \to \infty$, 
	\begin{align*}
		\frac{\sum_{i=1}^n 1 \{ \theta_i \geq \theta_\alpha, \hat v_{\alpha,i} \geq \hat \lambda_n \} }{\sum_{i=1}^n 1 \{\theta_i \geq \theta_\alpha\} } \overset{p}{\to} \beta(\lambda^*)
	\end{align*}
\end{thm}

\section{Simulation Evidence}
\label{sec:Simulation}
In this Section we describe two small simulation exercises designed to illustrate
performance of several competing methods for ranking and selection.  As a benchmark
for evaluating performance we consider several Oracle methods that presume knowledge
of the true distribution, $G$, generating the $\theta$'s as well as several
feasible methods that rely on estimation of $G$.  These are contrasted with more
traditional methods that are based on linear shrinkage rules of the Stein type. The linear shrinkage rule 
is the posterior mean of $\theta$ under the assumption that $G$ follows a Gaussian distribution 
with unknown mean and variance parameters. This is the commonly used estimator for ranking 
and selection in applied work, notably Chetty, Friedman and Rockoff (2014a, 2014b) for teacher 
evaluation and \citeasnoun{chetty2018impacts} for studying intergenerational mobility. 

Typically the 
linear shrinkage estimator is used in the context of heterogeneous known variances, this will 
be the model we focus on in our simulation experiments. The linear shrinkage formula defined 
in \eqref{tweedie} easily adapts to the heterogeneous variances case and leads to the 
Jame-Stein shrinkage rule with heterogeneous known variances. \citeasnoun{EfronMorris} 
introduced some further modifications. As we have already demonstrated, when variances 
are heterogeneous, the linear shrinkage estimator provides a different ranking than the
posterior tail probability rules. Further complications arise when we seek procedures 
that also control false discovery.  To estimate the false discovery rate for different 
thresholding values we requires knowledge of $G$.  If the Gaussian assumption on $G$ underlying
the linear shrinkage rules is misplaced, it may lead to an inaccurate estimates of FDR, 
and consequently to  procedures that fail to control for false discovery.

Performance will be evaluated primarily on the basis of power, which we define as the
proportion of individuals whose true $\theta_i$ exceeds the cutoff 
$\theta_\alpha = G^{-1} (1 - \alpha)$,
who are actually selected.  This is the sample counterpart of 
$\PP(\delta_i = 1, \theta_i \geq \theta_\alpha)/\PP(\theta_i \geq \theta_\alpha)$. 
FDR is calculated as the sample counterpart of $\PP(\delta_i = 1, 
\theta_i < \theta_\alpha)/\PP(\delta_i = 1)$, that is the proportion of
selected individuals whose true $\theta_i$ falls below the threshold.  
While our selection rules are {\it intended} to constrain FDR below the $\gamma$
threshold, as in other testing problems they are not always successful in this objective 
in finite samples so empirical power comparisons must be interpreted cautiously in view of this.
Nonetheless, asymptotic validity is assured by the results in Section \ref{sec:adaptive}.
We compare performance for three distinct $\alpha$ levels, 
$\{ 0.05, 0.10, 0.15 \}$ and three $\gamma$ levels $\{ 0.05, 0.10, 0.15 \}$. 

\subsection{The Student $t$ Setting}
\begin{figure}[h!]
    \begin{center}
    \resizebox{.9\textwidth}{!}{\includegraphics{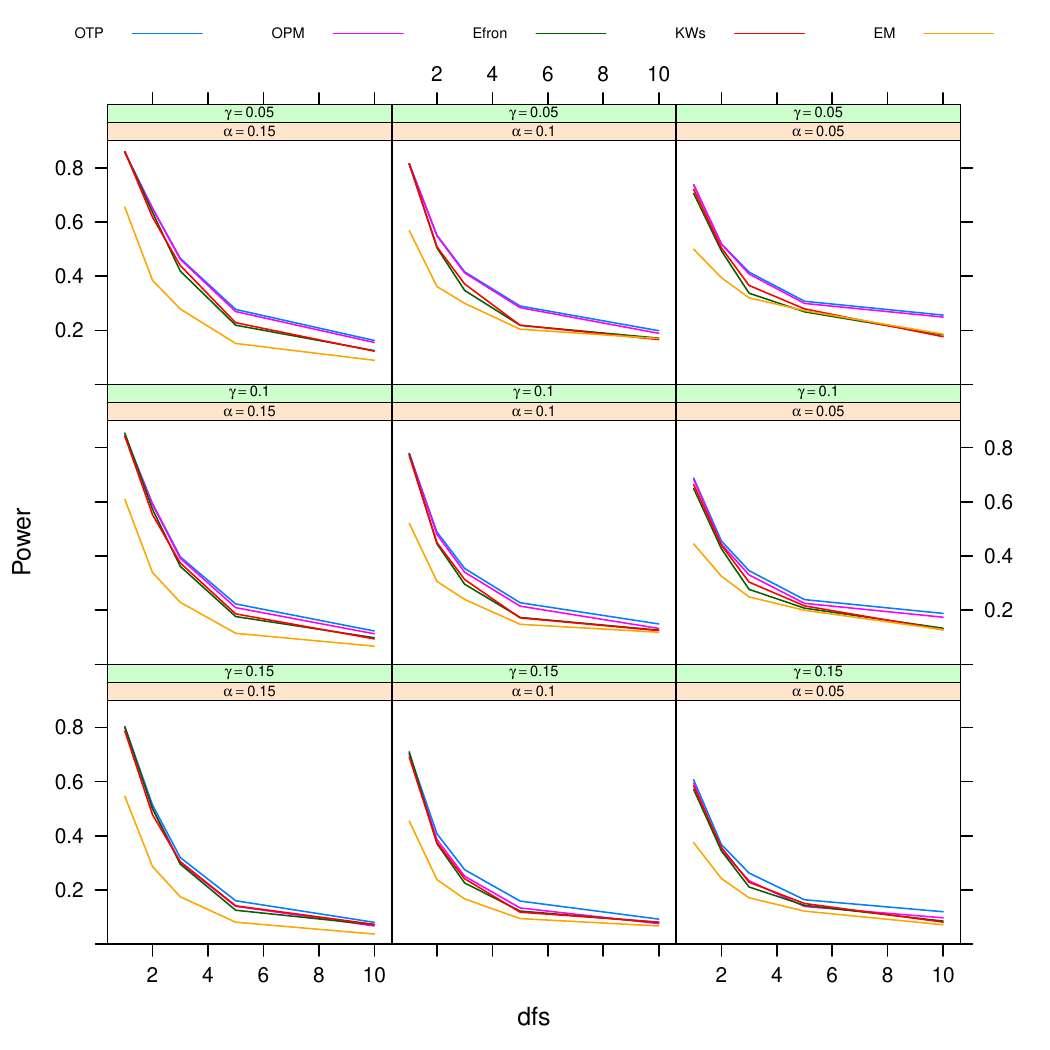}}
    \end{center}
    \caption{\small{Power Performance for Several Selection Rules with Student $t$ Signal.
    Capacity and FDR constraints are indicated at the top of each panel in the Figure.}}
    \label{fig: nt}
\end{figure} 
Our first simulation setting focuses on the effect of tail behavior of the distribution
on performance of competing rules.  For these simulations we take $G$ to be a discrete approximation
to Student $t$ distributions with degrees of freedom in the set $\{ 1,2,3,5,10\}$, and supported
on the interval $[-20,20]$.  The scale parameters of the Gaussian noise contribution are independent 
and uniformly distributed on the interval $[0.5,1.5]$. We report power performance for several alternative ranking and selection rules:  

\begin{description}
	\item [OTP] Oracle Tail Probability Rule
	\item [OPM] Oracle Posterior Mean Rule
	\item [Efron] Efron Tail Probability Rule 
	\item [KWs] Kiefer-Wolfowitz Smoothed Tail Probability Rule 
	\item [EM] Efron and Morris (1973) Linear Shrinkage Rule 
\end{description}
The KWs rule uses $\tilde G = \hat G * K_h$, with biweight kernel $K$ and bandwidth
$h$ equal to half the mean absolute deviation from the median of $\hat G$.  The Efron
rule uses his suggested default of a natural spline basis with five degrees of freedom
and penalty parameter 0.1.

We illustrate the results in Figure \ref{fig: nt}, where we plot empirical power against degrees
of freedom of the $t$ distribution for a selected set of values for the capacity constraint,
$\alpha \in \{ 0.05, 0.10, 0.15 \}$ and FDR constraint, $\gamma \in \{ 0.05, 0.10, 0.15 \}$ as
indicated at the top of each panel of the figure.  The most striking
conclusion from this exercise is the dramatic  decrease in power as we move toward the Gaussian
distribution.  At the Cauchy, $t_1$, power is quite respectable  for all choices of $\alpha$ and
$\gamma$, but power declines rapidly as the degrees of freedom increases, reenforcing our earlier conclusion
that the Gaussian case is extremely difficult.  We would stress, in view of this finding, that classical
linear shrinkage procedures designed for the Gaussian setting are poorly adapted to heavy tailed settings 
in which the reliability of selection procedures is potentially greatest. 

Careful examination of this figure also reveals that there is a slight advantage to the posterior
tail probability rules over the posterior mean procedures, both for the Oracle rules and for
our feasible procedures. There is surprisingly little sacrifice in power in moving from the
Oracle methods to the Efron or Kiefer-Wolfowitz rules.  The Efron and Morris selection rule is
very competitive in the almost Gaussian, $t_{10}$ setting but sacrifices considerable power in the
lower degrees of freedom settings due to the misspecification of the distribution $G$ 
and consequent inaccurate estimation of the false discovery rate.

\subsection{A Teacher Value-Added Setting}

Our second simulation setting is based on a discrete approximation of the data
structure employed in \citeasnoun{GGM} to study teacher value-addded methods. 
Several longitudinal waves of student test scores from the Los Angeles Unified
School District were combined in this study.  Here we abstract from many features 
of the full longitudinal structure of this data, and focus instead on comparing performance 
of several selection methods.  We maintain our standard known variance model in which
we observe $Y_i \sim \NN ( \theta_i, \sigma_i^2)$ with $\theta_i$'s drawn iidly from a
distribution $\tilde G$ estimated by \citeasnoun{GGM}.  This distribution was estimated
from the full longitudinal LA sample using the nonparametric maximum likelihood estimator
of Kiefer and Wolfowitz and then smoothed slightly by convolution with a biweight kernel
and illustrated in the left panel of Figure \ref{fig:LAmodel}.  
Variances, in keeping with our hypothesis in Section \ref{sec:KnownVariances}, are drawn from a 
distribution with density illustrated in the right panel of Figure \ref{fig:LAmodel}.  
We focus on selection from the left tail of the resulting distribution since it is
those teachers  whose jobs are endangered by recent policy recommendations in the
literature. (see for instance \citeasnoun{hanushek2011economic}). 

\begin{figure}[h!]
    \begin{center}
    \resizebox{.7\textwidth}{!}{\includegraphics{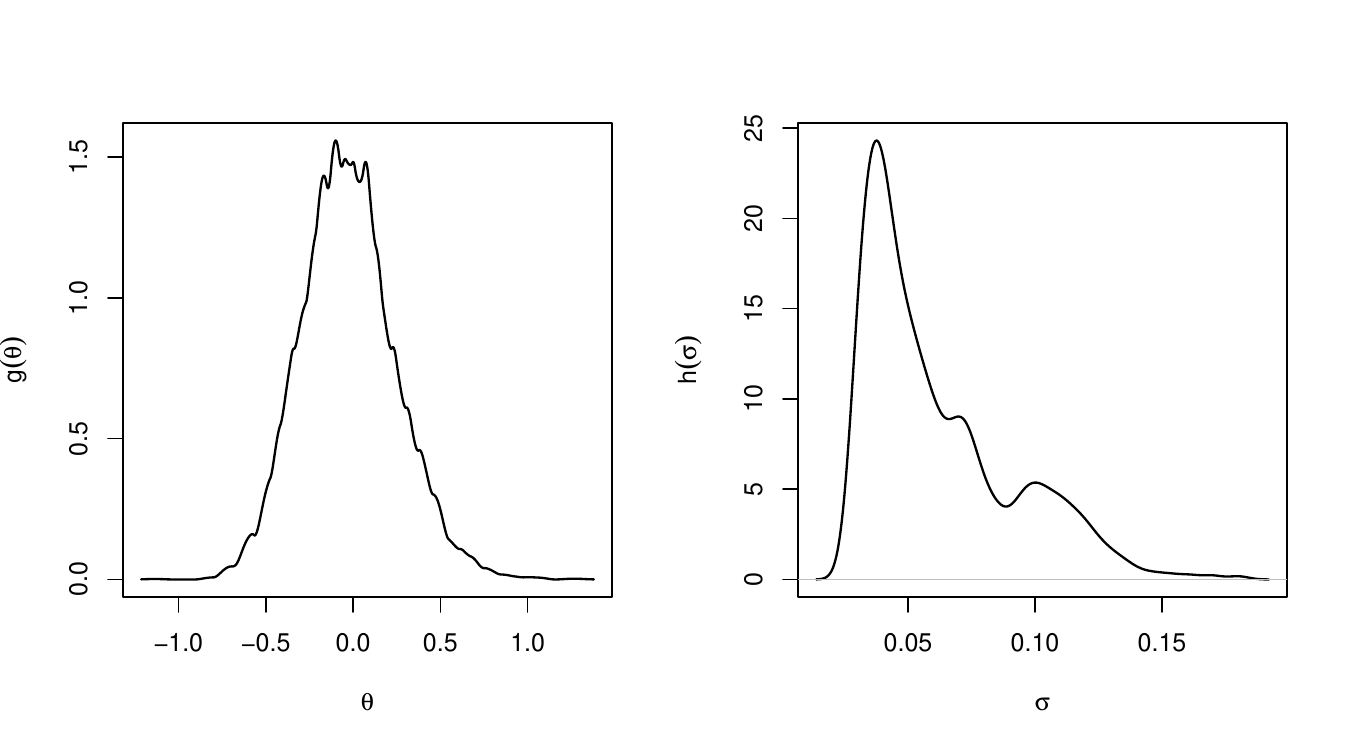}}
    \end{center}
    \caption{\small{Densities of ``Latent'' (Mean) Ability and Standard Deviation for 
    the Teacher Value Added Simulations}}
    \label{fig:LAmodel}
\end{figure} 
We draw samples of size 10,000 from the foregoing distribution and compute performance 
measures based on 100 replications.  The fitted densities for this simulation exercise
are based on a sample of roughly 11,000 teachers, so the simulation sample size is
chosen to be commensurate with this.  In Table \ref{tab:LArankinga} we report power,
FDR and the proportion selected by ten selection rules.  The Oracle rules, OTP and OPM,
based ranking by the tail probability and posterior mean criteria can be considered
benchmarks for the remaining feasible procedures.  Only the Oracle procedures can be
considered reliable from the perspective of adhering to the capacity and FDR constraints.
Consequently, some caution is required in the interpretation of the power comparisons
since feasible procedures can exhibit good power at the expense of violating these constraints.
This is analogous to the common difficulty in interpreting power in testing problems
when different procedures have differing size.
When FDR is constrained to 5\%, even the Oracle is only able to select about half of
the deserving individuals; OTP is consistently preferable to OPM as expected and 
power performance improves somewhat as the capacity constraint is relaxed.  Among the feasible
$G$-modeling selection procedures, the Efron rules have good power performance, but fail
to meet the FDR constraints.  We conjecture that somewhat less aggressive smoothing than
the default, $df = 5, c_0 = 0.1$ might help to rectify this.  In contrast the smoothed 
Kiefer-Wolfowitz rules are somewhat overly conservative in meeting the FDR constraints
and might benefit from somewhat more aggressive smoothing.

\nocite{Chettya,Chettyb}
Among the other procedures the linear posterior mean rule, LPM, as employed by
Chetty, Friedman and Rockoff (2014a, 2014b), 
and the linear posterior mean rule, EM, of \citeasnoun{EfronMorris} behave 
identically and exhibit somewhat erratic
FDR control due to the misspecified Gaussian assumption on $G$;
this leads to weaker power performance. 
As a further comparison, when the linear shrinkage rules are implemented without any FDR constraint, 
denoted LPM* and EM* in the table, as they typically would be used in practice, 
the false discovery proportion is considerably higher than the targeted $\gamma$. 
We also report the performance of MLE and P-value rules, implemented without FDR control; 
again both yield a higher FDR rate, making it difficult to evaluate their power performance.

{\tiny{
\begin{table}[!tbp]
\begin{center}
\begin{tabular}{lrrrrcrrrr}
\hline\hline
\multicolumn{1}{l}{\bfseries }&\multicolumn{4}{c}{\bfseries $\gamma=5\%$}&\multicolumn{1}{c}{\bfseries }&\multicolumn{4}{c}{\bfseries $\gamma=10\%$}\tabularnewline
\cline{2-5} \cline{7-10}
\multicolumn{1}{l}{}&\multicolumn{1}{c}{$\alpha=1\%$}&\multicolumn{1}{c}{$\alpha=3\%$}&\multicolumn{1}{c}{$\alpha=5\%$}&\multicolumn{1}{c}{$\alpha=10\%$}&\multicolumn{1}{c}{}&\multicolumn{1}{c}{$\alpha=1\%$}&\multicolumn{1}{c}{$\alpha=3\%$}&\multicolumn{1}{c}{$\alpha=5\%$}&\multicolumn{1}{c}{$\alpha=10\%$}\tabularnewline
\hline
{\bfseries Power}&&&&&&&&&\tabularnewline
~~OTP&$0.394$&$0.520$&$0.554$&$0.626$&&$0.494$&$0.625$&$0.661$&$0.733$\tabularnewline
~~OPM&$0.365$&$0.492$&$0.521$&$0.599$&&$0.484$&$0.620$&$0.654$&$0.731$\tabularnewline
~~ETP&$0.435$&$0.540$&$0.580$&$0.657$&&$0.540$&$0.647$&$0.688$&$0.759$\tabularnewline
~~KWTP&$0.355$&$0.477$&$0.521$&$0.614$&&$0.452$&$0.583$&$0.631$&$0.723$\tabularnewline
~~EPM&$0.398$&$0.511$&$0.552$&$0.632$&&$0.528$&$0.642$&$0.683$&$0.758$\tabularnewline
~~KWPM&$0.325$&$0.447$&$0.492$&$0.588$&&$0.440$&$0.576$&$0.624$&$0.719$\tabularnewline
~~LPM&$0.162$&$0.341$&$0.418$&$0.689$&&$0.246$&$0.460$&$0.542$&$0.805$\tabularnewline
~~LPM*&$0.726$&$0.781$&$0.796$&$0.829$&&$0.726$&$0.781$&$0.796$&$0.829$\tabularnewline
~~EM&$0.162$&$0.341$&$0.418$&$0.689$&&$0.246$&$0.460$&$0.542$&$0.805$\tabularnewline
~~EM*&$0.726$&$0.781$&$0.796$&$0.829$&&$0.726$&$0.781$&$0.796$&$0.829$\tabularnewline
~~MLE&$0.699$&$0.768$&$0.787$&$0.824$&&$0.699$&$0.768$&$0.787$&$0.824$\tabularnewline
~~P-val&$0.374$&$0.478$&$0.535$&$0.635$&&$0.374$&$0.478$&$0.535$&$0.635$\tabularnewline
\hline
{\bfseries FDR}&&&&&&&&&\tabularnewline
~~OTP&$0.050$&$0.050$&$0.051$&$0.051$&&$0.103$&$0.103$&$0.100$&$0.101$\tabularnewline
~~OPM&$0.047$&$0.050$&$0.051$&$0.053$&&$0.103$&$0.101$&$0.101$&$0.102$\tabularnewline
~~ETP&$0.070$&$0.059$&$0.061$&$0.064$&&$0.128$&$0.115$&$0.117$&$0.119$\tabularnewline
~~KWTP&$0.035$&$0.037$&$0.041$&$0.048$&&$0.082$&$0.081$&$0.085$&$0.096$\tabularnewline
~~EPM&$0.063$&$0.057$&$0.059$&$0.062$&&$0.129$&$0.115$&$0.116$&$0.118$\tabularnewline
~~KWPM&$0.038$&$0.040$&$0.045$&$0.051$&&$0.081$&$0.084$&$0.087$&$0.097$\tabularnewline
~~LPM&$0.016$&$0.025$&$0.033$&$0.083$&&$0.031$&$0.048$&$0.061$&$0.151$\tabularnewline
~~LPM*&$0.276$&$0.226$&$0.207$&$0.172$&&$0.276$&$0.226$&$0.207$&$0.172$\tabularnewline
~~EM&$0.016$&$0.025$&$0.033$&$0.083$&&$0.031$&$0.048$&$0.061$&$0.151$\tabularnewline
~~EM*&$0.276$&$0.225$&$0.207$&$0.172$&&$0.276$&$0.225$&$0.207$&$0.172$\tabularnewline
~~MLE&$0.304$&$0.238$&$0.216$&$0.177$&&$0.304$&$0.238$&$0.216$&$0.177$\tabularnewline
~~P-val&$0.627$&$0.526$&$0.467$&$0.365$&&$0.627$&$0.526$&$0.467$&$0.365$\tabularnewline
\hline
{\bfseries Selected}&&&&&&&&&\tabularnewline
~~OTP&$0.004$&$0.016$&$0.029$&$0.066$&&$0.006$&$0.021$&$0.037$&$0.082$\tabularnewline
~~OPM&$0.004$&$0.015$&$0.027$&$0.063$&&$0.005$&$0.021$&$0.036$&$0.081$\tabularnewline
~~ETP&$0.005$&$0.017$&$0.031$&$0.070$&&$0.006$&$0.022$&$0.039$&$0.086$\tabularnewline
~~KWTP&$0.004$&$0.015$&$0.027$&$0.064$&&$0.005$&$0.019$&$0.034$&$0.080$\tabularnewline
~~EPM&$0.004$&$0.016$&$0.029$&$0.067$&&$0.006$&$0.022$&$0.038$&$0.086$\tabularnewline
~~KWPM&$0.003$&$0.014$&$0.026$&$0.062$&&$0.005$&$0.019$&$0.034$&$0.080$\tabularnewline
~~LPM&$0.002$&$0.010$&$0.022$&$0.075$&&$0.003$&$0.014$&$0.029$&$0.095$\tabularnewline
~~LPM*&$0.010$&$0.030$&$0.050$&$0.100$&&$0.010$&$0.030$&$0.050$&$0.100$\tabularnewline
~~EM&$0.002$&$0.010$&$0.022$&$0.075$&&$0.003$&$0.014$&$0.029$&$0.095$\tabularnewline
~~EM*&$0.010$&$0.030$&$0.050$&$0.100$&&$0.010$&$0.030$&$0.050$&$0.100$\tabularnewline
~~MLE&$0.010$&$0.030$&$0.050$&$0.100$&&$0.010$&$0.030$&$0.050$&$0.100$\tabularnewline
~~P-val&$0.010$&$0.030$&$0.050$&$0.100$&&$0.010$&$0.030$&$0.050$&$0.100$\tabularnewline
\hline
\end{tabular}
\caption{Comparison of Performance of Several Selection Rules for the Teacher Value Added Simulation.\label{tab:LArankinga}}\end{center}
\end{table}

}}

\section{Ranking and Selection of U.S. Dialysis Centers}
\label{sec:Dialysis}

Motivated by important prior work on ranking and selection 
by Lin, Louis, Paddock and Ridgeway (2006, 2009) \nocite{LLPR06,LLPR09} illustrated by
applications to ranking U.S. dialysis centers, we have chosen to maintain this focus
to illustrate our own approach.  Kidney disease is a growing medical problem in the U.S
and considerable effort has been devoted to data collection and evaluation of the relative performance of 
the more than 6000 dialysis centers serving the afflicted population.  Centers are
evaluated on multiple criteria, but the primary focus of center ranking is their
standardized mortality rate, or SMR, the ratio of observed deaths to expected deaths
for center patients.  Allocating patients to centers is itself a complex task since
patients may move from one center to another in the course of a year.  Centers also
vary considerably in the mix of patients they serve.  Predictions from an estimated Cox proportional
hazard model that attempts to account for this heterogeneity are employed to estimate
expected deaths for each center.

Our analysis focuses exclusively on the SMR evaluation of centers using longitudinal
data from 2004-18 as reported in \citeasnoun{DataSource}.  We restrict attention to
3230 centers that have consistently reported SMR data over this sample period. Observed 
deaths, denoted $y_{it}$ for center $i$ in year $t$ are conventionally modeled as Poisson,
\[
y_{it} \sim \Pois (\rho_i \mu_{it})
\]
where $\mu_{it}$ is center $i$'s expected deaths as predicted by the Cox model
in year $t$ and $\rho_i$ is the center's unobserved mortality rate.  
We view $\mu_{it}$ as the effective sample size for the center, after adjustment
for patient characteristics of the center.  Center characteristics are expliciitly
excluded from the Cox model.  The classical variance stabilizing transformation for the
Poisson brings us back to the Gaussian model,
\[
z_{it} = \sqrt{y_{it}/\mu_{it}} \sim \NN (\theta_i, 1/w_{it}),
\]
where $\theta_i = \sqrt{\rho_i}$ and $w_{it} = 4 \mu_{it}$. 
Exchangeability of the centers yields a mixture model in which the parameter $\theta_i$,
is effectively assumed to be drawn iidly from a distribution, $G$. 
The predictions of expected mortality, $\mu_{it}$, are assumed to be sufficiently accurate that
we treat $w_{it}$ as known, and independent of $\theta_i \sim G$. 

Over short time horizons like 3 years we assume that centers have a fixed draw of $\theta_i$ from $G$, and
thus we have sufficient statistics for $\theta_i$, as,
\[
T_i = \sum_{t \in \mathcal{T}} w_{it} z_{it}/w_i \sim \NN (\theta_i, 1 /w_{i}),
\]
where the set $\mathcal{T}$ is the corresponding three year window and $w_i = \sum_t w_{it}$.  Given these ingredients it is straightforward
to construct a likelihood for the mixing distribution, $G$, and proceed with estimation of it. 

Our objective is then to select centers based on the posterior distributions 
of their $\theta_i$'s.  For example, the posterior tail probability of center $i$ is given by,
\[
v_\alpha(t_i, w_i) = \mathbb{P}(\theta_i \geq \theta_\alpha| t_i, w_i) =  
	\frac{\int_{\theta_\alpha}^{+\infty} f(t_i | \theta, w_i) 
	dG(\theta) }
	{\int_{-\infty}^{+\infty} f (t_i | \theta, w_i) 
dG(\theta) },
\]
where $f$ is the density function of $T_i$ conditional on $\theta_i$ and $w_i$. 
The capacity constraint requires choosing a thresholding value $\lambda_2^*(\alpha)$ such that 
\[
 	\alpha = \int \int \11\{v_\alpha(t, w) \geq \lambda_2^*(\alpha)\} 
	\varphi (t| \theta, w) dG(\theta) dH(w), 
\]
which can be approximated by $\frac{1}{n}\sum_i \11\{v_\alpha(t_i, w_i)\geq \lambda_2^*(\alpha) \}$,
and inverted to obtain the threshold.
Based on the discussion in Section \ref{sec:adaptive}, for the FDR constraint we chose a thresholding value $\lambda_1^*(\alpha, \gamma)$ such that 
\begin{equation}\label{eq: fdr}
	\gamma = \frac{\int \int \11\{v_{\alpha}(t, w) \geq \lambda_1^*(\alpha, \gamma)\} 
	(1- v_\alpha(t,w)) f(t|\theta,w) dG(\theta) dH(w)}
	{\int \int \11\{v_{\alpha}(t, w) \geq \lambda_1^*(\alpha, \gamma)\} 
	f(t|\theta, w) dG(\theta) dH(w)}
\end{equation}
where $H$ is the marginal distribution of the observed portion of
the variance effect. The numerator can be approximated by 
$\frac{1}{n}\sum_i (1- v_\alpha(t_i, w_i)) \11\{v_\alpha(t_i,w_i) \geq \lambda_1^*(\alpha,\gamma)\}$ 
and the denominator can be approximated by, 
$\frac{1}{n} \sum_i \11\{v_\alpha(t_i,w_i) \geq \lambda_1^*(\alpha,\gamma)\}$. 

The posterior mean ranking, in contrast, is based on, 
\[
\delta(t_i,w_i) = 	\mathbb{E}[\theta_i | t_i,w_i] = 
\int \theta f(t_i|\theta, w_i) dG(\theta).
\]
For the capacity constraint we choose a thresholding value $C_2^*(\alpha)$ such that 
\[
\alpha = \int \int \11\{\delta(t,w) \geq C_2^*(\alpha)\} f(t| \theta,  w) dG(\theta ) dH(w).
\]
For FDR constraint, we pick a thresholding value $C_1^*(\alpha, \gamma)$ such that 
\[
\gamma = \frac{\mathbb{P}(\delta(t,w) \geq C_1^*(\alpha, \gamma); 
\theta < \theta_\alpha)}{\mathbb{P}(\delta(t, w) \geq C_1^*(\alpha, \gamma))}.
\]
The right hand side of the FDR constraint can be approximated by 
\[\frac{1}{n} \sum_i \11\{\delta(t_i,s_i,w_i)\geq C_1^*(\alpha,\gamma)\} 
(1-v_\alpha(t_i,w_i))/ \frac{1}{n} \sum_i \11\{\delta(t_i,w_i)\geq C_1^*(\alpha,\gamma)\}, 
\]
while the right hand side of the capacity constraint can be approximated by 
\[
\frac{1}{n} \sum_i \11\{\delta(t_i,w_i)\geq C_2^*(\alpha)\}, 
\]
so $C_2^*(\alpha)$ is simply the empirical quantile of the $\delta(t_i, w_i)$. 

We will compare the foregoing ranking and selection rules with more naive rules based upon the
Poisson and Gaussian MLEs,  $\sum_{t \in \mathcal{T}} y_{it} / \sum_{t\in \mathcal{T}} \mu_{it}$, and $T_i$, respectively, a
variant of the much maligned P-value, as well as a linear shrinkage procedure. 
For these rules we do not attempt to control for FDR since this is how they are typically 
implemented in practice.

To help appreciate the difficulty of the selection task, Table \ref{tab: naive_31_1d} reports estimated FDR 
rates for several selection rules under a range of capacity constraints $\alpha$ for both 
right and left tail selection based on the data from 2004 - 2006. Right tail selection corresponds
to identifying centers whose mortality rate is higher than expected; left tail selection to
centers with mortality lower than expected.  To estimate FDR we require an estimate of the
distribution of distribution, $G$.  For this purpose we use the smoothed version of the Kiefer-Wolfowitz
NPMLE introduced in Section \ref{sec:Compound}.  The biweight bandwidth for the smoothing was chosen 
as the mean absolute deviation from the median of the discrete NPMLE, $\hat G$. 
The assessment of FDR reported in Table \ref{tab: naive_31_1d} reflects the
considerable uncertainty associated with the selected set of centers deemed by the capacity constraint to 
be in the upper (or lower) $\alpha$ quantile based upon our estimate of the distribution, $G$, of 
unobserved quality.
	
The MLE rule ranks centers based on their Gaussian MLE, $T_i$, while the Poisson-MLE rule ranks on
$\sum_t  y_{it} /\sum_t \mu_{it}$, which is the MLE of $\rho_i$ from the Poisson model. 
P-value ranks centers based on the variance stabilizing transformation from the Poisson model 
under the null hypothesis $\rho_i = 1$ and $\rho_i > 1$ as the alternative hypothesis for right selection 
and $\rho_i < 1$ for the left selection. All these rules ignore the compound decision perspective 
of the problem entirely.

Among the compound decision rules, we consider the linear (James-Stein) shrinkage rule, 
$\hat \mu_\theta + (T_i - \hat \mu_\theta) \hat \sigma_\theta^2 / (\hat \sigma_\theta^2 + 1/w_i)$ 
which is the posterior mean of $\theta_i$ based on the model $T_i \sim \mathcal{N}(\theta_i, 1/w_i)$ 
assuming that the latent variable $\theta_i$ follows a Gaussian
distribution with mean $\mu_\theta$ and variance $\sigma_\theta^2$. We also consider the 
\citeasnoun{EfronMorris}  estimator
which is  a slight modification of the James-Stein estimator.

Finally, PM and TP are the posterior mean of $\theta$ and posterior tail probability of 
$\theta \geq \theta_\alpha$, for right selection, and $\theta \leq \theta_\alpha$ for left selection
based on our estimated $\hat G$.  For both left and right tail selection, as $\alpha$ increases, 
the FDR rate decreases, indicating the selection task becomes easier. All rules that account for 
the compound decision perspective of the problem have slightly lower FDRs than those that consider
each center individually.

{\tiny{
\begin{table}[!tbp]
\begin{center}
\begin{tabular}{lrrrrr}
\hline\hline
\multicolumn{1}{l}{}&\multicolumn{1}{c}{$\alpha=4\%$}&\multicolumn{1}{c}{$\alpha=10\%$}&\multicolumn{1}{c}{$\alpha=15\%$}&\multicolumn{1}{c}{$\alpha=20\%$}&\multicolumn{1}{c}{$\alpha=25\%$}\tabularnewline
\hline
{\bfseries Right Selection}&&&&&\tabularnewline
~~MLE&$0.544$&$0.481$&$0.436$&$0.403$&$0.352$\tabularnewline
~~Poisson-MLE&$0.545$&$0.485$&$0.440$&$0.406$&$0.355$\tabularnewline
~~Pvalue&$0.532$&$0.475$&$0.432$&$0.399$&$0.349$\tabularnewline
~~Efron-Morris&$0.521$&$0.473$&$0.429$&$0.398$&$0.349$\tabularnewline
~~James-Stein&$0.521$&$0.473$&$0.429$&$0.398$&$0.349$\tabularnewline
~~PM&$0.517$&$0.472$&$0.428$&$0.398$&$0.349$\tabularnewline
~~TP&$0.517$&$0.471$&$0.428$&$0.397$&$0.349$\tabularnewline
\hline
{\bfseries Left Selection}&&&&&\tabularnewline
~~MLE&$0.611$&$0.565$&$0.481$&$0.449$&$0.393$\tabularnewline
~~Poisson-MLE&$0.600$&$0.561$&$0.478$&$0.448$&$0.393$\tabularnewline
~~Pvalue&$0.620$&$0.565$&$0.477$&$0.450$&$0.393$\tabularnewline
~~Efron-Morris&$0.595$&$0.552$&$0.472$&$0.445$&$0.391$\tabularnewline
~~James-Stein&$0.595$&$0.552$&$0.472$&$0.445$&$0.391$\tabularnewline
~~PM&$0.592$&$0.552$&$0.473$&$0.445$&$0.391$\tabularnewline
~~TP&$0.589$&$0.550$&$0.471$&$0.444$&$0.390$\tabularnewline
\hline
\end{tabular}
\caption{FDR Estimates: 2004-2006\label{tab: naive_31_1d}}\end{center}
\end{table}
}}

The \citeasnoun{Stars} assigns ratings of five stars down to one star to centers in the proportions 
$\{0.22, 0.30, 0.35, 0.09, 0.04\}$ respectively.  We will abbreviate these ratings to the
conventional academic scale of A-F.  To illustrate the conflict between the selection 
criteria we plot in Figure \ref{fig:conflict} the 
centers selected for the grade A (Five star, which consists 22\% of the centers that suppose to have their true mortality rate being the lowest) category with and without FDR control. Centers are
characterized by pairs, $(T_i, w_i)$, consisting of their weighted mean standarized mortality, $T_i$,
and their estimate of the precision, $w_i$, of these mortality estimates.  In each plot the solid
curves represent the decision boundaries of the selection rule under comparison.  
Centers with low mortality and relatively high precision appear toward the northwest in each figure.  

Panel (a) of the figure compares the posterior tail probability selection with the MLE, or fixed
effect, selection.  The selection boundary for the MLE is the (red) vertical line, since the MLE
ignores the precision of the estimates entirely.  The selection boundary for the tail probability rule 
is indicated by the (blue) curve.  A few centers with high precision excluded by the MLE rule are
selected by the TP rule, and on the contrary a few centers with low precision are selected by the
MLE rule, but excluded by the TP criterion. 
Panel (b) imposes FDR control with $\gamma = 0.20$ on the TP selection with an estimated thresholding value implied by the FDR constraint using the smoothed NPMLE. The MLE selection is the same as in Panel (a) without the FDR control. We see that under TP rule with FDR control, the number of selected centers is reduced considerably. Instead of selecting 711 centers allowed by the capacity constraint, it selects only 230 centers. In comparison, the MLE rule under capacity constraint has an estimated FDR rate at 0.431.
Panel (c) compares selected centers by the TP rule with those selected by a James-Stein linear
shrinkage rule.  Now the TP rule tolerates a few more low precision centers, while it is the
James-Stein rule that demands higher precision to be selected.
Finally, in Panel (d) we again subject the TP rule to FDR control of 20 percent, 
while the James-Stein rule continues to adhere only to the capacity constraint. The TP boundary scales 
back substantially, suggesting that a large proportion of the extra selections made by James-Stein 
linear shrinkage rules are likely to be false discoveries. In fact, the estimated FDR rate of the James-Stein rule under just capacity constraint is also 0.431, the same as that of the MLE rule.


\begin{figure}
    \begin{center}
    \resizebox{.9\textwidth}{!}{\includegraphics{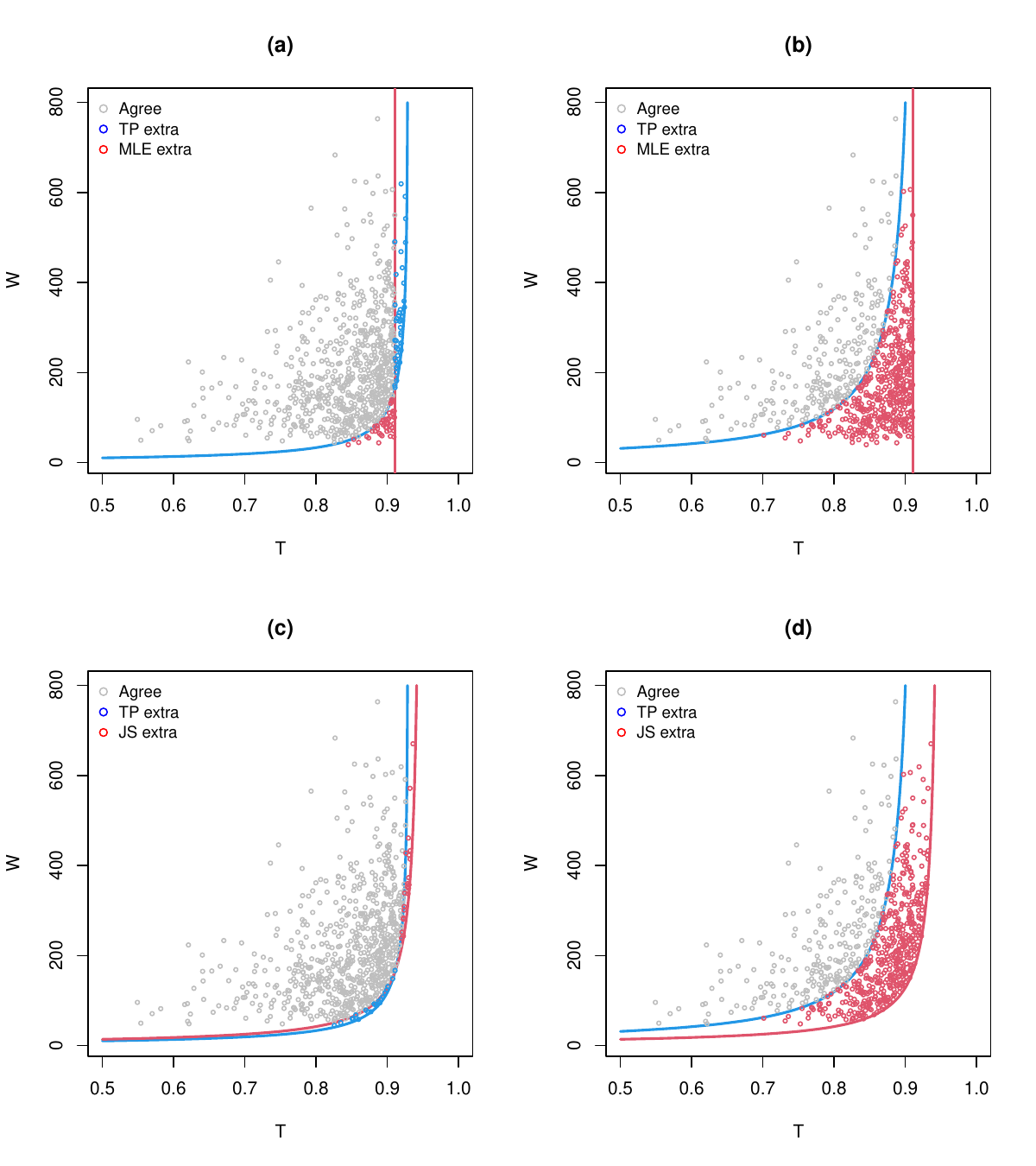}}
    \end{center}
    \caption{\small{Contrasting Selections for A-rated centers: 
    The two upper panels compare posterior tail probability selection
    with MLE (fixed effects) selection, while the lower panels compare
    TP selection with James-Stein (linear shrinkage) selection.  Left
    panels impose capacity control only, while the right panels impose
    20 percent FDR control for the TP rule.  The estimated FDR rate for both the MLE and James-Stein selection under capacity constraint, using the smoothed NPMLE estimator for $G$, is 0.431. Comparisons are based on the 2004-2006 data.}}
    \label{fig:conflict}
\end{figure}

Given the longtitudinal structure of the Dialysis data, it would be possible to 
consider the models in Section \ref{sec:UnknownVar} that allow for unobserved variance heterogeneity. 
We refrain from doing so partly due to space considerations and because we are reluctant to assume
stationarity of random effects over longer time horizons.

\subsection{Temporal Stability, Ranking and Selection}
Given the longitudinal nature of the data, it is natural to ask, ``How stable
are rankings over time, and isn't there some temporal dependence in the observed data
that should be accounted for?''  Perhaps surprisingly, the year-to-year dependence in the   
observed mortality is quite weak.  In Figure \ref{fig: ar1} we plot a histogram
of estimated AR(1) coefficients for the 3230 centers; it is roughly centered at zero
and slightly skewed to the left.  We do not draw the conclusion from this that there
is no temporal dependence in the observed $y_{it}$, but only that there is considerable
heterogeneity in the nature of this dependence with roughly as many centers exhibiting negative
serial dependence as those with positive dependence.   Our approach of considering brief,
3-5 year, windows of presumed stability in center performance is consistent with the
procedures of the official ranking agency.  In each of
these windows we can compute a ranking according to one of the criteria introduced above,
and it is of interest to see how much stability there is in these rankings.

\begin{figure}
    \begin{center}
    \resizebox{.6\textwidth}{!}{\includegraphics{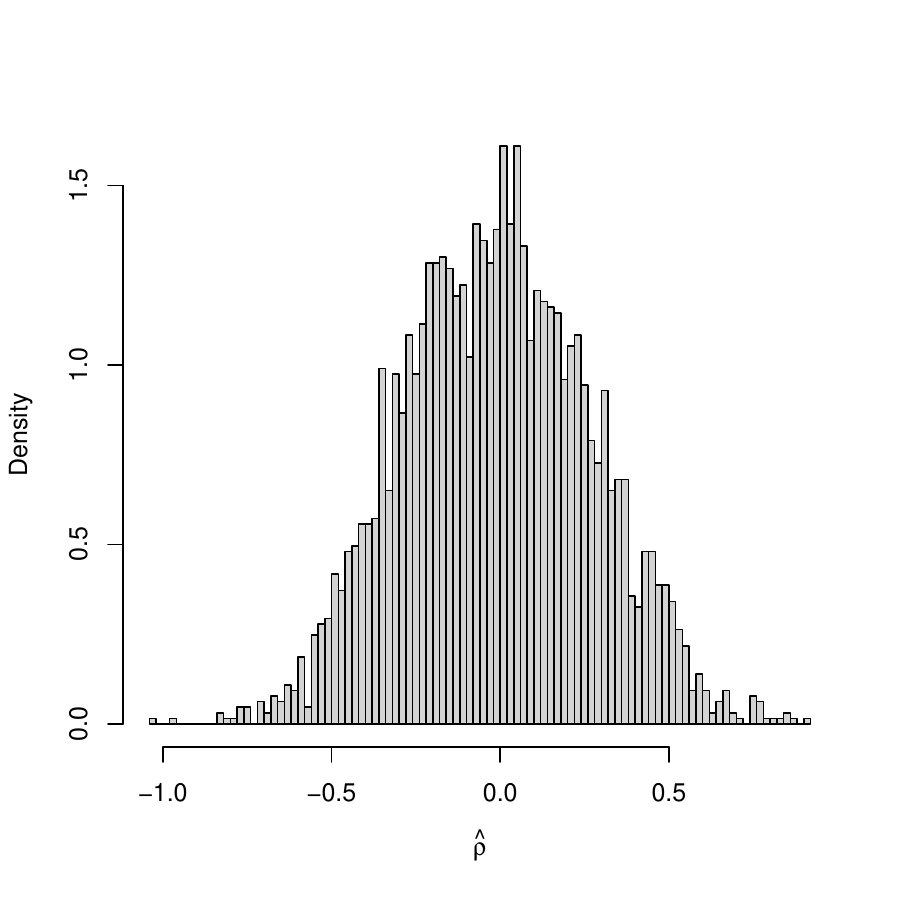}}
    \end{center}
    \caption{\small{Histogram of estimated AR(1) coefficients for 3230
    Dialysis centers based on annual data 2004-2017}}
    \label{fig: ar1}
\end{figure}

To address this question we consider rankings based on the posterior tail probability
criterion for three year windows.  In each of the 5 3-year windows we assign centers
letter grades, A-F, with proportions
$\{0.22, 0.30, 0.35, 0.09, 0.04 \}$ respectively. Table \ref{tab: TM} reports the estimated
transition matrix between these categories, so entry $i,j$ in the matrix represents the
estimated probability of a center in state $i$ moving to state $j$ in the next period.

\begin{table}[!tbp]
\begin{center}
\begin{tabular}{lrrrrr}
\hline\hline
\multicolumn{1}{l}{}&\multicolumn{1}{c}{A}&\multicolumn{1}{c}{B}&\multicolumn{1}{c}{C}&\multicolumn{1}{c}{D}&\multicolumn{1}{c}{F}\tabularnewline
\hline
A&$0.441$&$0.328$&$0.201$&$0.024$&$0.006$\tabularnewline
B&$0.247$&$0.360$&$0.327$&$0.059$&$0.007$\tabularnewline
C&$0.122$&$0.286$&$0.440$&$0.112$&$0.040$\tabularnewline
D&$0.062$&$0.181$&$0.441$&$0.210$&$0.106$\tabularnewline
F&$0.021$&$0.085$&$0.346$&$0.219$&$0.329$\tabularnewline
\hline
\end{tabular}
\caption{Estimated First Order Markov Transition Matrix: Entry
$i,j$ of the matrix estimates the probability of a transition
from state $i$ to state $j$ based on posterior tail probability
rankings for 3-year longitudinal grouping of the center data.\label{tab: TM}}\end{center}
\end{table}

It is obviously difficult to maintain an ``A'' rating for more than a couple of periods,
but centers with poor performance are also likely to move into the middle of the rankings.
Although, as we have seen, there is no guarantee that the posterior tail probability criterion
yields a nested ranking, nestedness does hold in this particular application.  Posterior mean
ranking yields similar transition behavior.  The high degree of mobility between rating
categories reenforces our conclusion that ranking and selection into rating categories is
subject to considerable uncertainty.

\section{Conclusions}
Robbins's compound decision framework is well suited to ranking and selection problems,
and nonparametric maximum likelihood estimation of mixture models offers a powerful tool
for implementing empirical Bayes rules for such problems.  Posterior tail probability 
selection rules perform better than posterior mean rules when precision is heterogeneous.
Ranking and selection is especially difficult in Gaussian settings where classical linear
shrinkage methods are most appropriate.  Nonparametric empirical Bayes methods can
substantially improve upon selection methods based on linear shrinkage and traditional
$p$-values when the latent mixing distribution is not Gaussian both in terms of power
and false discovery rate.

\appendix
\section{Proofs}
{\footnotesize{

\begin{proof}[Lemma \ref{lem:Nestedness}]
We can write:
	\begin{align*}
	\nabla_y \frac{\alpha f_1(y)}{f(y)} & = \frac{\int_{[\theta_\alpha, +\infty)} \nabla_y \log \varphi(y|\theta,\sigma^2) \varphi(y| \theta, \sigma^2) dG(\theta)}{\int _{(-\infty, +\infty)} \varphi(y| \theta, \sigma^2) dG(\theta)} \\
	& -  \frac{\int_{[\theta_\alpha, +\infty)}  \varphi(y| \theta,\sigma^2) dG(\theta)}{\int _{(-\infty, +\infty))} \varphi(y| \theta, \sigma^2) dG(\theta)} \frac{\int_{(-\infty, +\infty)} \nabla_y \log \varphi(y|\theta, \sigma^2) \varphi(y| \theta,\sigma^2) dG(\theta)}{\int _{(-\infty, +\infty)} \varphi(y| \theta, \sigma^2) dG(\theta)}\\
	& = \mathbb{E}\Big[ \11\{\theta \geq \theta_\alpha\} \nabla_y \log \varphi(y|\theta, \sigma^2)| Y\Big ] - \mathbb{E}\Big [\11\{\theta \geq \theta_\alpha\}|Y] \mathbb{E}[\nabla_y\log \varphi(y|\theta, \sigma^2)|Y\Big ]\\
	& = \text{Cov}\Big[ \11\{\theta \geq \theta_\alpha\} , \nabla_y \log \varphi(y|\theta,\sigma^2)  |Y\Big]\geq 0 
	\end{align*}
The last inequality holds since $\nabla_y\log \varphi(y|\theta,\sigma^2)$ is increasing in $\theta$ for each fixed $\sigma^2$, 
by the Gaussian assumption and the fact that covariance of monotone functions of $\theta$ is non-negative (see \citeasnoun{schmidt2014inequalities}) assuming the existence of $\mathbb{E}_{\theta|Y}[\nabla_y \log \varphi(Y|\theta,\sigma^2)|Y]$ which we assume.  Nesting follows from the monotonicity of the $v_\alpha(y)$ 
criterion: monotonicity of $v_\alpha(y)$ implies that there exists $t_{\alpha}$ such that 
$\11\{v_\alpha(y) \geq \lambda_\alpha/(1+\lambda_\alpha)\} = \11\{y \geq t_{\alpha}\}$, 
hence if $\alpha_1 > \alpha_2$ and $\mathbb{P}(y \geq t_{\alpha_1}) = 
\alpha_1$ and $\mathbb{P}(y \geq t_{\alpha_2})  =\alpha_2$, then it must be that 
$t_{\alpha_1} \leq t_{\alpha_2}$, implying nestedness.
\end{proof}

\vspace{3mm}
\begin{proof}[Lemma \ref{lem: mon1}]
    Denote the three decision criteria, $v_1 (y)  = \EE (\theta | Y = y)$,
	$v_2 (y)  = \PP (\theta \geq G^{-1} (1-\alpha)  | Y = y)$ and
	$v_3 (y)  = \EE (\theta 1(\theta \geq G^{-1} (1-\alpha))  | Y = y)$.
	Assuming that $\mathbb{E}[\theta|Y] < \infty$, $\mathbb{E}_{\theta|Y}[\nabla_y \log \varphi(y|\theta,\sigma^2)|Y] < \infty$ and $\mathbb{E}_{\theta|Y}[\theta \nabla_y \log \varphi(y|\theta,\sigma^2)|Y]<\infty$, the calculation leading to the proof of Lemma 1 shows:
    \begin{align*}
	\nabla_y v_1 (y) & = \text{Cov} (\theta , \nabla_y \log \varphi (y|\theta) | Y = y)\\
	\nabla_y v_2 (y) & = 
	    \text{Cov} (1(\theta \geq G^{-1} (1 - \alpha)) , \nabla_y \log \varphi (y|\theta) | Y = y)\\
	\nabla_y v_3 (y) & = 
	    \text{Cov} (\theta 1(\theta \geq G^{-1} (1 - \alpha)) , \nabla_y \log \varphi (y|\theta) | Y = y).
    \end{align*}
    Thus, the monotonicity of $\nabla_y \log \varphi (y|\theta,\sigma^2)$ implies they all yield identical rankings.
\end{proof}

\vspace{3mm}
\begin{proof}[Proposition \ref{prop: homo}]
The Bayes rule for the non-randomized selections can be characterized as, 
	\[
	\delta_i^* = \begin{cases}
	1  , & \text{if  } v_\alpha(y_i) \geq  \tau_1^* ( 1- v_\alpha(y_i) - \gamma) + \tau_2^*  \\
	0  , &\text{if }v_\alpha(y_i) < \tau_1^* ( 1- v_\alpha(y_i) - \gamma) + \tau_2^*
	\end{cases}
	\]
with Karush-Kuhn-Tucker conditions, 
\begin{eqnarray}
\tau_1^* \Big ( \mathbb{E}\Big[ \sum_{i=1}^n \Big \{(1- v_\alpha(y_i)) \delta_i^* - \gamma \delta_i^* \Big\}\Big] \Big) &= 0 \label{c2}\\ 
\tau_2^* \Big(\mathbb{E}\Big[ \sum_{i=1}^n \delta_i^*\Big] - \alpha n\Big)&= 0 \label{c3}\\
\mathbb{E}\Big[ \sum_{i=1}^n \Big \{(1- v_\alpha(y_i)) \delta_i^* - \gamma \delta_i^* \Big\}\Big] &\leq0 \label{c4}\\
\frac{\11}{n}\mathbb{E}\Big[ \sum_{i=1}^n \delta_i^* \Big] - \alpha & \leq 0 \label{c5}\\
\tau_1^* & \geq 0\\
\tau_2^* & \geq 0
\end{eqnarray}

The Bayes rule takes the form of thresholding on the posterior tail probability $v_\alpha(y)$ 
and since $v_\alpha(y)$ is monotone in $y$ as shown in Lemma \ref{lem: mon1}, 
it is therefore a thresholding rule on $Y$, $\delta_i^* = \11\{y_i \geq  t^*\}$ with 
cutoff $t^*$ depending on the values of $(\tau_1^*, \tau_2^*, \alpha, \gamma)$. 
Condition (\ref{c4}) is equivalent to the condition that the marginal false discovery rate, mFDR, since it requires 
\[
\mathbb{E}\Big [\sum_{i=1}^n \{(1-v_\alpha(y_i) \delta_i^*\}\Big] / \mathbb{E}\Big[ \sum_{i=1}^n \delta_i^*\Big] \leq \gamma 
\]
and we can show that the left hand side quantity is precisely the mFDR since 
\begin{align*}
\text{mFDR}(t^*) 
= &  \mathbb{P}(\delta_i^* = 1, \theta_i \leq \theta_\alpha)/\mathbb{P}(\delta_i^*=1)\\ 
= & (1-\alpha)  \int \11 \{ y \geq t^*\} f_0(y) dy / \int \11 \{y \geq  t^*\} f(y) dy\\
= & \int \11 \{y \geq  t^*\} (1 - v_\alpha(y) ) f(y) dy /\int \11\{y \geq  t^*\} f(y) dy\\
= & \mathbb{E}\Big[\sum_{i=1}^n \{(1-v_\alpha(y_i))\delta_i^*\}\Big]/\mathbb{E}\Big[\sum_{i=1}^n \delta_i^*\Big]\\ 
= & \frac{\int_{-\infty}^{\theta_\alpha} \tilde{\Phi}((t^* - \theta)/\sigma) dG(\theta)} 
{\int_{-\infty}^{+\infty} \tilde{\Phi}((t^* - \theta)/\sigma)dG(\theta)} \leq \gamma. 
\end{align*}
For any mixing distribution $G$, as $t^*$ increases, it becomes less likely for condition (\ref{c4}) to bind. 
And as $t^*$ approaches $-\infty$, left side of (\ref{c4})
approaches $1-\alpha- \gamma$, and hence we've restricted $\gamma < 1-\alpha$ to avoid cases where the 
condition (\ref{c4}) never binds.  On the other hand, condition (\ref{c5}) is equivalent to 
\[
\mathbb{P}(\delta_i^* = 1) - \alpha = \int_{-\infty}^{+\infty} \tilde{\Phi}((t^* - \theta)/\sigma)dG(\theta) - \alpha 
\leq 0.
\] 
As $t^*$ increases, it also becomes less likely that condition (\ref{c5}) binds. 
Therefore, we can define, 
	\begin{align*}
	t_1^* &= \min\{t: \text{mFDR}(t)- \gamma \leq0 \}\\
	t_2^* & = \min \{t:  \int_{-\infty}^{+\infty} \tilde{\Phi}((t - \theta)/\sigma)dG(\theta)-\alpha \leq 0 \}
	\end{align*}
When $t_1^* < t_2^*$, the feasible region for $Y$ defined by inequality (\ref{c5})  is a strict subset of that 
defined by inequality (\ref{c4}). When $t_1^* > t_2^*$, then the feasible region defined by inequality (\ref{c4}) 
is a strict subset of that defined by inequality (\ref{c5}). When $t_1^* = t_2^*$, the feasible regions coincide. 
This case occurs when $\text{mFDR}(t_1^*) = \gamma$ and $\mathbb{P}(y \geq t_2^*) = \alpha$, so
$\mathbb{E}[v_\alpha(Y) \11\{v_\alpha(Y) \geq \lambda^*\}] = 
\alpha - \alpha\gamma$ with $v_\alpha(t_1^*) = \lambda^*(\alpha,\gamma)$. 
Again, the strict thresholding enforced by the statement of the proposition can be relaxed slightly by randomizing
the selection probability of the last unit so that the active constraint is satisfied exactly.
\end{proof}

\vspace{3mm}
\begin{proof}[Proposition \ref{prop: rule_ysigma}]
In the proof we will suppress the dependence of $\lambda^*$ on the $(\alpha, \gamma)$. 
The argument is very similar to the proof for Proposition \ref{prop: homo}, except that now the feasible region 
defined by constraint (\ref{c4}) and (\ref{c5}) is a two-dimensional region for $(y_i, \sigma_i)$. 
Since the posterior tail probability $v_\alpha(y, \sigma)$ is monotone in $y$ for any fixed 
$\sigma$ as a result of Lemma \ref{lem: mon1}, the optimal rule can be reformulated as a 
thresholding rule on $Y$ again, $\delta_i^* = \11\{y_i > t_{\alpha}(\lambda^*,\sigma_i)\}$ 
except now the threshold value also depends on $\sigma_i$. 
	
Now consider the constraint (\ref{c4}) and (\ref{c5}). Condition (\ref{c4}) is equivalent to the condition 
that, 
\[
\frac{\mathbb{E}\Big[ \sum_{i=1}^n \{(1-v_\alpha(y_i,\sigma_i))\delta_i^*\}\Big]} 
{\mathbb{E}\Big[ \sum_{i=1}^n \delta_i^*\Big]} - \gamma = 
\frac{\int \int_{-\infty}^{\theta_\alpha}\tilde \Phi((t_\alpha(\lambda^*,\sigma) - \theta)/ \sigma)dG(\theta)dH(\sigma)}
{\int \int_{-\infty}^{+\infty }\tilde \Phi((t_\alpha(\lambda^*,\sigma) - \theta)/\sigma)dG(\theta)dH(\sigma)}-\gamma \leq 0. 
\]
For any marginal distributions $G$ and $H$ of $(\theta,\sigma)$ and for a fixed pair of $(\alpha,\gamma)$, as $\lambda^*$ increases, 
$t_\alpha(\lambda^*,\sigma)$ also increases for any $\sigma > 0$ and therefore it is less likely for condition (\ref{c4}) to bind. 
On the other hand, condition (\ref{c5}) is equivalent to, 
\[ 
\mathbb{P}(\delta_i^* = 1) - \alpha = \int \int \tilde \Phi((t_\alpha(\lambda^*,\sigma) - \theta)/\sigma) dG(\theta) dH(\sigma) - 
\alpha \leq 0. 
\]
So as $\lambda^*$ increases, it is also less likely for condition (\ref{c5}) to bind. 
Thus, when $\lambda_1^* < \lambda_2^*$, the feasible region on $(Y, \sigma)$ defined by inequality constraint (\ref{c5}) 
is a strict subset of that defined by inequality (\ref{c4}). 
When $\lambda_1^* > \lambda_2^*$, then the feasible region defined by inequality (\ref{c4}) is a strict subset 
of that defined by inequality (\ref{c5}). When $\lambda_1^* = \lambda_2^*$, the feasible regions coincide;
this case occurs when 
\[
\mathbb{E}\Big[ v_\alpha(Y, \sigma) \11\{v_\alpha(Y, \sigma) \geq \lambda^*\}\Big] = \alpha - \alpha \gamma
\]
where the expectation is taken with respect to the joint distribution of $(Y, \sigma)$. 

Finally regarding the existence of $\lambda^*$, note that existence of  a solution for $\lambda_2^*$ 
for any $\alpha \in (0,1)$ follows from the fact that for
any fixed $\alpha$, $f_2(\alpha, \lambda) = \mathbb{P}(v_\alpha(y,\sigma) > \lambda) - \alpha$, 
is a decreasing function in $\lambda$ and $\lambda_2^*(\alpha)$ is defined as the zero-crossing point 
of  $f_2(\alpha, \lambda)$. Note that $f_2(\alpha,0) = 1-\alpha$ and $f_2(\alpha,1) = -\alpha$. 
Therefore for any $\alpha\in (0,1)$, we can always find a $\lambda_2^*(\alpha) \in (0,1)$ 
such that $f_2(\alpha, \lambda_2^*(\alpha)) = 0$.  Now consider $f_1(\alpha,\gamma,\lambda) = 
\mathbb{E}[(1-v_\alpha(y,\sigma) - \gamma) \11\{v_\alpha(y,\sigma) > \lambda\}]$. 
For a fixed pair of $(\alpha,\gamma)$, $\lambda_1^*(\alpha,\gamma)$ is defined as the zero crossing 
point of $f_1(\alpha,\gamma,\lambda)$. Note that $f_1(\alpha,\gamma,\lambda)$ decreases first and 
then increases in $\lambda$ with its minimum  achieved at $\lambda = 1-\gamma$. We also know 
that $f_1(\alpha,\gamma,0) = 1-\gamma - \mathbb{E}[v_\alpha(y,\sigma)] = 1- \gamma - \alpha$ and 
$f_1(\alpha,\gamma,1) = 0$. Hence as long as $\gamma < 1-\alpha$, the zero-crossing 
$\lambda_1^*(\alpha,\gamma)$ exists. The condition $\gamma < 1-\alpha$ is imposed to rule out 
cases where FDR constraint never binds.
\end{proof}

\vspace{3mm}
\begin{proof}[Lemma \ref{lem: nest1}]
	Note that for any cutoff value $\lambda$, the mFDR can be expressed as, 
	\[
	\text{mFDR}(\alpha,\lambda) = \mathbb{E}[(1-v_{\alpha}(y_i,\sigma_i)) 
	\11\{v_{\alpha}(y_i,\sigma_i)\geq  \lambda\}]/\mathbb{P}[v_{\alpha}(y_i,\sigma_i) \geq  \lambda].
	\]
	Thus, mFDR depends both on the cutoff value $\lambda$ and on $\alpha$ since $v_\alpha$, 
	is a function of $\alpha$, and consequently its density function is also indexed by $\alpha$. 
	
	First, we will show that $\nabla_{\lambda} \text{ mFDR}(\alpha,\lambda ) \leq  0$ for all $\alpha \in (0,1)$. 
	Differentiating with respect to $\lambda$ gives, 
	\begin{align*}
		\nabla_\lambda \frac{\int_{\lambda}^1  (1-v) f_{ v}(v; \alpha ) dv}{\int_{\lambda}^{1} f_{ v}(v; \alpha)dv}& = \frac{-(1-\lambda) f_{ v}(\lambda; \alpha) \int _\lambda^1 f_{ v}(v; \alpha) dv + \int_\lambda^1 (1-v)f_{ v}(v;\alpha )dv f_{ v}(\lambda; \alpha )}{(\int_\lambda^1 f_{ v}(v; \alpha )dv)^2}\\
		& = \frac{f_{ v}(\lambda; \alpha)}{(\int_\lambda^1 f_{ v}(v; \alpha)dv)^2} \Big ( \int_\lambda^1 (1-v) f_{ v}(v; \alpha)dv - \int_\lambda^1 (1-\lambda) f_{ v}(v; \alpha ) dv\Big) \\
		& \leq  0.
	\end{align*}
	Next, to establish that $\nabla_\alpha \text{ mFDR}(\alpha,\lambda) \leq 0$, 
	we differentiate with respect to $\alpha$, to obtain,
	\begin{align*}
		&\nabla_{\alpha}  \frac{\int_\lambda^1  (1-v) f_{ v}(v; \alpha ) dv}{\int_\lambda^1 f_{ v}(v; \alpha )dv}\\
		&  =  \frac{\int_\lambda^1 (1-v) \nabla_{\alpha} \log f_{ v}(v;\alpha ) f_{ v}(v; \alpha )dv}{\int _\lambda^1 f_{ v}(v; \alpha ) dv} - \frac{\int_\lambda^1 (1-v) f_{ v}(v; \alpha ) dv}{\int_\lambda^1 f_{ v}(v; \alpha )dv} \frac{\int_\lambda^1 \nabla_{\alpha} \log f_{ v}(v;\alpha ) f_{ v}(v; \alpha ) dv}{\int_\lambda^1 f_{ v}(v; \alpha )dv}\\
		& = \frac{\mathbb{E}[(1- v) \nabla_{\alpha} \log f_{ v}( v; \alpha ) \11\{ v \geq  \lambda\}]}{\mathbb{P}( v \geq \lambda )} - \frac{\mathbb{E}[ (1-v) \11\{ v\geq \lambda \}]}{\mathbb{P}( v \geq  \lambda)} \frac{\mathbb{E}[\nabla_{\alpha} \log f_{ v}( v; \alpha )\11\{ v \geq \lambda\}]}{\mathbb{P}( v\geq \lambda )}\\
		& = \mathbb{E}[(1- v) \nabla_{\alpha} \log f_{ v}( v;\alpha ) |  v \geq \lambda] - \mathbb{E}[1-v |  v\geq \lambda ] \mathbb{E}[\nabla_{\alpha}\log f_{ v}( v; \alpha )| v\geq \lambda ]\\
		& = \text{cov}[1- v, \nabla_{\alpha} \log f_{ v} ( v; \alpha )| v \geq \lambda ] \leq 0.
	\end{align*}
	where the last inequality holds because $\nabla_{\alpha}\log f_{ v}(v;\alpha )$ is non-decreasing in $v$. 
	
	Now suppose we have the cutoff value $\lambda_1^*(\alpha_2,\gamma)$ such that, 
	\[
	\mathbb{E}[(1-v_{\alpha_2}(y_i,\sigma_i)) \11\{v_{\alpha_2}(y_i,\sigma_i) > 
	\lambda_1^*(\alpha_2,\gamma)\}]/\mathbb{P}(v_{\alpha_2}(y_i,\sigma_i) > \lambda_1^*(\alpha_2,\gamma)]=\gamma . 
	\]
	If we maintain the same cutoff value for $v_{\alpha_1}(y_i,\sigma_i)$ with $\alpha_1 > \alpha_2$, 
	given the second property of $\text{mFDR}$, we know, 
	\[
	\mathbb{E}[(1-v_{\alpha_1}(y_i,\sigma_i))\11\{v_{\alpha_1}(y_i,\sigma_i) \geq  \lambda_1^*(\alpha_2,\gamma)\}]/\mathbb{P}(v_{\alpha_1}(y_i,\sigma_i)\geq  \lambda_1^*(\alpha_2,\gamma)) \leq \gamma. 
	\]
	If equality holds, then by definition we have $\lambda_1^*(\alpha_2, \gamma) = \lambda_1^*(\alpha_1,\gamma)$, 
	if strict inequality holds, then by the first property of $\text{mFDR}$, in order to increase 
	$\text{mFDR}$ level to be equal to $\gamma$, we must have $\lambda_1^*(\alpha_1, \gamma) < \lambda_1^*(\alpha_2,\gamma)$. 
\end{proof}

\vspace{3mm}
\begin{proof}[Corollary \ref{cor: nest1}]
	For any $\alpha_1 > \alpha_2$, we have $v_{\alpha_1}(y,\sigma) \geq  v_{\alpha_2}(y,\sigma)$ for all pair of $(y, \sigma) \in \mathbb{R} \times \mathbb{R}_{+}$. When the condition in Lemma \ref{lem: nest1} holds, then $v_{\alpha_1}(y,\sigma) \geq  v_{\alpha_2}(y,\sigma) > \lambda_1^*(\alpha_2, \gamma) \geq  \lambda_1^*(\alpha_1,\gamma)$, which implies $\Omega_{\alpha_2 ,\gamma}^{FDR} \subseteq \Omega_{\alpha_1,\gamma}^{FDR}$. 
\end{proof}

\vspace{3mm}
\begin{proof}[Lemma \ref{lem: nest2}]
	Since $t_{\alpha}(\lambda^*_2(\alpha),\sigma)$ defines the boundary of the selection region under the capacity constraint for a fixed level $\alpha$. The condition imposed that as $\alpha$ increases, for each fixed $\sigma$, the thresholding value for $Y$ decreases, hence nestedness of the selection region. 
\end{proof}

\vspace{3mm}
\begin{proof}[Lemma \ref{lem: nest3}]
	Based on results in Lemma \ref{lem: nest1} and Lemma \ref{lem: nest2} and the fact that $\Omega_{\alpha,\gamma} = \Omega_{\alpha,\gamma}^{FDR} \cap \Omega_{\alpha}^C$, we have nestedness of the selection set. 
\end{proof}

\vspace{3mm}
\begin{proof}[Proposition \ref{prop: PMnest}]
	The capacity constraint requires that 
	\[
	\alpha = \mathbb{P}(M(y, \sigma) \geq  C_2^*(\alpha)) = \int \int \11\{M(y, \sigma) \geq  C_2^*(\alpha)\} f(y |\theta,\sigma) dG(\theta) dH(\sigma)
	\]
	For any $\alpha_1 > \alpha_2$, it is then clear that $C_2^*(\alpha_1) \leq C_2^*(\alpha_2)$. Given the monotonity of $M(y, \sigma)$ for each fixed $\sigma$ established in Lemma \ref{lem: mon1}, the selection set based on capacity constraint is nested. 
	For FDR constraint, we require 
	\begin{equation}
	\gamma = \frac{\int \int_{-\infty}^{\theta_\alpha}\11\{M(y,\sigma) \geq  C_1^*(\alpha,\gamma)\} f(y|\theta,\sigma) dG(\theta) dH(\sigma)}{\int \int \11\{M(y,\sigma) \geq  C_1^*(\alpha,\gamma)\}f(y|\theta,\sigma) dG(\theta)dH(\sigma)} \label{eq: PMFDR}
	\end{equation}
Fix $\gamma$, it suffices to show that if $\alpha_1 > \alpha_2$, then $C_1^*(\alpha_1, \gamma) \leq C_1^*(\alpha_2, \gamma)$. First, solve for $C_1^*(\alpha_2, \gamma)$ from equation (\ref{eq: PMFDR}). Now suppose we use this same thresholding value when we increase capacity to $\alpha_1 > \alpha_2$, we evaluate the right hand side of equation (\ref{eq: PMFDR}). Since $\theta_{\alpha_1}\leq \theta_{\alpha_2}$, then the numerator decreases, 
\begin{align*}
&\int \int_{-\infty}^{\theta_{\alpha_1}}\11\{M(y,\sigma) \geq  C_1^*(\alpha_2,\gamma)\} f(y|\theta,\sigma) dG(\theta) dH(\sigma)\\
&\leq \int \int_{-\infty}^{\theta_{\alpha_2}}\11\{M(y,\sigma) \geq  C_1^*(\alpha_2,\gamma)\} f(y|\theta,\sigma) dG(\theta) dH(\sigma),
\end{align*}
while the denominator does not change. The only way to satisfy the equality (\ref{eq: PMFDR}) again is to decrease the thresholding value, therefore $C_1^*(\alpha_1, \gamma) \leq C_1^*(\alpha_2, \gamma)$. The result in the Proposition is then reached since the selection set is the intersection of the selection set under capacity constraint and that under FDR constraint. 
\end{proof}

\vspace{3mm}

\begin{proof}[Lemma \ref{lem: nonmon}]
The logarithm of the Gamma density of $S_i$ takes the form, 
	\[
	\log \Gamma(S_i |r_i, \sigma_i^2) =  r_i \log (r_i /\sigma_i^2) - \log (\Gamma(r_i)) + (r_i - 1) \log S_i - S_i \frac{r_i}{\sigma_i^2},
	\]
hence 
$\nabla_s \Gamma(s|r, \sigma^2) = \Gamma(s|r, \sigma^2) \Big( \frac{r-1}{s} - \frac{r}{\sigma^2}\Big).$  
Fixing $y$ and differentiating with respect to $s$, we have, 
\begin{align*}
\nabla_s v_\alpha & (y,s) = \frac{\int \int _{\theta_\alpha}^{+\infty} f(y|\theta,\sigma^2) \Gamma(s|r,\sigma^2) \Big[ \frac{r-1}{s}-\frac{r}{\sigma^2}\Big] dG(\theta,\sigma^2)}{\int \int f(y|\theta,\sigma^2)\Gamma(s|r,\sigma^2) dG(\theta,\sigma^2)}\\
		& - \frac{\int\int_{\theta_\alpha}^{+\infty} f(y|\theta,\sigma^2) \Gamma(s|r,\sigma^2) dG(\theta,\sigma^2) }{\int \int f(y|\theta,\sigma^2)\Gamma(s|r,\sigma^2)dG(\theta,\sigma^2)} \frac{\int \int f(y|\theta,\sigma^2) \Gamma(s|r,\sigma^2) \Big[ \frac{r-1}{s}-\frac{r}{\sigma^2}\Big] dG(\theta,\sigma^2)}{\int \int f(y|\theta,\sigma^2)\Gamma(s|r,\sigma^2) dG(\theta,\sigma^2)}\\
& = \mathbb{E}\Big [\11\{\theta\geq \theta_\alpha\} \Big(\frac{r-1}{s}-\frac{r}{\sigma^2} \Big) |Y = y,S = s\Big] - \mathbb{E}\Big[ \11\{\theta\geq \theta_\alpha\} |Y = y, S = s\Big] \mathbb{E}\Big[ \frac{r-1}{s}-\frac{r}{\sigma^2}|Y = y,S = s\Big] \\
& = - \text{Cov}\Big[ \11\{\theta\geq \theta_\alpha\}, \frac{r}{\sigma^2}|Y = y, S = s\Big] 
\end{align*}
The covariance term can take either sign since we do not restrict the distribution $G$, 
so $v_\alpha(Y, S)$ need not be monotone in $S$.	
On the other hand, if we fix $s$ and differentiate with respect to $y$, 
\begin{align*}
\nabla_y v_\alpha(y, s) & = \mathbb{E}\Big[ \11\{\theta \geq \theta_\alpha\} \Big[ -\frac{y - \theta}{\sigma^2/T}\Big]|Y = y,S = s\Big] - \mathbb{E}\Big[ \11\{\theta \geq \theta_\alpha\}|Y=y,S =s\Big] \mathbb{E}\Big[ - \frac{y-\theta}{\sigma^2/T}|Y=y,S=s\Big] \\
& = \text{Cov}\Big[ \11\{\theta\geq \theta_\alpha\}, \frac{\theta - y}{\sigma^2/T}|Y=y,S= s\Big]. 
\end{align*}
Again, the covariance term can take either sign, depending on the correlation of $\theta$ and $\sigma^2$ 
conditional on $(Y,S)$. Therefore, fixing $S$, $v_\alpha(Y,S)$ need not be a monotone function of $Y$. 
\end{proof}

\vspace{3mm}
\begin{proof}[Proposition \ref{prop: nonmon}]
The proof is very similar to that of Proposition \ref{prop: rule_ysigma}, 
the only difference is that we can no longer formulate the decision rule by
simply thresholding on $Y$ because the transformation $v_\alpha(Y,S)$ need not be monotone in 
$Y$ for fixed values of $S$ as shown in Lemma \ref{lem: nonmon}, hence $\lambda_1^*(\alpha,\gamma)$ and 
$\lambda_2^*(\alpha)$ must now be defined directly through the random variable $v_\alpha$. The first constraint states that 
\[
\mathbb{E}\Big[ \sum_{i=1}^n \{(1-v_\alpha(y_i,s_i)) \delta_i^* \Big] / \mathbb{E}\Big[ \sum_{i=1}^n \delta_i^*\Big] \leq \gamma 
\]
with $\delta_i^* = 1\{ v_\alpha(y_i, s_i) \geq \lambda^*\}$
For each fixed $\alpha$, let the density function for $v_\alpha(Y,S)$ be denoted as $f_v(\cdot; \alpha) $, then the constraints can be formulated as 
\[
\int_{\lambda^*}^1 (1- v) f_v(v; \alpha) dv / \int _{\lambda^*}^{1} f_v(v; \alpha) dv 
\]
which is non-increasing in $\lambda^*$, hence the constraint becomes less likely to bind as $\lambda^*$ increases. On the other hand the second constraint states that, 
\[
\mathbb{P}(\delta_i^* = 1) - \alpha = \int_{\lambda^*}^{1} f_v(v; \alpha) dv - \alpha 
\] 
For each fixed $\alpha \in (0,1)$, this constraint also becomes less likely to bind as $\lambda^*$ increases. 
\end{proof} 

\vspace{3mm}

\begin{proof} [Theorem \ref{thm:oracle}] 
To prove Theorem \ref{thm:oracle}, we first introduce some additional notation and prove several
lemmas. Let 
	\begin{align*}
	H_{n,0}(t) &= 1 - H_n(t) = \frac{1}{n} \sum_{i=1}^n 1\{v_{\alpha,i} \geq  t\}\\
	H_{n,1}(t) & = \frac{1}{n} \sum_{i=1}^n (1-v_{\alpha,i}) 1\{v_{\alpha, i} \geq t\}\\
	Q_n(t) & = H_{n,1}(t)/H_{n,0}(t) \\
V_n(t) &=\frac{1}{n} \sum_{i=1}^n 1\{v_{\alpha,i} \geq t \} 1\{\theta_i \leq \theta_\alpha\}\\
	H_0(t) & = 1- H(t) = \mathbb{P}(v_{\alpha,i} \geq t)\\
	H_1(t) & = \mathbb{E}\Big[ (1- v_{\alpha,i})1\{v_{\alpha,i} \geq t\}\Big] \\
Q(t) & = H_1(t)/H_0(t)
\end{align*}

\begin{lemma} \label{lemma: A1}
	Under Assumption \ref{ass1}, as $n \to \infty$, 
	\begin{align*}
		&\underset{t \in [0,1]}{\sup} | H_{n,0}(t) - H_0(t)| \overset{p}{\to} 0\\
		& \underset{t \in [0,1]}{\sup} |H_{n,1}(t) - H_1(t)| \overset{p}{\to} 0
		\end{align*}
	\end{lemma} 

\begin{proof}[Proof of Lemma \ref{lemma: A1}] 
Under Assumption \ref{ass1} and the fact that $v_{\alpha,i} \in [0,1]$, the weak law of large numbers implies that we have for any $t \in [0,1]$, as $n \to \infty$, 
\begin{align*}
&H_{n,0}(t) \overset{p}{\to} H_0(t)\\
&H_{n,1}(t) \overset{p}{\to} H_1(t)
\end{align*}
By the Glivenko-Cantalli theorem, the first result is immediate. To prove the second result, it suffices to 
show that for any $\epsilon>0$, as $n \to \infty$, 
\[
\mathbb{P}\Big( \underset{t\in [0,1]}{\sup} |H_{n,1}(t) - H_1(t)| > \epsilon\Big) \to 0.
\]
It is clear that since $v_{\alpha,i}$ has a continuous distribution, $H_1(t)$ is a monotonically decreasing and bounded function in $t$ with $H_1(0) = 1- \alpha$ and $H_1(1) = 0$. It is also clear that the function $H_{n,1}(t)$ is monotonically decreasing in $t$, so we can find $m_{\epsilon} < \infty$ points such that 
$0 = t_{0} < t_1< \dots < t_{m_{\epsilon}} = 1$,
and for any $j \in \{1,2,\dots, m_\epsilon\}$, we have
$H_1(x_j) - H_1(x_{j-1}) \leq \epsilon/2$.
For any $t \in [0,1]$, there exists a $j$ such that $x_{j-1} \leq t \leq x_j$ and 
\begin{align*}
H_{n,1}(t) - H_1(t) &\leq H_{n,1}(t_{j-1}) - H_1(t_{j})\\
& = (H_{n,1}(t_{j-1}) - H_1(t_{j-1})) + (H_1(t_{j-1}) - H_1(t_j))\\
& \leq | H_{n,1}(x_{j-1}) - H_1(x_{j-1})| + \epsilon/2 \leq \underset{j}{\max} | H_{n,1}(t_{j-1}) - H_1(t_{j-1})| + \epsilon/2
\end{align*}
Likewise we can show that $H_{n,1}(t) - H_1(t) \geq - \underset{j}{\max}|H_{n,1}(t_j) - H_1(t_j)| - \epsilon/2$, hence 
\[
\underset{t \in [0,1]}{\sup} |H_{n,1}(t) - H_1(t)| \leq  \underset{j}{\max}|H_{n,1}(t_j) - H_1(t_j)|  + \epsilon/2. 
\]
Since $m_\epsilon$ is finite then for any $\delta > 0$, there exists $N$ such that for all $n \geq N$, 
\[
\mathbb{P}\Big( \underset{j}{\max} | H_{n,1}(t_j) - H_1(t_j)| \geq \epsilon/2\Big) \leq \delta
\]
which then implies that 
\begin{align*}
\mathbb{P}\Big( \underset{t \in [0,1]}{\sup} |H_{n,1}(t) - H_1(t)| \geq \epsilon\Big)  & \leq \mathbb{P}\Big(\frac{\epsilon}{2} + \underset{j}{\max} | H_{n,1}(t_j) - H_1(t_j)| \geq \epsilon\Big) \\
& = \mathbb{P}\Big( \underset{j}{\max} | H_{n,1}(t_j) - H_1(t_j)| \geq \epsilon/2\Big)  \to 0
\end{align*}
\end{proof} 



\begin{lemma} \label{lemma: A2}
	Under Assumption \ref{ass1} and $\alpha < 1- \gamma$, $Q(1-\gamma) < \gamma$. 
\end{lemma}
\begin{proof}[Proof of Lemma \ref{lemma: A2}]
	Define $\bar Q(t) = \mathbb{E}\Big[ (1- v_{\alpha,i} - \gamma) 1\{v_{\alpha,i} \geq t\}\Big]$, then $Q(t) = \gamma$ implies $\bar Q(t) = 0$. Since $Q(t)$ is monotonically decreasing in $t$ as shown in the proof of Lemma \ref{lem: nest1} , it suffices to prove that $\bar Q(1-\gamma) < 0$. To this end, note that $\nabla_t \bar Q(t) < 0$ for $t < 1- \gamma$ and $\nabla_t \bar Q(t) >  0$ for $t > 1- \gamma$, hence $\bar Q(t)$ obtains its minimun value at $t = 1- \gamma$. Note that $\bar Q(0) = 1- \gamma - \alpha$ and $\bar Q(1) = 0$, thus $\bar Q(1- \gamma) < 0$. 
\end{proof} 

\bigskip 
\noindent Now we are ready to prove Theorem \ref{thm:oracle}. We prove the first statement, as the second statement can be shown with a similar argument. First we show that, $\underset{t \leq 1- \gamma}{\sup} \Big| Q_n(t) - Q(t)\Big|  \overset{p}{\to} 0$, since, 
\begin{align*}
\Big| Q_n(t) - Q(t)\Big| & = \Big|  \frac{H_0(t) H_{n,1}(t) - H_1(t) H_{n,0}(t)}{H_{n,0}(t) H_0(t)}  \Big | = \Big| \frac{H_0(t) (H_{n,1}(t) - H_1(t)) - H_1(t) (H_{n,0}(t) - H_0(t))}{H_{n,0}(t) H_0(t)}\Big| \\
& \leq \frac{ H_0(0) \underset{t}{\sup} |H_{n,1}(t) - H_1(t)| + H_1(0) \underset{t}{\sup} | H_{n,0}(t) - H_0(t)|}{H_0(1-\gamma) \Big( H_0(1-\gamma) - \underset{t}{\sup} |H_{n,0}(t) - H_0(t)| \Big)} \overset{p}{\to} 0
\end{align*}
uniformly for any $t \leq 1- \gamma$. The last inequality holds because $\underset{t\leq 1- \gamma}{\min} H_0(t) = H_0(1-\gamma)$ by monotonicity of $H_0(t)$. With a similar argument, we can also show that $\underset{t \leq 1- \gamma}{\sup} \Big| \frac{V_n(t)}{H_{n,0}(t)} - Q(t)\Big| \overset{p}{\to} 0$. Using this result and the fact that $Q(1-\gamma) < \gamma$ by Lemma \ref{lemma: A2}, we have $\mathbb{P}\Big (|Q_n(1- \gamma) - Q(1-\gamma)| < \frac{\gamma - Q(1-\gamma)}{2}\Big ) \to 1$ and therefore, $\mathbb{P}(Q_n(1-\gamma) < \gamma ) \to 1$ and $\mathbb{P}(\lambda_{2n} \leq  1- \gamma) \to 1$ by the definition of $\lambda_{2n}$. Since $\lambda_n \leq \lambda_{2n}$ by definition, we also have $\mathbb{P}(\lambda_n \leq 1- \gamma) \to 1$. On the other hand, 
\[
Q_n(\lambda_{n}) - \frac{V_n(\lambda_{n})}{H_{n,0}(\lambda_{n})}  \geq \underset{t \leq  1- \gamma}{\inf}  \Big( Q_n(t) - Q(t) + Q(t) - V_n(t)/H_{n,0}(t)\Big) = o_p(1)
\]
Since $Q_n(\lambda_{n}) \leq Q_n(\lambda_{2n})  \leq \gamma$, it follows that 
\[
\frac{V_n(\lambda_{n}) }{H_{n,0}(\lambda_{n}) \bigvee 1} \leq \frac{V_n(\lambda_{n}) }{H_{n,0}(\lambda_n)} \leq \gamma + o_p(1).
\]
Since $\frac{V_n(\lambda_{n}) }{H_{n,0}(\lambda_{n}) \bigvee 1}$ is upper bounded by 1, by Fatou's lemma, we have 
\[
\underset{n \to \infty}{\limsup} \mathbb{E}\Big[ \frac{V_n(\lambda_{n}) }{H_{n,0}(\lambda_{n}) \bigvee 1}\Big]  \leq \gamma 
\]

\end{proof} 

\vspace{3mm} 

\begin{lemma} \label{lemma: A3} 
	Under Assumption \ref{ass1} and \ref{ass2}, as $n \to \infty$, 
	\[
	\hat \theta_\alpha \to \theta_\alpha \quad a.s.
	\]
\end{lemma} 
\begin{proof}[Proof of Lemma \ref{lemma: A3}] 
	See Lemma 21.2 in \citeasnoun{vdv}. 
\end{proof}
	
\begin{lemma} \label{lemma: A4} 
	Under Assumption \ref{ass1} and \ref{ass2}, as $n \to \infty$, 
	\[
	\sup_{i } |\hat v_{\alpha,i} - v_{\alpha,i}| \to 0 \quad a.s.
	\]
\end{lemma}
\begin{proof}[Proof of Lemma \ref{lemma: A3}]
Since 
\[
\hat v_{\alpha,i} = \frac{\int _{\hat \theta_\alpha}^{+\infty} f(D_i| \theta) d\hat G_n(\theta)}{\int_{-\infty} ^{+\infty} f(D_i|\theta) d\hat G_n(\theta)} 
\]
where we denote $D_i$ as data with a density function $f(D_i|\theta)$. When variances are known, then $D_i = \{y_i, \sigma_i\}$ and $f(D_i|\theta) = \frac{1}{\sigma_i} \varphi((y_i - \theta)/\sigma_i)$ and when variances are unknown, then $D_i = \{y_i, s_i\}$ and $f(D_i|\theta) =   \frac{1}{\sqrt{\sigma^2/T}} \varphi((y_i - \theta)/\sqrt{\sigma^2/T}) \Gamma(s_i |r, \sigma^2/r)$ with $r = (T-1)/2$ with $\varphi(\cdot)$ and $\Gamma(\cdot|\cdot, \cdot)$ being the standard normal and gamma density function, respectively.  

We first analyze the denominator and prove 
\begin{equation} \label{eq: denom}
\underset{x}{\sup } \left | \int_{-\infty}^{+\infty} f(x|\theta) d\hat G_n(\theta) - \int_{-\infty}^{+\infty}  f(x|\theta) dG(\theta) \right |\to 0 \quad a.s.
\end{equation} 
Let $f_n (x) = \int f(x | \theta) d\hat G_n(\theta)$ and $f(x) = \int f(x|\theta) dG(\theta)$. Under Assumption \ref{ass2}, we have $\frac{1}{2} \int \Big( \sqrt{f_n(x)} - \sqrt{f(x)}\Big)^2 d\mu(x) \to 0$ almost surely, which implies that $\int |f_n(x) - f(x)| dx \to 0$ almost surely. If $f_n(x)$ and $f(x)$ are Lipschitz continuous, we proceed by contradiction. Suppose \eqref{eq: denom} doesn't hold, then there exists $\epsilon > 0$ and a sequence $\{x_n\}_{n \geq 1}$ such that $|f_n(x_n) - f(x_n)| \geq \epsilon$ for all $n$. By Lipschitz continuity of $f_n$ and $f$, there exists $C$ such that 
\begin{align*}
	&|f_n(x_n + \delta) - f_n(x_n)| \leq C \| \delta\|\\
	&|f(x_n + \delta) - f(x_n)| \leq C \|\delta\|
	\end{align*} 
And therefore there exists $\eta > 0$ and for all $\|y - x_n\| \leq \eta$, $|f_n(y) - f(y)| \geq \epsilon/2$, which then implies 
\[
\int |f_n(x) - f(x)| dx \geq \int 1\{\|y - x_n\| \leq \eta\} |f_n(y) - f(y)| dy \geq \frac{\epsilon}{2} \int  1\{\|y - x_n\| \leq \eta\} dy 
\]
which contradicts $\int |f_n(x) - f(x) | dx \to 0$ almost surely. 
To prove that the functions $f_n$ and $f$ are Lipschitz continuous. Note that it suffices to prove that for each fixed parameter $\theta$, $|f(x|\theta) - f(y|\theta)| \leq C_\theta \|x - y\|$ and $\sup_\theta C_\theta < \infty$. This clearly holds for the Gaussian density since the Gaussian density is everywhere differentiable and has bounded first derivative. Under Assumption \ref{ass1} with $T \geq 4$, the Gamma density is also everywhere differentiable and has bounded first derivative, and thus is Lipschitz continuous. 

We next analyze the numerator and show 
\begin{equation}\label{numer}
\underset{x}{\sup} \left | \int_{\hat \theta_\alpha}^{+\infty} f(x|\theta) d\hat G_n(\theta) - \int _{\theta_\alpha}^{+\infty} f(x|\theta) dG(\theta) \right | \to 0 \quad a.s.
\end{equation} 
Note that
\begin{align*}
& \left 	|\int_{\hat \theta_\alpha} ^{+\infty} f(x |\theta) d\hat G_n(\theta) - \int_{\theta_\alpha}^{+\infty} f(x|\theta) dG(\theta) \right | \\
& \leq \left| \int _{\hat \theta_\alpha}^{+\infty} f(x|\theta) d\hat G_n(\theta) - \int_{\theta_\alpha}^{+\infty} f(x|\theta) d\hat G_n(\theta)\right | + \left |\int_{\theta_\alpha}^{+\infty} f(x|\theta) d \hat G_n(\theta) - \int_{\theta_\alpha}^{+\infty} f(x|\theta) dG(\theta) \right |. 
\end{align*}
The first term converges to 0 uniformly due to Lemma \ref{lemma: A3}. To show the second term also converges to zero uniformly, we make use of the result that if $\hat G_n$ weakly converges to $G$, which holds under Assumption \ref{ass2}, then $\underset{g \in \mathcal{BL}}{\sup} |\int g d\hat G_n - \int g dG| \to 0$ where $\mathcal{BL}$ is the class of bounded Lipschitz continuous functions. Note that $f(x|\theta) 1\{\theta \geq \theta_\alpha\}$ is bounded and continuous except at $\theta = \theta_\alpha$. So we construct a smoothed version of $f(x|\theta) 1\{\theta \geq \theta_\alpha\}$, denoted as $g(x|\theta)$, by replacing $1\{\theta \geq \theta_\alpha\}$ by a piecewise linear function taking value zero for $\theta < \theta_\alpha$ and value 1 for $\theta \geq \theta_\alpha + \epsilon$ and taking the form $-\theta_\alpha/\epsilon + \theta/\epsilon$ for $\theta \in [\theta_\alpha, \theta_\alpha + \epsilon]$, then $g \in \mathcal{BL}$. The result \eqref{numer} then holds by showing that 
\[
\underset{x}{\sup}  \left |\int_{\theta_\alpha}^{\theta_\alpha + \epsilon} f(x|\theta) d\hat G_n(\theta) + \int_{\theta_\alpha}^{\theta_\alpha + \epsilon} f(x|\theta) d G(\theta) \right |  \to 0 \quad a.s. 
\]
which holds by Assumptions \ref{ass1} and \ref{ass2}. 
	\end{proof} 

\vspace{3mm}

\begin{proof}[Proof of Theorem \ref{thm:adaptive}] 
	Define analogously 
	\begin{align*}
		\hat H_{n,0}(t) & = \frac{1}{n} \sum_{i=1}^n 1\{\hat v_{\alpha_i}\geq t\}\\
				\hat H_{n,1}(t) & = \frac{1}{n} \sum_{i=1}^n (1- \hat v_{\alpha,i}) 1\{\hat v_{\alpha,i} \geq t\}\\
				\hat Q_n(t) & = \hat H_{n,1}(t)/ \hat H_{n,0}(t)\\
		\end{align*} 
	We first show 
	\begin{align*}
		&\underset{t\in [0,1] } {\sup} | \hat H_{n,0}(t) - H_0(t)|\overset{p}{\to} 0\\
		& \underset{t\in [0,1] } {\sup} | \hat H_{n,1}(t) - H_1(t)|\overset{p}{\to} 0
		\end{align*} 
	We will prove the second statement now, and the first can be proved using a similar argument. To prove the second statement, it suffices to show that 
	\[
	\underset{t \in [0,1]}{\sup} \left | \hat H_{n,1}(t) - H_{n,1}(t)\right| \overset{p}{\to} 0
	\]
To this end, note 
\begin{align*}
&	\underset{t\in [0,1]}{\sup} \left |\frac{1}{n} \sum_i (1- \hat v_{\alpha,i}) 1\{\hat v_{\alpha_i} \geq t\} - \frac{1}{n} \sum_i (1-v_{\alpha,i}) 1\{v_{\alpha,i}\geq t\}\right| \\
	&  = \underset{t\in [0,1]}{\sup}  \left | \frac{1}{n} \sum_i (1- \hat v_{\alpha,i}) 1\{\hat v_{\alpha_i} \geq t\} - \frac{1}{n} \sum_i (1- v_{\alpha,i}) 1\{\hat v_{\alpha,i} \geq t\}\right| \\
	& + \underset{t \in [0,1]}{\sup} \left | \frac{1}{n} \sum_i (1- v_{\alpha,i}) 1\{\hat v_{\alpha,i} \geq t\} - \frac{1}{n} \sum_i (1- v_{\alpha,i}) 1\{v_{\alpha,i} \geq t\}\right| \\
	& \leq \frac{1}{n} \sum_i |\hat v_{\alpha,i}- v_{\alpha,i}| + \underset{t\in [0,1]}{\sup} \frac{1}{n} \sum_i \left | 1\{\hat v_{\alpha,i} \geq t\} - 1\{v_{\alpha,i} \geq t\}\right| 
\end{align*}
The first term is implied by the result in Lemma \ref{lemma: A4}. The second term can be written as 
\begin{align*} 
& \underset{t \in [0,1]}{\sup} 	\frac{1}{n}\sum_i \Big  |1\{\hat v_{\alpha,i} \geq t\} - 1\{v_{\alpha,i} \geq t\}\Big  | \\
	& =\underset{t \in [0,1]}{\sup}  \frac{1}{n} \sum_i \Big  [ 1\{\hat v_{\alpha,i} \geq t, v_{\alpha_i} < t\} + 1\{\hat v_{\alpha,i} < t, v_{\alpha,i} \geq t\}\Big ]\\
	& = \underset{t \in [0,1]}{\sup}  \frac{1}{n} \sum_i \Big[ 1\{\hat v_{\alpha,i} \geq t, t - e < v_{\alpha_i} < t\} + 1\{\hat v_{\alpha,i} < t, t \leq v_{\alpha,i} < t + e\} \Big] \\
	& + \frac{1}{n} \sum_i \Big[ 1\{\hat v_{\alpha, i} \geq t, v_{\alpha,i} < t - e\} + 1\{\hat v_{\alpha,i} < t, v_{\alpha,i} \geq t + e\}\Big]  \\
	& \leq\underset{t \in [0,1]}{\sup}  \frac{1}{n} \sum_i 1\{ t - e \leq v_{\alpha,i} \leq t + e\} + \frac{1}{ne} \sum_i |\hat v_{\alpha,i} - v_{\alpha,i}| \\
	& \leq \underset{t \in [0,1]}{\sup} | H_0(t+e) - H_0(t-e)| + 2 \underset{t \in [0,1]}{\sup} |H_{n,0}(t) - H_0(t)| +  \frac{1}{ne} \sum_i |\hat v_{\alpha,i} - v_{\alpha,i}| 
\end{align*} 
for some $e> 0$ arbitrarily small and bounded away from zero, the right hand side converges to zero in probability by results in Lemma \ref{lemma: A1} and the uniform continuity of $H_0$. Using similar arguments in the proof of Theorem \ref{thm:oracle}, we can establish that $\underset{t \leq 1- \gamma} {\sup} \Big| \hat Q_n(t) - Q(t) \Big| \overset{p}{\to} 0$ and $\underset{t\leq 1- \gamma}{\sup}  \Big| \frac{\hat V_n(t)}{\hat H_{n,0}(t)} - Q(t) \Big| \overset{p}{\to} 0$ with $\hat Q_n(t)  = \hat H_{n,1}(t) / \hat H_{n,0}(t)$ and $\hat V_n(t)= \frac{1}{n} \sum_{i=1}^n 1\{\hat v_{\alpha,i} \geq t\} 1\{ \theta_i \leq \theta_\alpha\}$ and consequently $\underset{n \to \infty} {\limsup} \mathbb{E}\Big[ \frac{\hat V_n(\hat \lambda_n)}{\hat H_{n,0}(\hat \lambda_n) \bigvee 1}\Big] \leq \gamma$. 
	\end{proof} 

\vspace{3mm} 
\begin{proof}[Proof of Theorem \ref{thm:adaptivepower}]

	We first show that $\hat \lambda_{1n}\overset{p}{\to} \lambda_{1}^*$ and $\hat \lambda_{2n} \overset{p}{\to} \lambda_2^*$, then by the continuous mapping theorem, we have $\hat \lambda_n = \max\{\hat \lambda_{1n}, \hat \lambda_{2n}\} \overset{p}{\to} \max\{\lambda_1^*, \lambda_2^*\} = \lambda^*$. The second statement follows from Lemma \ref{lemma: A3}. The first statement holds because by the argument for Theorem \ref{thm:adaptive}, we have 
	\begin{equation} \label{eq: FDRc}
	\underset{t \geq 1- \gamma}{\sup} \Big| \hat Q_{n}(t) - Q(t) \Big| \overset{p}{\to} 0
	\end{equation}
	And therefore for any $\epsilon>0$ not too large, we have $\underset{t \leq \lambda^* - \epsilon}{\inf} Q(t) > \gamma$ and $Q(\lambda^* + \epsilon) > \gamma$ by monotonicity of $Q(t)$. Combined with \eqref{eq: FDRc}, we have $\hat \lambda_{1n} \overset{p}{\to} \lambda_1^*$.

	Now define
	\begin{align*}
				H_{n,2} & = \frac{1}{n} \sum_{i=1}^n v_{\alpha,i} \11 \{ v_{\alpha,i} \geq t\}\\
		\hat  H_{n,2}(t)& = \frac{1}{n}\sum_{i=1}^n \hat v_{\alpha,i} \11\{\hat v_{\alpha,i} \geq t\}\\
		\hat U_{n}(t) & = \frac{1}{n} \sum_{i=1}^n \11\{\theta_i \geq \theta_\alpha, \hat v_{\alpha,i} \geq t\} \\
		U_n(t) & = \frac{1}{n} \sum_{i=1}^n \11\{\theta_i \geq \theta_\alpha, v_{\alpha,i} \geq t\}\\
		H_2(t) & = \mathbb{P}(\theta_i \geq \theta_\alpha, v_{\alpha,i} \geq t) = \alpha \beta(t) 
		\end{align*}
		It suffices to prove that $\frac{1}{n} \sum_{i=1}^n \11\{\theta_i \geq \theta_\alpha, \hat v_{\alpha,i} \geq \hat \lambda_n\} \overset{p}{\to} H_2(\lambda^*)$. 
	Using a similar argument as for Theorem \ref{thm:adaptive}, we can show 
\begin{align*}
& \underset{t \in [0,1]}{\sup}  \Big |\hat H_{n,2}(t) - H_2(t)\Big| \overset{p}{\to} 0\\
& \underset{t \in [0,1]}{\sup}  \Big |\hat U_{n}(t) - H_2(t)\Big| \overset{p}{\to} 0
\end{align*}
Combining this  with the result that $\hat \lambda_n \overset{p}{\to} \lambda^*$, we have $\frac{1}{n} \sum_{i=1}^n \11\{\theta_i \geq \theta_\alpha, \hat v_{\alpha,i} \geq \hat \lambda_n\} \overset{p}{\to} H_2(\lambda^*)$ by continuity of $H_2$. 
\end{proof} 
}}

\section{A Discrete Bivariate Example}
\label{sec:DiscreteEx}
{\footnotesize{

In this appendix we consider a case where $G$ is a discrete distribution joint distribution
in  the pairs, $\theta,\sigma)$, in particular, 
$G(\theta,\sigma) = 0.85 \delta_{(-1, 6)} + 0.1 \delta_{(4, 2)} + 0.05 \delta_{(5, 4)}$. 
In contrast to the discrete example in Section \ref{eg: discreteS3}, 
the unobserved variance $\sigma^2$ is now clearly informative about
$\theta$ for this distribution $G$. 
	
We will focus on the capacity constraint $\alpha = 0.05$ so $\theta_\alpha = 5$. 
For $T = 9$, the level curves for tail probability and posterior mean are shown in 
Figure \ref{fig.ndpanel} and the selection set comparison for one sample realization in Figure \ref{fig.ndpanelc}. The right panel of the figure plots the selection
boundaries for the two ranking criteria for $\gamma = 10\%$.  The
non-monotonicity of $v_\alpha(y,s)$ in both $y$ and $s$ is apparent. The posterior mean criteria 
based on $\mathbb{E}(\theta|Y,S)$, prefers individuals with smaller variances 
compared to the rule based on the tail probability. 
Since sample variances $S$ are informative about $\theta$, when the sample variance 
is small and selection is based on posterior tail probability, the oracle is aware that 
such a small sample variance is only likely when $\theta=4$, hence will only make a selection when 
we observe a very large $y$. As a result, the Oracle sets a higher selection 
threshold on $y$ to avoid selecting individuals with true effect $\theta = 4$. On the other hand, 
the posterior mean criterion also tries to use information from the sample variance, 
but not as effectively for our selection objective.  This can be seen in the level curves in the 
middle panel of Figure \ref{fig.ndpanel}. When the sample variance is small, 
the posterior mean shrinks very aggressively towards 4, thereby sacrificing valuable information from $y$. 
For a wide range of values for the sample mean $y$, the posterior mean delivers a value 
close to 4, thus failing to distinguish between those with $\theta = 5$ and those with $\theta = 4$. 
Consequently, the posterior mean rule sets a lower thresholding value on $y$ for the selection region when 
sample variance is small resulting in inferior power performance,
as shown in Table \ref{tab.ndpanel}.  

\begin{figure}[h!]
	\begin{center}
		\resizebox{.95\textwidth}{!}{\includegraphics{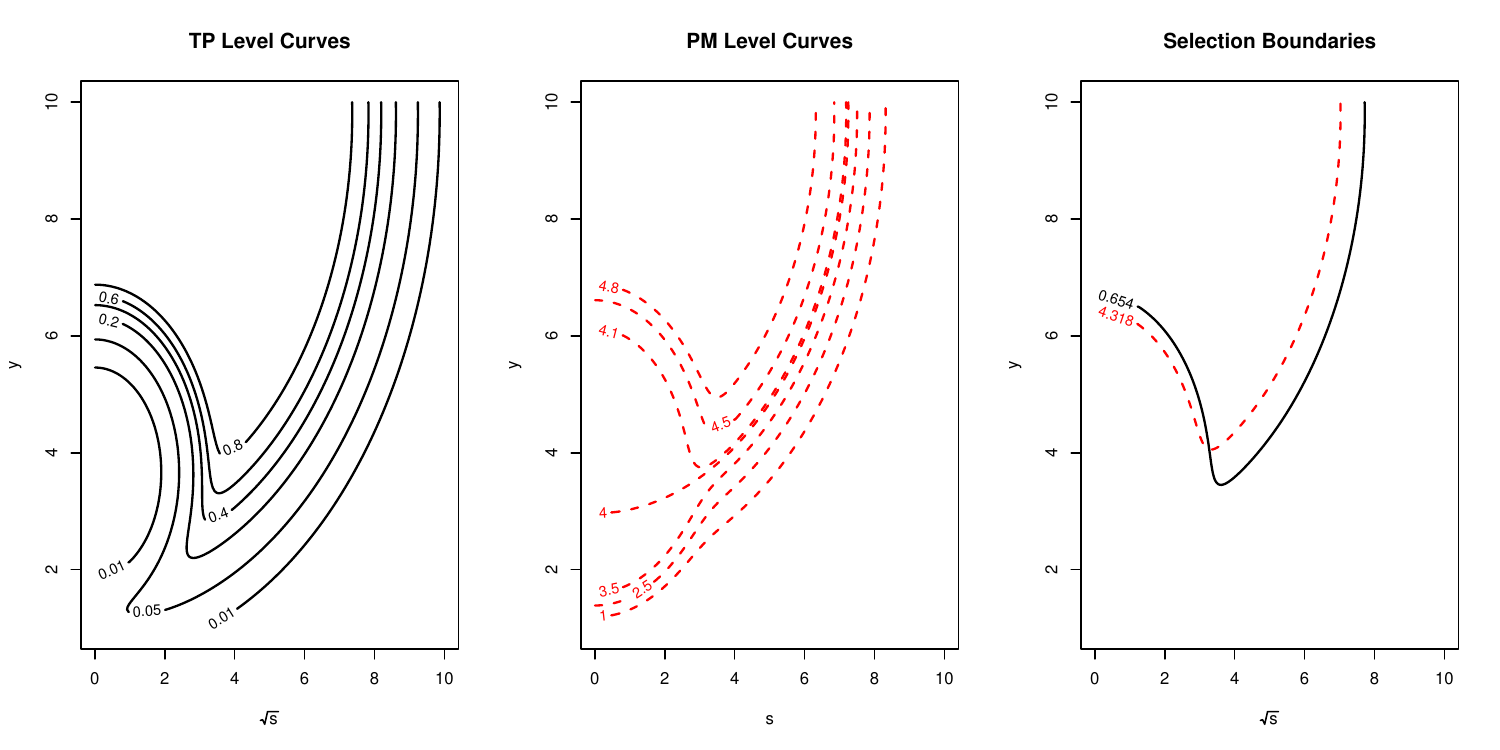}}
	\end{center}
	\caption{\small The left panel plots the level curves for the posterior tail probability 
	criterion and the middle panel depicts the level curves for posterior mean criterion. 
	The right panel plots the selection boundary based on posterior mean ranking (shown 
	as the red dashed lines) and the posterior tail probability ranking (shown as 
	the black solid lines) for $\alpha = 5\%$ and $\gamma = 10\%$ with $G(\theta,\sigma^2)$ 
	follows a three points discrete distribution.}
	\label{fig.ndpanel}
\end{figure}

\begin{figure}[h!]
	\begin{center}
		\resizebox{.95\textwidth}{!}{\includegraphics{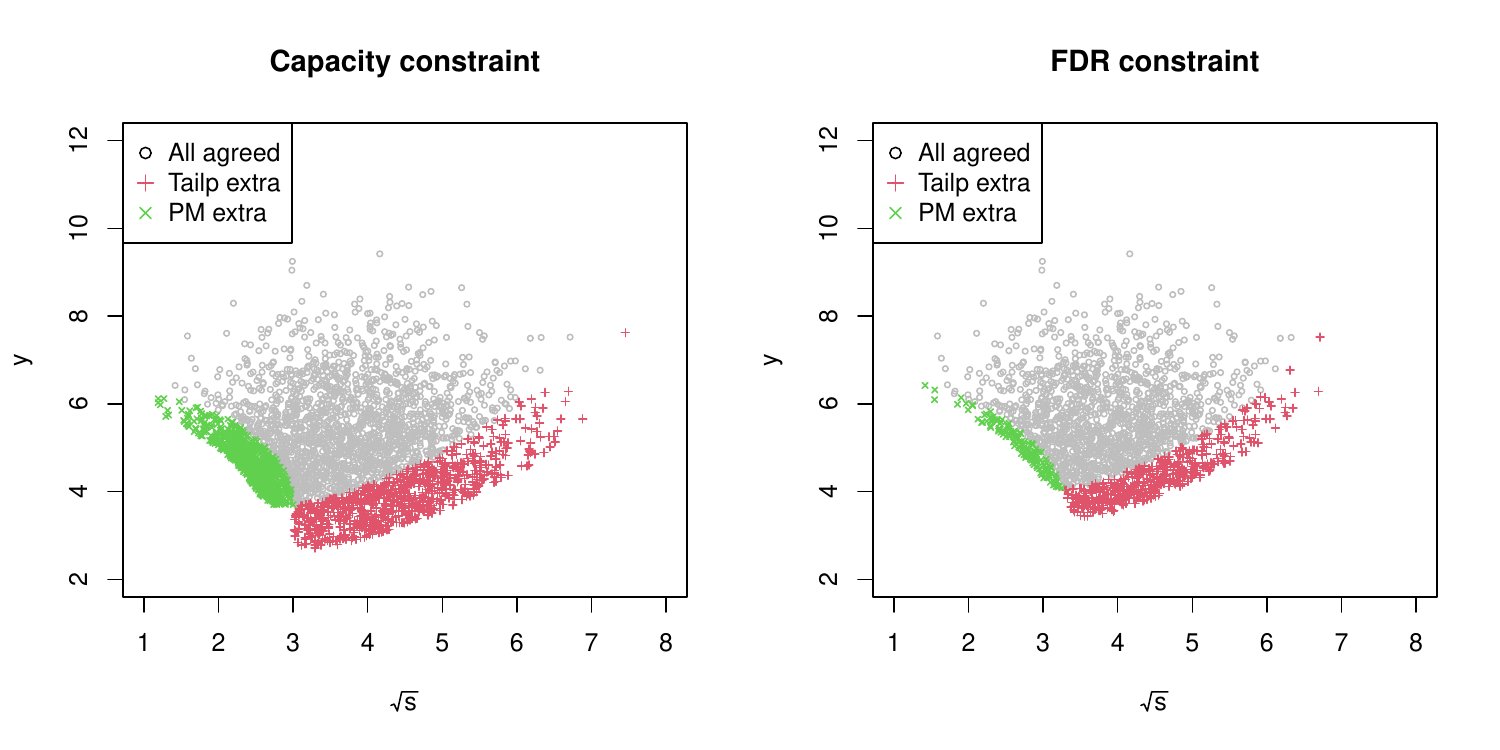}}
	\end{center}
	\caption{\small Selection set comparison for one sample realization from the three points discrete distribution model: The left panel shows in black circles the agreed selected elements by both the posterior mean and posterior tail probability criteria under the capacity constraint, extra elements selected by the posterior mean are marked in green and extra elements selected by the posterior tail probability rule are marked in red. The right panel shows the comparison of the selected sets under both the capacity and FDR constraint with $\alpha = 5\%$ and $\gamma = 10\%$.}
	\label{fig.ndpanelc}
\end{figure}

Table \ref{tab.ndpanel} reports several performance measures over 200 simulation repetitions with 
$n = 50,000$. There, we consider four additional methods for ranking: 
\begin{itemize}
    \item  MLE:  ranking of the maximum likelihood estimators, $Y_i$  for each of the $\theta_i$, 
    \item  P-values: ranking of the P-values of the conventional one-sided test of the null hypothesis 
	$H_0: \theta < \theta_\alpha$ , 
    \item  PM-NIX:  ranking of the posterior means based on the normal-inverse-chi-square (NIX) prior distribution,
    \item  TP-NIX:  ranking of the posterior tail probability based on NIX prior distribution
\end{itemize}

{\tiny{
\begin{table}[!tbp]
\begin{center}
\begin{tabular}{lrrrcrrrcrrr}
\hline\hline
\multicolumn{1}{l}{\bfseries }&\multicolumn{3}{c}{\bfseries $\gamma=1\%$}&\multicolumn{1}{c}{\bfseries }&\multicolumn{3}{c}{\bfseries $\gamma=5\%$}&\multicolumn{1}{c}{\bfseries }&\multicolumn{3}{c}{\bfseries $\gamma=10\%$}\tabularnewline
\cline{2-4} \cline{6-8} \cline{10-12}
\multicolumn{1}{l}{}&\multicolumn{1}{c}{Power}&\multicolumn{1}{c}{FDR}&\multicolumn{1}{c}{SelProp}&\multicolumn{1}{c}{}&\multicolumn{1}{c}{Power}&\multicolumn{1}{c}{FDR}&\multicolumn{1}{c}{SelProp}&\multicolumn{1}{c}{}&\multicolumn{1}{c}{Power}&\multicolumn{1}{c}{FDR}&\multicolumn{1}{c}{SelProp}\tabularnewline
\hline
PM&$0.217$&$0.010$&$0.011$&&$0.482$&$0.050$&$0.025$&&$0.580$&$0.100$&$0.032$\tabularnewline
TP&$0.252$&$0.010$&$0.013$&&$0.561$&$0.050$&$0.030$&&$0.697$&$0.100$&$0.039$\tabularnewline
P-value&$0.651$&$0.349$&$0.050$&&$0.651$&$0.350$&$0.050$&&$0.651$&$0.349$&$0.050$\tabularnewline
MLE&$0.611$&$0.390$&$0.050$&&$0.611$&$0.390$&$0.050$&&$0.610$&$0.390$&$0.050$\tabularnewline
PM-NIX&$0.611$&$0.390$&$0.050$&&$0.611$&$0.390$&$0.050$&&$0.610$&$0.390$&$0.050$\tabularnewline
TP-NIX&$0.619$&$0.382$&$0.050$&&$0.619$&$0.382$&$0.050$&&$0.618$&$0.382$&$0.050$\tabularnewline
\hline
\end{tabular}
\caption{Performance comparison for ranking procedures based on posterior mean,  
posterior tail probability, the P-value and the MLE of $\theta_i$.  
All results are based on 200 simulation repetitions with n = 50,000 for 
G following the three point discrete distribution and T = 9 or when G is assumed 
to follow a normal-inverse-chi-square distribution. For the first two rows, number reported in the table correspond to performance when both capacity and FDR constraints are in place. For the last four rows, only capacity constraint is in place.\label{tab.ndpanel}}\end{center}
\end{table}

}}

The first two of these selection rules ignore the compound decision perspective of the problem entirely. 
The other two ranking criteria we consider are based on posterior mean and tail probability assuming 
$G$ follows a normal-inverse-chi-square (NIX) distribution (denoted as PM-NIX and TP-NIX in 
Table \ref{tab.ndpanel}). The parameters of the NIX distribution are estimated from the data, 
hence these rules can be viewed as generalization of James-Stein estimator for homogeneous variances 
and the Efron-Morris shrinkage estimator for known heterogeneous variances case. We refer the details 
of the NIX distribution and the posterior distribution of $(\theta, \sigma^2)$ to 
Example \ref{eg: normaNIX}. 


We report the power and false discovery rates for the posterior mean (denoted as PM) and posterior tail probability (denoted as TP) selection as well as 
the proportion of selected observations for $\alpha = 5\%$ and for several different $\gamma$ under both capacity and FDR control. For all other four selection rules we only impose the capacity constraint, as how they are usually implemented in current practise. 
For the PM and TP rules, from the proportion selected observations we can infer whether the FDR constraint or the capacity 
constraint is binding in each configuration.  Ranking based on the posterior tail probability 
clearly has better power performance for each of configurations 
 when compared to the posterior mean ranking. When selecting as few as 5\%, FDR constraints are binding for both PM and TP rule 
for all ranges of $\gamma \in \{1\%,5\%,10\%\}$. Among all the other rules, we see that the false discovery rate is around 40\% and PM-NIX has identical performance 
as ranking based on the MLE for $\theta$; this can be understood by noting that the posterior mean of 
$\theta$ under the NIX prior is simply linear shrinkage of the MLE of $\theta$, hence it does not alter 
individual rankings between the two methods. TP-NIX behaves similarly to PM-NIX, with slightly better power and slightly lower false discovery rate.



\section{A Counterexample:  Non-nested Selection Regions} \label{sec: S3eg}
{\footnotesize{

Thus far we have stressed conditions under which selection regions are
nested with respect to $\alpha$, that is, for $\alpha_1 < \alpha_2 < \cdots < \alpha_m$, we 
have $ \Omega_{\alpha_1} \subseteq \Omega_{\alpha_2}  \subseteq \cdots \subseteq \Omega_{\alpha_m}$, 
for the selection regions.  However, this need not hold when there is variance heterogeneity
and when nesting fails we can have seemingly anomalous situations in which units are
selected by the tail probability rule at some stringent, low $\alpha$, but are then 
rejected for some less stringent, larger $\alpha$'s.
To illustrate this phenomenon we will neglect the FDR constraint and focus on our
discrete mixing distribution, $G = 0.85 \delta_{-1} + 0.10 \delta_{2} + 0.05 \delta_{5}$,
with $\sigma \sim U[1/2,4]$.  The selection regions are depicted
in Figure \ref{fig.Ceg12} for $\alpha \in \{0.04, 0.05, 0.06, 0.08\}$.  Units are selected
when their observed pair, $(y_i , \sigma_i)$ lies above these curves for various $\alpha$'s.
When $\sigma$ is small we see, as expected, that selection is nested: if a unit is selected
at low $\alpha$ it stays selected at larger $\alpha$'s.  However, when $\sigma = 3$, we
see that there are units selected at $\alpha = 0.05$ and even $\alpha = 0.04$ and yet
they are rejected for $\alpha = 0.06$.  How can this be?

\begin{figure}[h!]
    \begin{center}
    \resizebox{.8\textwidth}{!}{\includegraphics{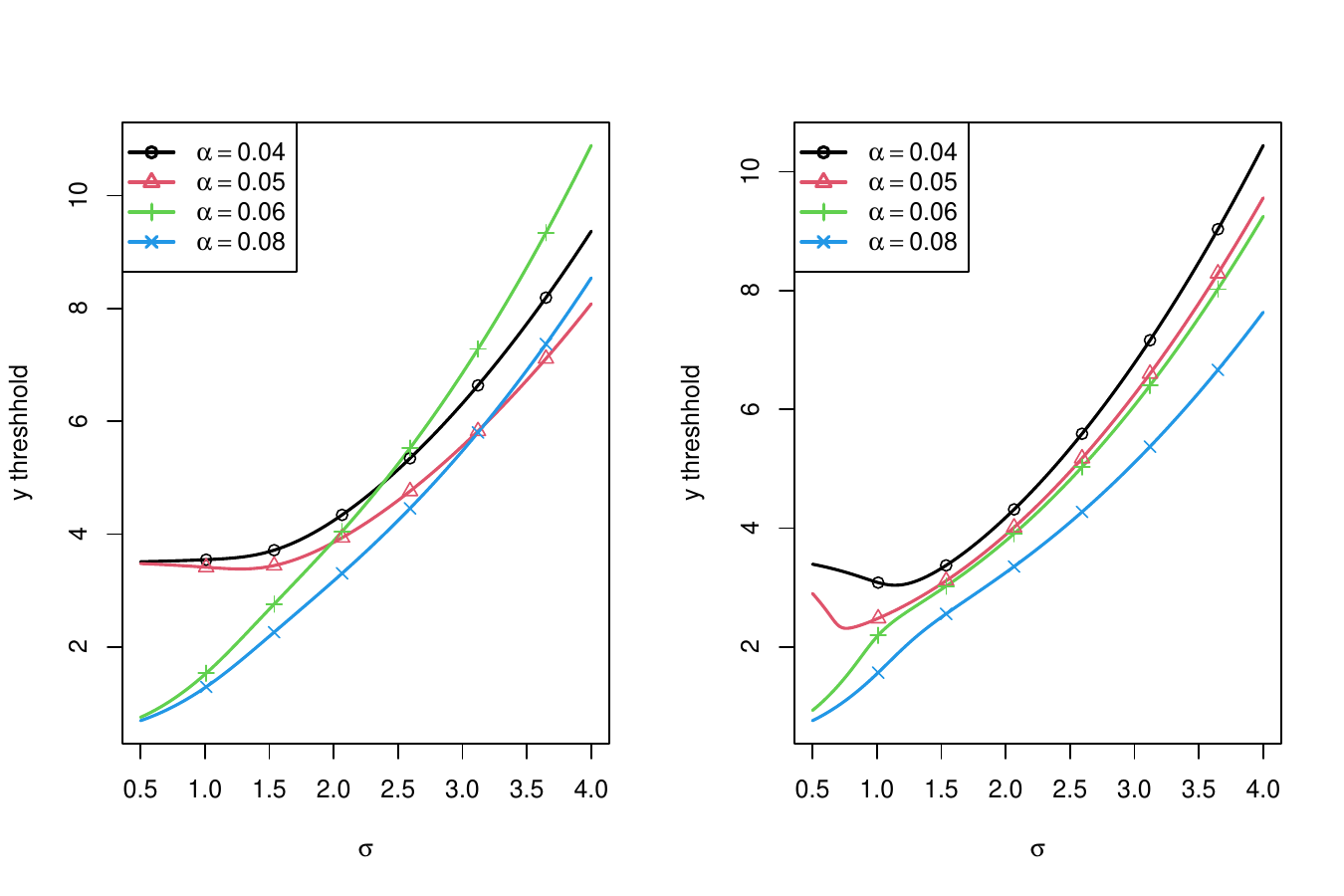}}
    \end{center}
    \caption{\small Oracle selection boundaries (with just capacity constraint) for several $\alpha$ levels 
    for the tail probability criterion (left panel) and posterior mean criterion (right panel) 
    with a discrete example with $G = 0.85 \delta_{-1} + 0.10 \delta_{2} + 0.05 \delta_{5}$,
    and $\sigma \sim U[1/2,4]$.  Crossing of the boundaries implies that the selection
    regions are non-nested as explained in the text.}
    \label{fig.Ceg12}
\end{figure}

Imagine you are the oracle, so you know $G$, and when you decide to select with $\alpha =0.06$
you know that you will have to select a few $\theta = 2$ types, since there are only
5 percent of the $\theta = 5$ types.  Your main worry at that point is to try to avoid
selecting any $\theta = -1$ types;  this can be accomplished but only by avoiding the
high $\sigma$ types.  In contrast when $\alpha = 0.05$ so we are trying to vacuum up
all of the $\theta = 5$ types it is worth taking more of a risk with high $\sigma$ types
as long as their $y_i$ is reasonably large.

The crossing of the selection boundaries and non-nestedness of the selection regions is
closely tied up with the tail probability criterion and the $\alpha$ dependent feature of the
hypothesis.  If we repeat our
exercise with the same $G$, and $\sigma$ distribution, but select according to posterior
means, we get the nested selection boundaries illustrated in Figure \ref{fig.Ceg12}.  
Proposition \ref{prop: PMnest} establishes this to be a general phenomenon for any distributution $G$.  

It should be noted that the crossing of selection boundaries we have illustrated
seems to have been anticipated in \citeasnoun{HN}, who consider similar tail criteria.
They propose a ranking scheme that assigns rank equal to the smallest $\alpha$ for
which a unit would be selected as a way to resolve the ambiguities generated by crossing.
We don't see a compelling decision theoretic rationale for this revised ranking rule,
instead we prefer to maintain some separation between the ranking and selection
problems and focus on risk assessment as a way to reconcile them.

\begin{figure}[h!]
	\begin{center}
		\resizebox{.5\textwidth}{!}{\includegraphics{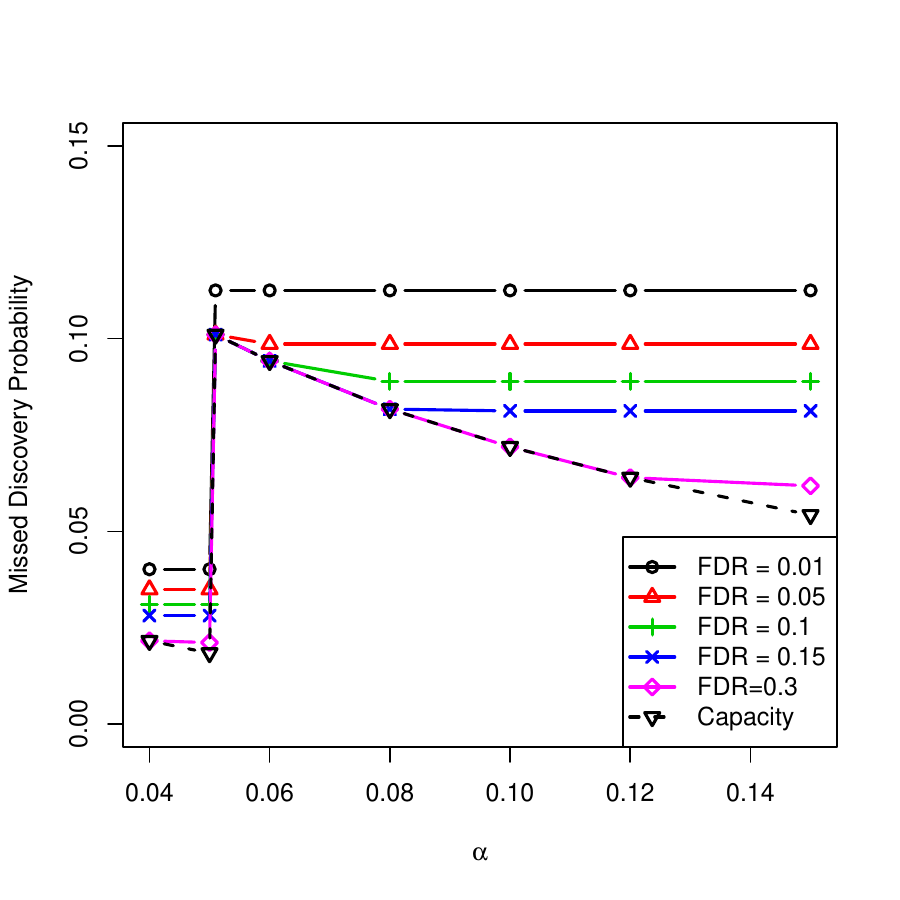}}
	\end{center}
	\caption{\small Oracle risk evaluation for several $\alpha$ levels 
		for the tail probability criterion and a discrete example with 
		$G = 0.85 \delta_{-1} + 0.10 \delta_{2} + 0.05 \delta_{5}$,
		and $\sigma \sim U[1/2,4]$. The solid lines correspond to the 
		evaluation of the loss function specified in (\ref{loss}) with 
		both capacity and FDR control constraints. The dotted line 
		corresponds to oracle risk when only the capacity constraint is imposed. }
	\label{fig.Ceg1_obj}
\end{figure}

The risk based on the loss function defined in (\ref{loss}) clearly depends on $\alpha$ and $\gamma$. More specifically it consists of three pieces, the leading term has the interpretation of ``missed discovery'' probability, which we try to minimize and the second and third pieces correspond to the FDR and capacity constraints respectively each weighted by a Lagrangian multiplier. Focusing on the first term, we have 
\begin{align*}
\mathbb{E}[ H_i (1- \delta_i)] & = \mathbb{E}[\11\{\theta_i \geq \theta_\alpha \} (1-\delta_i)]\\
& = \mathbb{P}[\theta_i \geq \theta_\alpha] - \mathbb{E}[\delta_i \11\{\theta_i \geq \theta_\alpha\}] \\
& = \mathbb{P}[\theta_i \geq \theta_\alpha]  - \int \int_{\theta_\alpha}^{+\infty} \Big[1- \Phi((t_\alpha(\lambda,\sigma) - \theta)/\sigma) \Big]dG(\theta) dH(\sigma)
\end{align*}
where $\lambda$, depends on $\alpha$ and $\gamma$, is determined by either the false discovery rate control or the capacity constraint, whichever binds. 

A feature of discrete mixing distributions, $G$, is that the first term in the loss, 
$\mathbb{P}(\theta_i \geq \theta_\alpha)$, is piece-wise constant, with jumps occurring 
only at discontinuity points of $G$, while the second term depends on both $\theta_\alpha$ 
and the cut-off values $t_\alpha(\lambda,\sigma)$. When the capacity constraint binds 
there exist  ranges of $\alpha$ such that $\theta_\alpha$ remains constant, 
while $t_\alpha(\lambda,\sigma)$ decreases for each $\sigma$, hence the risk with just capacity 
constraint binding is a decreasing function for $\alpha$ in the interval $(0.05, 0.15)$. 
On the other hand, when the FDR constraint binds, it can be shown that the cut-off value 
$t_\alpha(\lambda,\sigma)$ is constant.  
To see this recall that the cutoff $\lambda$ determined by the FDR constraint is defined as 
$\mathbb{E}[(1-v_\alpha(Y,\sigma)) \11\{v_\alpha(Y,\sigma) \geq \lambda\}] = 
\gamma \mathbb{P}(v_\alpha(Y,\sigma) \geq  \lambda)$, so when $\theta_\alpha$ is constant
over a range of $\alpha$ the distribution of $v_{\alpha}(Y,\sigma)$ does not change, 
and consequently the value $\lambda$ is constant over that range of $\alpha$.  

Figure \ref{fig.Ceg1_obj} evaluates risk based on the optimal selection rule for
various $\alpha$ and FDR levels, $\gamma \in \{0.01, 0.05, 0.1, 0.15, 0.3\}$. 
The solid curves correspond to risk evaluated at the optimal Bayes rule defined in 
Proposition \ref{prop: rule_ysigma}. The dotted line corresponds to the risk evaluated 
at the Bayes rule when only the capacity constraint is imposed. As $\gamma$ increases, 
the risk decreases as expected. For FDR levels as stringent as $\gamma = 0.01$, 
the FDR constraint binds and the risk is piece-wise constant.  As the FDR level is relaxed, 
there are range of $\alpha$ such that capacity constraint becomes binding, and 
the risk decreases after the initial jump at $\alpha = 0.05$. 

The evaluation of the risk for this particular example indicates that it is easier to select the top 5\% individuals, 
those with $\theta = 5$. As we intend to select more in the right tail, we are facing more uncertainty. 
This also motivates a more systematic choice of $(\alpha,\gamma)$.  Although selection based on the tail
probability criterion can lead to non-nested selection regions, we conclude this sub-section by demonstrating
that posterior mean selection is necessarily nested.

\begin{proposition}\label{prop: PMnest}
Let the density function of $y$ conditional on $\theta$ and $\sigma$ be denoted as $f(y|\theta,\sigma)$. 
If selection is based on the posterior mean, $\delta_i = \{M(y, \sigma) \geq c(\alpha, \gamma)\}$ with 
\[
M(y, \sigma) := \frac{\int \theta f(y|\theta,\sigma) dG(\theta)}{\int f(y|\theta,\sigma)dG(\theta)}
\]
and $c(\alpha, \gamma)$ is chosen to satisfy both the capacity constraint at level $\alpha$ and the FDR constraint at 
level $\gamma$, then the selection regions, defined as 
$\Lambda_{\alpha,\gamma}= \{(y, \sigma): M(y, \sigma) \geq c(\alpha, \gamma)\}$, 
are nested, that is, for any $\alpha_1 > \alpha_2$, $\Lambda_{\alpha_2 , \gamma} \subseteq \Lambda_{\alpha_1, \gamma}$. 
\end{proposition}

\bibliography{Ranking}

\end{document}